\newif\ifpods
\podstrue
\podsfalse
\newif\ifarxiv
\ifpods
\arxivfalse
\else
\arxivtrue
\fi

\ifpods
\documentclass[acmsmall,review,anonymous]{acmart}
\else
\documentclass[acmsmall,nonacm]{acmart}
\fi

\author{Christian Janos Lebeda}
\affiliation{%
  \institution{Inria, Université de Montpellier, INSERM}
  \country{France}}
\email{christian-janos.lebeda@inria.fr}
\author{Aleksandar Nikolov}
\affiliation{%
  \department{Department of Computer Science}
  \institution{University of Toronto}
  \country{Canada}}
\email{sasho.nikolov@utoronto.ca}
\author{Haohua Tang}
\affiliation{%
  \department{Department of Computer Science}
  \institution{University of Toronto}
  \country{Canada}}
\email{haohua.tang@mail.utoronto.ca}

\usepackage[T1]{fontenc}
\usepackage{lmodern}

\usepackage{bbm}

\usepackage{multirow}
\usepackage{dirtytalk}
\usepackage{amsmath}

\usepackage{amssymb}
\usepackage{amsthm}
\usepackage{mathtools}

\usepackage[capitalize,noabbrev]{cleveref} 
\usepackage{url}
\usepackage{microtype}
\usepackage{relsize}
\usepackage{natbib}

\usepackage{algorithm}
\usepackage{algpseudocode}
\usepackage{arydshln}

\usepackage{listings}
\usepackage{xspace}
\usepackage{multirow}
\usepackage{thmtools}
\usepackage{thm-restate}
\usepackage{breqn}
\usepackage{subfig}

\usepackage{booktabs}       
\usepackage{nicefrac}       
\usepackage{microtype}      
\usepackage{colortbl}       
\usepackage{stackengine}
\usepackage{lipsum}
\usepackage{dsfont}
\usepackage{graphicx}
\usepackage{tabularx}
\usepackage{float}
\usepackage{mathtools}
\usepackage{mathrsfs}
\usepackage{enumitem}
\usepackage{xspace}
\usepackage{xparse}
\usepackage[export]{adjustbox}

\usepackage{todonotes}

\usepackage[english]{babel}
\usepackage{amsthm}

\ifpods
\newtheorem{theorem}{Theorem}
\newtheorem{corollary}[theorem]{Corollary}
\newtheorem{definition}[theorem]{Definition}
\newtheorem{lemma}[theorem]{Lemma}
\newtheorem{observation}[theorem]{Observation}
\newtheorem{fact}[theorem]{Fact}

\theoremstyle{remark}
\newtheorem{remark}{Remark}
\else
\newtheorem{theorem}{Theorem}[section]
\newtheorem{lemma}[theorem]{Lemma}
\newtheorem{corollary}[theorem]{Corollary}

\theoremstyle{definition}
\newtheorem{definition}[theorem]{Definition}

\theoremstyle{remark}
\newtheorem{remark}[theorem]{Remark}
\fi

\DeclareMathOperator{\tr}{tr}

\newcommand{\uni}{\mathcal{U}}
\newcommand{\R}{\mathbb{R}}
\newcommand{\C}{\mathbb{C}}
\DeclareMathOperator{\supp}{supp}

\newcommand{\christian}[1]{\textcolor{blue}{Christian: #1}}
\newcommand{\sasho}[1]{\textcolor{green}{Sasho: #1}}
\newcommand{\haohua}[1]{\textcolor{teal}{Haohua: #1}}

 \renewcommand{\christian}[1]{}
 \renewcommand{\sasho}[1]{}
 \renewcommand{\haohua}[1]{}

\usepackage{mdframed}

\newmdenv[
  topline=false,
  bottomline=false,
  rightline=false,
  skipabove=\topsep,
  skipbelow=\topsep
]{lined}

\newcommand{\gausstime}{\lambda}

\newcommand{\fq}{F} 
\newcommand{\abs}[1]{\vert #1\vert}
\newcommand{\ip}[1]{\langle #1\rangle}
\renewcommand{\SS}{\mathcal{S}}
\DeclareMathOperator{\E}{\mathbb{E}}
\DeclareMathOperator{\err}{err}

\title[Optimal Gaussian Noise for Differentially Private Marginal and Product Queries]{Weighted Fourier Factorizations: Optimal Gaussian Noise for Differentially Private Marginal and Product Queries}

\begin{document}

    
\begin{abstract}
    We revisit the task of releasing marginal queries under differential privacy with additive (correlated) Gaussian noise. 
    We first give a construction for answering arbitrary workloads of weighted marginal queries, over arbitrary domains.
    Our technique is based on releasing queries in the Fourier basis with independent noise with carefully calibrated variances, and reconstructing the marginal query answers using the inverse Fourier transform.
    We show that our algorithm, which is a factorization mechanism, is exactly optimal among all factorization mechanisms, both for minimizing the sum of weighted noise variances, and for minimizing the maximum noise variance. Unlike algorithms based on optimizing over all factorization mechanisms via semidefinite programming, our mechanism runs in time polynomial in the dataset and the output size. This construction recovers results of Xiao et al.~[Neurips 2023] with a simpler algorithm and optimality proof, and a better running time. 
    
    We then extend our approach to a generalization of marginals which we refer to as product queries. 
    We show that our algorithm is still exactly optimal for this more general class of queries.
    Finally, we show how to embed extended marginal queries, which allow using a threshold predicate on numerical attributes, into product queries. We show that our mechanism is \emph{almost} optimal among all factorization mechanisms for extended marginals, in the sense that it achieves the optimal (maximum or average) noise variance up to lower order terms.
\end{abstract}

\maketitle

\christian{I added this note so that we remember to disable notes before uploading}

\section{Introduction}
\label{sec:introduction}

In this work we study marginal queries and generalizations under differential privacy. We consider datasets $D$ in which each data point is specified by $d$ categorical attributes (we discuss our generalization to numerical queries later). A marginal query is given by a subset $S$ of the attributes, and asks, for each possible setting $t$ of the attributes in $S$, for the number of data points in $D$ that have attributes in $S$ with values agreeing with $t$. For example, consider a dataset that tracks sex, education level, place of residence, marital status, presence or absence of some genetic markers, and whether a person has been diagnosed with a certain disease. Then the answer to a marginal query corresponding to the sex attribute, one of the genetic marker attributes, and the disease diagnosis attribute, is a 3-dimensional table with $2\times 2\times 2$ cells: one for each setting of these three attributes. The cells of the table give the number of males that have the genetic marker and have been diagnosed with the disease, the number of females that have the genetic marker and have been diagnosed, the number of males that don't have the marker and have been diagnosed, etc. 

Marginal queries like these are known by different names, e.g., OLAP data cubes, and contingency tables, and are ubiquitous when summarizing high-dimensional data, including data from surveys, clinical studies, and official statistics. Often, however, the underlying data is sensitive, and protecting its privacy is sometimes even mandated by law. In these situations, releasing marginal queries can raise significant privacy concerns. Answers to a rich enough set of marginal queries can reveal enough about a dataset to enable an adversary to reconstruct most of the data~\cite{KRSU10}, or to infer the membership of a given data point in the dataset~\cite{BunUV14}. For this reason, we adopt the differential privacy framework~\cite{dwork06calibrating}, and study marginal query release subject to the constraints of this framework.

A concrete motivating example to keep in mind are the marginal queries released by the US Census Bureau for the 2020 Census of Housing and Population. After it was discovered that prior disclosure avoidance systems failed to adequately protect the privacy of the census data~\cite{GarfinkelAM19,jason2020privacy,UScensus2021reconstruction}, the US Census Bureau implemented a new differentially private algorithm, the TopDown algorithm, for the 2020 Census~\cite{censusTopDown}. The TopDown algorithm releases estimates of selected marginal queries partitioned based on a geographical hierarchical structure. 
Since the accuracy of certain marginal queries has high practical importance additional privacy budget is allocated for those estimates. We similarly allow the user to specify an importance weight for each marginal query. 


As another motivating application, marginal queries have also been used as a subroutine when generating synthetic data~\cite{McKenna_Miklau_Sheldon_2021,McKennaSM19,ZhangEtal21-PrivSyn}. 
Those techniques typically select a set of marginals that are then privately estimated using Gaussian noise. 
The synthetic data is generated so that it roughly matches the measured marginals. 

In some applications it is also natural to consider more complex extensions of marginal queries. A particularly natural example are the extended marginal queries (called prefix-marginals in~\cite{McKenna_Miklau_Hay_Machanavajjhala_2023}), in which the attributes are partitioned into categorical and numerical, and, for a set $S$ of attributes, the extended marginal query asks, for each possible setting of the categorical attributes in $S$, and each possible choice of prefix intervals for the numerical attributes in $S$, how many data points agree with the settings of the categorical attributes and have numerical attribute values lying in the chosen prefix intervals. Extended marginals allow us, for example, to ask how many data points correspond to males under the age of 35, or to females under the age of 45, etc. Workloads of extended marginals also appear in products released by the US Census Bureau, see examples described by McKenna et al.~\cite{McKenna_Miklau_Hay_Machanavajjhala_2023}.

Marginal query release under differential privacy has been studied extensively since differential privacy was first introduced. An exhaustive account of this line of work is beyond the scope of this paper, but we highlight the results most relevant to our contributions. The early work of Barak, Chaudhuri, Dwork, Kale, McSherry, and Talwar~\cite{BCDKMT07}, which initiated the formal study of differentially private marginal query release, considered binary attributes and proposed a method to release answers to marginal queries by adding Laplace noise to a workload of Fourier queries, i.e., queries that compute the Fourier transform of the empirical distribution of the data. They further showed how to make the query answers consistent with some real dataset. Later work showed lower bounds on the minimum error necessary to answer marginal queries under differential privacy~\cite{KRSU10,BunUV14}, and investigated the trade-offs between privacy, accuracy, and computational complexity in answering marginal queries~\cite{UllmanV10,ThalerUV12,DworkNT15}. 
McKenna, Miklau, Hay, and Machanavajjhala proposed an efficient method to approximately optimize over a class of private algorithms for answering marginal queries~\cite{McKenna_Miklau_Hay_Machanavajjhala_2023}. In particular, they consider algorithms that first compute private estimates of other marginal queries, and then reconstruct answers to the marginal queries that were originally asked. 
Xiao, He, Zhang, and Kifer~\cite{xiao2023optimal} gave an explicit algorithm for answering marginal queries which is efficient and optimal over the same class of private algorithms considered by McKenna et al. Concurrent to our work, the same authors, joined by Toksoz and Ding, extended their technique to support more general queries, including extended marginal queries, but did not show any optimality results for this extension~\cite{xiao2025optimal}. 
We discuss their work and its relation to ours in more detail below.


\subsection{Problem Setup}
We consider a dataset $D := (x^{(1)}, \ldots, x^{(n)})$, consisting of a sequence of $n$ points from a data universe $\uni$. Each data point has $d$ attributes, where the $i$-th attribute is drawn from the set $\uni_i := \{0, \ldots, m_i-1\}$.\footnote{We can accommodate any finite set by mapping it bijectively to $\{0, \ldots, m_i - 1\}$ for some $m_i$.} Marginal queries are then specified by sets $S \subseteq [d]$, and partial assignments $t \in \uni_S := \prod_{i \in S}\uni_i$. I.e., $t$ has a value $t_i \in \uni_i$ for each $i \in S$, and specifies an assignment to each attribute in $S$. Then the marginal queries are given by $q_{S,t}(D)$, equal to the number of data points $x^{(i)}$ in $D$ agreeing with $t$ on $S$, i.e., $q_{S,t}(D) := |\{i: x_j^{(i)} = t_j\ \forall j \in S\}|$. A marginal query workload $Q_{\SS}$ is then specified by a collection of sets of attributes $\SS$. We always assume that, for each $S \in \SS$, all possible marginal queries $q_{S,t}(D)$ for all $t \in \uni_S$ are asked. I.e., we assume that, for each $S \in \SS$, we need to privately estimate the full table of marginals corresponding to $S$.

To generalize this set-up to extended marginals, we partition the set of attributes $[d]$ into the categorical attributes $C$, and the numerical attributes $N$. We can now redefine $q_{S,t}(D)$ to equal 
\[
q_{S,t}(D)
|\{i: x_j^{(i)} = t_j\ \forall j \in S\cap C, x_j^{(i)} \le t_j \ \forall j \in S \cap N\}|.
\]
Once again, an extended marginal query workload is specified by a collection of sets of attributes $\SS$, and we assume that, for each $S \in \SS$, all possible queries $q_{S,t}(D)$ for all $t \in \uni_S$ are included in the workload.

Informally, differential privacy requires that the randomized algorithm $\mathcal{A}$ that takes as input a private dataset $D$, and outputs (approximate) answers to the queries in $Q_\SS$, has the property that the probability distribution on outputs of $\mathcal{A}(D)$ is similar to the probability distribution on outputs of $\mathcal{A}(D')$ for any $D'$ that is neighboring to $D$. 
In our work, we adopt the add/remove notion of neighboring, i.e., $D$ and $D'$ are neighboring (denoted $D \sim D'$) if and only if we can get $D'$ from $D$ by adding or removing at most one data point from $D$. We refer to Section~\ref{sec:prelim-dp} 
for the formal definitions. 

\ifpods \else
Let us take the Gaussian mechanism as a baseline~\cite{DinurNissim03,DworkN04,DworkKMMN06OurDataOurselves}. 
For a function $f:\uni^* \to \R^K$,
the Gaussian mechanism releases $f(D) + Z$, where $Z$ is a $K$-dimensional normally distributed random vector with independent coordinates, 
and variance $(\Delta f)^2$ at each coordinate.\footnote{In the rest of the introduction we ignore the privacy parameters. Technically, we give the results for $1$-GDP.}
Here $\Delta f$ is the sensitivity of $f$, and equals $\max_{D\sim D'}\|f(D) - f(D')\|_2$. If we take $f$ to be the function that maps $D$ to the true answers to all marginal queries in $\SS$, then it is easy to see that $\Delta f = \sqrt{|\SS|}$, since adding or removing a data point can only affect $q_{S,t}(D)$ for a single partial assignment $t$, and for that $t$ we have $|q_{S,t}(D) - q_{S,t}(D')| = 1$. 
We can improve over this baseline by correlating the Gaussian noise, as we discuss below.\fi

\subsection{Our Contributions, and Factorization Mechanisms}
\label{sec:contributions}

\paragraph{Weighted marginal workloads.}
Our first contribution is a mechanism that privately answers any workload of marginal queries over arbitrary finite domains. 
Motivated by applications, such as the US Census, in which different marginal queries have different importance, we allow the queries to be weighted, and optimize a weighted average of the noise variances. In particular, we assign a non-negative weight $p(S)$ to each attribute set $S$ and measure error as the square root of the weighted sum of variances, where the variance of $q_{S,t}$ is scaled by $\frac{p(S)}{|\uni_S|}$. In this normalization, the weight $p(S)$ of $S$ is split among the $|\uni_S|$ partial assignments $t$ defining different marginal queries on $S$. We call this notion of error the weighted root mean squared error. This is similar to error measures considered in prior work, e.g., by McKenna et al.~\cite{McKenna_Miklau_Hay_Machanavajjhala_2023}, and Xiao et al.~\cite{xiao2023optimal}. 
Our algorithm 
is inspired by the Fourier theoretic approach of~\cite{BCDKMT07}. We first compute the Fourier coefficients of the empirical distribution of $D$, and add independent Gaussian noise to each resulting Fourier query, with the variance of the noise carefully chosen to minimize the total noise added. This step already achieves the privacy guarantees. We then use the inverse Fourier transform to reconstruct answers to the marginal queries from the noisy Fourier query estimates. The resulting algorithm is roughly as efficient as the baseline Gaussian mechanism, and has running time polynomial in the dataset size, the dimension $d$, and the number of queries. While this approach is similar to that of Barak et al.~\cite{BCDKMT07}, we extend it to Gaussian noise, non-binary domains, and optimize it by choosing the noise variances non-uniformly. 
We discuss our approach in more details in \cref{sec:overview}.
\ifpods
\else
The algorithm is described precisely as Algorithm~\ref{alg:gen}, and its guarantees for weighted root mean squared error are given in Theorem~\ref{thm:upper-bound-general-avg}. 
\fi


\paragraph{Optimality among Factorization Mechanisms.} This algorithm is an instance of the factorization mechanisms framework of Edmonds, Nikolov, and Ullman~\cite{EdmondsNU20}, which itself generalizes the matrix mechanism of Li, Miklau, Hay, McGregor, and Rastogi~\cite{LiMHMR15}.\footnote{It is common to refer to all factorization mechanisms as matrix mechanisms. We prefer the name factorization mechanisms, and reserve the name matrix mechanism for the original mechanisms proposed in~\cite{LiMHMR15}.} In general, a factorization mechanism ``factors'' the queries $Q$ into strategy queries $R$ and a reconstruction matrix $L$. The strategy queries should be linear, in the sense that we can write them as $R(D) = \sum_{i = 1}^n R(x^{(i)})$. The true query answers are given by $LR(D)$. To use this factorization as a private mechanism, we release $R(D)$ using the Gaussian mechanism, and multiply the resulting private estimate by $L$. In the case of our algorithm, the strategy queries are given by the Fourier queries, and the reconstruction matrix is derived from the inverse Fourier transform.

Factorization mechanisms have been the subject of intense research, and can significantly improve over the Gaussian mechanism baseline in many settings~\cite{McKenna_Miklau_Hay_Machanavajjhala_2023,HenzingerUU23,xiao2023optimal,HenzingerU25,LiuUZ24,lebeda2024} (see also the recent survey~\cite{factmech-survey}). The class of factorization mechanisms is also equivalent to mechanisms that add unbiased correlated Gaussian noise to the true query answers~\cite{NikolovT24}. Furthermore, factorization mechanism are known to achieve nearly optimal error among all differentially private mechanisms in several important settings~\cite{NikolovTZ16-GeometryOfDPSparseApproximate,EdmondsNU20,NikolovT24}. For example, factorization algorithms are optimal up to absolute constants among unbiased differentially private estimates of the true query answers~\cite{NikolovT24}, and among all differentially private algorithms when the dataset size $n$ is large enough~\cite{EdmondsNU20}. Similar techniques have also been used to reduce error for hierarchical queries~\cite{DawsonGK0KLMMNS23}. Nevertheless, in other settings biased and data-dependent algorithms may achieve smaller error~\cite{HayRMS10}. 

As our next contribution, we show that our mechanism achieves \emph{optimal weighted root mean squared error} among all factorization mechanisms for any marginal workload.
More precisely, no factorization mechanism can achieve smaller weighted average noise variance. 
This result is given in \ifpods\else\cref{thm:fact-main}\fi. 

Let us remark here that it is possible to optimize error (measured as any linear function of the noise variances) over all factorization mechanisms using semidefinite programming~\cite{EdmondsNU20}. The resulting running times are polynomial in the number of queries $|Q_{\SS}|$ and the universe size $|\uni|$. In the case of marginal queries, however, $|\uni|$ is exponential in $d$ or worse, and this running time is prohibitive. Our results, by contrast, achieve running times polynomial in $d$ and give explicit factorizations.

Next, we consider the problem of minimizing the maximum variance over all queries in an arbitrary workload of marginals. Here we reuse our algorithm for weighted root mean squared error. We choose weights $p^*(S)$ for $S \in \SS$ that sum to $1$, and maximize the weighted root mean squared error. Since the resulting maximization problem is concave in the weights, we can solve it efficiently using standard methods. Once the $p^*(S)$ weights are computed, we simply run Algorithm~\ref{alg:gen}. First order optimality conditions show that, for these weights, the weighted root mean squared error equals the maximum variance, as $p^*$ is supported on queries for which the noise has maximum variance. Once again, the algorithm runs in polynomial time in $n$, $d$, and $|Q_{\SS}|$, and it achieves \emph{optimal maximum variance} among all factorization mechanisms. The guarantees of the algorithm are given in
\ifpods
\else
Theorem~\ref{thm:upper-bound-general}, and its optimality is proved in Theorem~\ref{thm:fact-main}. 

It is worth noting that, unlike the algorithm described above, this one is not a ``closed form'' solution, since it relies on an optimization routine to find $p^*$. Nevertheless, the algorithm still has the same structure of measuring Fourier queries and reconstructing marginal query answers using the inverse Fourier transform, and only the variance of the noise added to each Fourier query depends on $p^*$. Moreover, $p^*$ can be found much more efficiently than optimizing over all factorization mechanisms.
\fi

We note that there are known asymptotic expressions for the optimal error achievable on the workload of all $k$-way marginals over binary domains by either factorization~\cite{EdmondsNU20,NikolovT24}, or arbitrary differentially private mechanism~\cite{BunUV14}. By contrast, we are interested in exactly optimal factorizations and exact expressions for their error.

\paragraph{Extensions.} We extend our technique to more expressive workloads, which we refer to as product queries. We associate a function $\phi_j:\uni_j \to \R$  to each attribute $j \in [d]$ and define, for $S \subseteq [d]$ and $t \in \uni_S$, the queries $q^\phi_{S,t}(D)$ as  $\sum_{i=1}^n\prod_{j \in S}\phi_j(t_j-x^{(i)}_j)$, where the difference is interpreted $\bmod\,m_j$. Using the same approach as for marginals, we give factorization mechanisms for arbitrary workloads of product queries, and prove their optimality among all factorization mechanisms with respect to weighted root mean squared error, and maximum variance. The mechanisms are roughly as computationally efficient as the ones for marginals. The upper bound is given as \cref{thm:upper-bound-product-avg}~and~\cref{thm:upper-bound-product} in~\ifpods\else\cref{sec:weighted-product-queries}\fi, and the lower bound is in~\cref{thm:fact-main-product}. 

We then show that we can embed any workload of extended marginal queries into a workload of product queries, to which we can apply our mechanism (\cref{thm:upper-bound-pref-avg}). This approach is similar to the design of explicit factorization mechanisms for prefix (i.e., threshold) queries by embedding the into ``circulant queries'', i.e., queries that count points in an interval on a circle~\cite{Choquette-ChooM23,HenzingerU25}. Indeed, prefix queries are a special case of extended marginals when there is only one numerical attribute, and our mechanism for this special case reduces to the factorization mechanism of Henzinger and Upadhyay~\cite{HenzingerU25}. We also show lower bounds for extended marginal queries that match the error of our mechanisms up to lower order terms, both for root mean squared error and maximum variance (\cref{thm:ext-marginals-lb}). In particular, the lower and upper bounds converge to the same value as the minimum domain size of numerical attributes grows to infinity. For the special case of prefix queries, our lower bounds are slightly weaker than the best known lower bound~\cite{matouvsek2020factorization}, but only by an additive constant. At the same time they are significantly more general. 

\paragraph{Comparison with ResidualPlanner.}
In closely related prior work,
Xiao et al.~\cite{xiao2023optimal} proposed a mechanism, ResidualPlanner, which 
efficiently computes private estimates for an arbitrary workload of marginals, and achieves optimal error among all factorization mechanisms for a class of error measures that includes weighted root mean squared error and maximum variance.
The authors define a \emph{subtraction matrix} and construct \emph{residual queries} by combining subtraction matrices using the Kronecker product. These residual queries serve as the strategy queries for the mechanism.
ResidualPlanner computes all residual queries required for the marginal workload, adds \emph{correlated} Gaussian noise to each vector of residual query answers, and recovers marginal estimates by inverting the transformation. 



To speed up processing time, ResidualPlanner represents martices implicitly. 
The authors show how to optimize the privacy budget used for each residual query for each error measure in the class they consider. Notably, for weighted root mean squared error, they give closed form expressions for the optimal privacy budget allocation. The extended version of the paper \citep{xiao2025optimal} (which is concurrent with this work) introduces ResidualPlanner+, which supports more general workloads similar to our product queries, but requires an externally provided \emph{strategy replacement matrix}, which is used to construct the subtraction matrix. 

There is significant overlap between our results and \cite{xiao2023optimal,xiao2025optimal}, but the techniques we use are significantly different. Next we go into more details about how our results compare with ResidualPlanner and ResidualPlanner+, and argue that our Fourier-theoretic approach offers some advantages.
First, we improve over Residual Planner~\cite{xiao2023optimal} in the following ways:
\begin{itemize}
    \item Running time: Reconstructing estimates of marginal queries is the bottleneck for the running time of our mechanisms. \cite{xiao2023optimal} reconstruct estimates to a $k$-way marginal with domain $\uni_S$ in time $O(k \vert \uni_S \vert^2)$.
    We reconstruct estimates in time $O(\vert \uni_s \vert \log (\vert \uni_S \vert))$. 
    \ifarxiv\christian{They claim $k \vert \uni_s \vert^2$ in the Theorem. The proof suggests that they might actually achieve the stronger $k 2^k \vert \uni_s \vert$. We still beats this unless $k < loglog(m)$.}\fi
    \item Simplicity: We argue that our mechanism is simpler than ResidualPlanner. Our strategy queries compute Fourier coefficients, rather than using Kronecker products of custom matrices.
    Since our technique releases a number of sensitivity $1$ queries, our privacy proofs are significantly simpler. Computing the variance of any marginal estimate is similarly easy.
    \item Explicit factorization: A technical difference from our technique is that the noise added to the answers of a residual query in ResidualPlanner is correlated for non-binary data, which is not standard for the factorization framework. While ResidualPlanner can be expressed equivalently as a standard factorization mechanism, Xiao et al.~state they avoid doing so because of the complexity of the re-formulation. By contrast, we show that our mechanisms can be expressed as standard factorization mechanisms in a straightforward way.
    \item Simpler lower bounds: Xiao et al.~show that ResidualPlanner is optimal among all factorization mechanisms for a class of error measures via a symmetry argument. They argue that the covariance matrix of the noise distribution of a factorization mechanism must satisfy certain symmetries, induced by the structure of marginal queries, and that ResidualPlanner finds an optimal mechanism with these symmetries. The technical details of this argument are fairly involved. By contrast, we derive simple necessary and sufficient optimality conditions for \emph{any} factorization mechanism.  Verifying that our mechanism satisfies these optimality conditions is then straightforward.
\end{itemize}
We improve over ResidualPlanner+ for generalizations of marginals~\cite{xiao2025optimal} in the following ways:
\begin{itemize}
    \item External input: We give simple explicit algorithms for estimating any workload of product queries. ResidualPlanner+ relies on \emph{strategy replacement} matrices as external input to the framework.
    \item Optimality results: ResidualPlanner+ provides no strong theoretical optimality guarantees for the error it achieves, beyond the ones already known for ResidualPlanner. The error they achieve is also heavily dependent on the externally provided strategy replacement matrices. In contrast, we give matching upper and lower bounds for product queries both for weighted root mean squared error and maximum variance, using the same technique we used for marginal queries. We also give a lower bound for extended marginal queries that matches the leading term of our upper bound. Here we use the singular value lower bound on the error of factorization mechanisms~\cite{LiM13}, and use insights from the lower bound for product queries in order to compute the necessary estimates of sums of singular values for extended marginals. Note that, by contrast, it is unclear how to adapt the symmetry argument used to show the optimality of ResidualPlanner to extended marginal queries. 
\end{itemize} 

\subsection{Technical Intuition}
\label{sec:overview}


Here we present the central idea behind our approach.
For simplicity, we present the results for $1$-GDP\footnote{The definition of GDP is deferred to \ifpods\else\cref{sec:prelim-dp}\fi. In this section, we only use the fact that the Gaussian mechanism satisfies differential privacy, or, formally, $1$-GDP, and that post-processing preserves it.} and focus on estimating all $2$-way marginal queries over binary attributes. 
In the end of this section we discuss how we generalize the technique to other settings.
All results apply naturally to $\mu$-GDP by scaling all noise samples by $1/\mu$.

\ifpods \else The Gaussian mechanism would release estimates for all $2$-way marginals by adding independent noise from $\mathcal{N}(0, {d \choose 2})$ to each query.
We aim to reduce the magnitude of noise by taking advantage of the inherent correlation between marginal queries. \fi

We privately estimate aggregate queries in the Fourier basis. 
In the case of binary attributes, the queries we are interested in are relatively simple. 
For all subsets of $[d]$ with at most $2$ elements we answer aggregate queries of the form
\[
\fq_{\emptyset}(D) = \sum_{i \in [\vert D \vert]} (-1)^{0} = \vert D \vert\,; \quad
\fq_{\{j\}}(D) = \sum_{i \in [\vert D \vert]} (-1)^{x_j^{(i)}}\,; \quad
\fq_{\{j, k\}}(D) = \sum_{i \in [\vert D \vert]} (-1)^{x_j^{(i)} + x_k^{(i)}}.
\]
These queries are,  up to scaling, the Fourier coefficients up to level $2$ of the empirical distribution of dataset points. 

It is easy to see that each of these queries has sensitivity $1$. 
We can thus release these queries privately by adding noise from $\mathcal{N}(0, {d \choose 2} + d + 1)$ to each of them. 
However, we can do slightly better if we add less noise to the queries that we will reuse for multiple marginal queries later. 
Similarly to \cite{lebeda2024} whose mechanism handles the case of $1$-way marginals, we scale the variance inversely proportional to the square root of the number of queries in which each Fourier coefficient appears. Specifically, for all $j, k \in [d]$ we release the following:
\begin{align*}
\sigma^2 = {d \choose 2} + d\sqrt{d-1} + \sqrt{d \choose 2}
\quad & \quad \tilde \fq_{\emptyset}(D) = \fq_{\emptyset}(D) + \mathcal{N}\left(0, \sigma^2/\sqrt{d \choose 2}\right) \,; \\
\tilde \fq_{\{j\}}(D) = \fq_{\{j\}}(D) + \mathcal{N}\left(0, \sigma^2/\sqrt{d - 1}\right) 
\quad & \quad \tilde \fq_{\{j, k\}}(D) = \fq_{\{j, k\}}(D) + \mathcal{N}(0, \sigma^2) \,.
\end{align*}

We can release all the noisy queries above under our privacy constraints. 
We omit the proof here.

From these noisy estimates, we can recover unbiased estimates for the $2$-way marginals as post-processing. Notice that for any $j \in [d]$, $t_j \in \{0,1\}$, and any $x \in \{0,1\}^d$,
\[
\frac{1 + (-1)^{t_j} (-1)^{x_j}}{2} = 
\begin{cases}
    1 & t_j = x_j\\
    0 & t_j \neq x_j
\end{cases}.
\]
Therefore, for any $j, k \in [d]$ and $t \in \{0,1\}^{\{j,k\}}$,  
\begin{align*}
    q_{\{j,k\},t}(x) 
&\coloneq \mathbbm{1}\{t_j = x_j\} \cdot  \mathbbm{1}\{t_k = x_k\} 
= \frac{1 + (-1)^{t_j} (-1)^{x_j}}{2}\cdot \frac{1 + (-1)^{t_k} (-1)^{x_k}}{2} \\
&= \frac14(1 + (-1)^{t_j} (-1)^{x_j} + (-1)^{t_k} (-1)^{x_k} + (-1)^{t_j + t_k}(-1)^{x_j + x_k}) \\
&= \frac14(\fq_\emptyset(x) + (-1)^{t_j} \fq_{\{j\}}(x) + (-1)^{t_k} \fq_{\{k\}}(x) + (-1)^{t_j + t_k}\fq_{\{j,k\}}(x)) \,.
\end{align*}
As a result, a marginal query on the attributes $\{j,k\}$ with assignment $t \in \{0,1\}^{\{j,k\}}$ can be estimated by
\[
 \tilde{q}_{\{j,k\}, t}(D) = \frac{1}{4}\left( \tilde \fq_{\emptyset}(D) + (-1)^{t_j} \cdot \tilde \fq_{\{j\}}(D) + (-1)^{t_k} \cdot \tilde \fq_{\{k\}}(D) + (-1)^{t_j + t_k} \cdot \tilde \fq_{\{j,k\}}(D)  \right) \,.
\]
Then the error $\tilde{q}_{\{j,k\}, t}(D)-{q}_{\{j,k\}, t}(D)$ is distributed as a mean zero Gaussian with variance
\[
    \frac{1}{4^2} \left(1 + 2/\sqrt{d-1} + 1/\sqrt{d \choose 2} \right) \cdot \left({d \choose 2} + d\sqrt{d-1} + \sqrt{d \choose 2}\right) \,.
\]

The expression above approaches ${d \choose 2}/16$ as $d$ increases, where the main source of the error is the estimate of $\tilde \fq_{\{j,k\}}(D)$. By contrast, a straightforward application of the Gaussian mechanism would add noise with variance ${d\choose 2}$ to each query.
The approach above generalizes naturally to $k$-way marginals by estimating Fourier coefficients up to level $k$.
Each marginal query for binary attributes can be recovered using $2^k$ Fourier queries, and the standard deviation of the error approaches $2^{-k} \sqrt{d \choose k}$ as $d$ increases.
In the more general setting where an attribute has one of $m$ values the Fourier coefficients are complex roots of unity, but the sensitivity is still bounded by $1$.
The improvement over i.i.d. noise is not as large as for binary attributes, since the marginals are not as correlated.
Nevertheless, we still achieve an improvement and the standard deviation of the error approaches
$(1 - 1/m)^{k} \sqrt{d \choose k}$.

Then, in \ifpods\else\cref{sec:upper-bound-arbitrary-marginal} \fi we consider much more general workloads, in which attributes can have different domain size and we assign a weight to the error of marginal queries for each subset of attributes.
The algorithm is slightly more complicated in this setting, but the underlying idea is the same.
We privately estimate the Fourier coefficients required to recover all non-zero weight marginals, and we allocate additional privacy budget to important coefficients that are reused for many marginals.
Finally, in \ifpods\else\cref{sec:upper-bound-product-and-extended} \fi we extend our technique to product queries, and to extended marginals. The underlying idea of the mechanism still remains the same, but the product queries affects the privacy budget allocation.

\subsection*{Organization}

\ifpods
\else
The remainder of the paper is organized as follows.
In \cref{sec:preliminaries} we present the problem of privately estimating marginal queries and provide the relevant background.
In \cref{sec:algorithm} we present our mechanisms for privately releasing marginal queries. 
In \cref{sec:upper-bound-product-and-extended} we extend our mechanism to product queries and extended marginals. 
In \cref{sec:lower-bound} we show lower bounds for factorization mechanisms for answering weighted marginal and product queries that match our upper bounds, as well as a lower bound for extended marginals showing our mechanism is optimal up to lower order terms.
\fi

\section{Preliminaries}
\label{sec:preliminaries}

\subsection{Marginal Queries}
\label{sec:prelim-marginals}

We consider data points $x \in \uni$ with $d$ attributes where the $i$-th attribute has domain $\uni_i \coloneq \{0,1,\dots,m_i-1\}$.
In general, we have $\uni \coloneq \prod_{i \in [d]} \uni_i$. A commonly studied special case is binary attributes where we have $x \in \{0,1\}^d$.

A $k$-way marginal query is parameterized by a set $S \subseteq [d]$ where $\vert S \vert = k$ and an assignment $t \in \uni_S \coloneq \prod_{i \in S} \uni_i$.
For a data point $x \in \uni$ a marginal query $q_{S,t} \colon \uni \rightarrow \{0,1\}$ evaluates to
\[
q_{S, t}(x) = \begin{cases}
    1 & \text{if } \forall i \in S : x_i = t_i, \\
    0 & \text{otherwise.}
\end{cases}
\]

We define a marginal query of a dataset $D$ with $n$ data points as $q_{S,t}(D) = \sum_{i=1}^{\vert D \vert} q_{S,t}(x^{(i)})$, where $x^{(i)}$ is the $i$-th data point in $D$ for some arbitrary order.
Our goal is to privately estimate the value of marginal queries for a dataset under differential privacy. 
Note that for any $S \subseteq [d]$, there are a total of $\vert \uni_S \vert = \prod_{i \in S} m_i$ marginal queries, one for each possible assignment of attributes in $S$. 
We always assume that, once we choose $S$, we ask each possible marginal query for each possible setting of the attributes in $S$.
This is standard practice both in research papers~(e.g. \cite{McKennaSM19, McKenna_Miklau_Sheldon_2021}), and in deployments of differential privacy~(e.g. \cite{censusTopDown}).
Specifically, we are given a workload $Q_{\SS}$ of marginal queries defined by a collection $\SS$ of subsets of $[d]$. 
For each subset $S \in \SS$, the workload contains all queries $q_{S,t}$ for all $t \in \uni_S := \prod_{i \in S}\uni_i$. 
We use the notation $\SS_{\downarrow}$ for the collection of all sets $R \subseteq [d]$ such that $R\subseteq S$ for some $S \in \SS$. I.e., this is the closure of $\SS$ under taking subsets.

\ifarxiv
\paragraph{Extensions}

We extend our technique to a generalization of marginal queries that we call product queries. We define product queries by functions $\phi = (\phi_1, \ldots, \phi_d)$, $\phi_j:\uni_j \to \R$, and a query $q^\phi_{S,t}$ on a single data point is defined by 
\[
q^\phi_{S,t}(x) := \prod_{j \in S} \phi_j((t_j - x_j)\bmod m_j).
\]
We recover standard marginals $q_{S,t}$ from product queries by setting $\phi_j(z) := \mathbbm{1}\{z = 0\}$.

We also consider an extension of marginals where attributes can be categorical or numerical. The queries for categorical attributes match standard marginal queries. For numerical attributes, we consider prefix or suffix predicates. 
We thus define $T_i:= \uni_i$ for $i \in C$, and $T_i := \{-m_i, \ldots, 0, \ldots, m_i-1\}$ 
for $i\in N$, and $T_S := \prod_{i \in S}T_i$, where $C$ and $N$ are the set of categorical and numerical attributes.
We can now redefine the query $q_{S,t}$ for $S \subseteq [d]$ and $t \in T_S$ as 
\[
q_{S,t}(x) := \left(\prod_{j \in S \cap C} \mathbbm{1}\{x_j = t_j\}\right)\left(\prod_{\substack{j \in S \cap N\\t_j \ge 0}}\mathbbm{1}\{x_j \le t_j\}\right)\left(\prod_{\substack{j \in S \cap N\\t_j < 0}}\mathbbm{1}\{x_j \ge |t_j|\}\right)
\]
We show how to embed extended marginals as product queries in \cref{sec:product-apps}.

\fi

\subsection{Differential Privacy}
\label{sec:prelim-dp}


Differential privacy~\cite{dwork06calibrating} is a framework for preserving privacy by ensuring that the distributions of a mechanism for any pair of neighboring datasets are not too far apart.
The difference between neighboring datasets, as defined below, corresponds to adding or removing all data about any individual data point.
Several definitions of differential privacy exists, and
our results apply to any variant satisfied by the Gaussian mechanism. 
We present our results using Gaussian Differential Privacy because it exactly describes the privacy guarantees of additive Gaussian noise.
Informally, GDP ensures that the pair of distributions for neighboring datasets is at least as hard to distinguish as two normal distributions.

\begin{definition}[Neighboring datasets]
    \label{def:neighboring}
    Two datasets $D$ and $D'$ are neighboring, denoted $D \sim D'$, if we can obtain one dataset of the pair by adding one data point to the other dataset.
\end{definition}

\ifarxiv
\begin{definition}[{Gaussian Differential Privacy~\cite[Definition~4]{DongRS22-GDP}}]
    \label{def:gdp}
    A randomized mechanism $\mathcal{M} \colon \uni^* \rightarrow \mathcal{R}$ satisfies $\mu$-GDP if for all pairs of neighboring datasets $D \sim D'$ it holds that
    \[
        T(\mathcal{M}(D),\mathcal{M}(D')) \geq T(\mathcal{N}(0, 1), \mathcal{N}(\mu, 1)) \,,
    \]
    where $T(P,Q): [0,1] \rightarrow [0,1]$ denotes the trade-off function for two distributions $P$ and $Q$ defined on the same space. 
    The tradeoff function is defined as 
    \[         
        T(P,Q)(\alpha) = \inf\{\beta_\phi : \alpha_\phi \leq \alpha \} \,,     
    \] 
    where the infimum is taken over all (measurable) rejection rules $\phi$, and $\alpha_\phi$ and $\beta_\phi$ denote the type I and type II error rates, respectively. 
\end{definition}

The Gaussian mechanism~\cite{DinurNissim03,DworkN04,DworkKMMN06OurDataOurselves} is one of the most important tools in differential privacy. 
The mechanism achieves the desired privacy guarantees by adding unbiased Gaussian noise to all queries scaled by the $\ell_2$ sensitivity. 
Applying the Gaussian mechanism directly to our setting gives us a baseline for privately estimating marginal queries.

\begin{lemma}[The Gaussian mechanism]
    \label{lem:gaussian}
    Let $q \colon \uni^* \rightarrow \mathbb{R}^d$ be a set of queries with $\ell_2$ sensitivity $\Delta q \coloneq \max_{D \sim D'} \|q(D) - q(D')\|_2$. 
    Then the mechanism that outputs $q(X) + Z$ where $Z \sim \mathcal{N}\left(0, \frac{(\Delta q)^2}{\mu^2} I_d\right)$ satisfies $\mu$-GDP.
\end{lemma}

\begin{lemma}[Gaussian noise for marginal queries]
    \label{lem:baseline}
    Let $\SS = (S_1,\dots,S_m)$ be a collection of sets such that $S_i \subseteq [d]$.
    Then the mechanism that for each $i \in [m]$ and each assignment $t \in \uni_{S_i}$ independently samples noise $Z_{S_i,t} \sim \mathcal{N}(0, m/\mu^2)$ and releases $q_{S_i, t}(D) + Z_{S_i,t}$ satisfies $\mu$-GDP.
\end{lemma}

\begin{proof}
    Notice that for each $S_i$, adding or removing a data point changes exactly one marginal query by $1$ while the remaining $(\prod_{i \in S} \vert \uni_i \vert) - 1$ queries are unaffected. 
    The $\ell_2$ sensitivity for all queries in $\SS$ is thus $\sqrt{\sum_{i \in [m]} 1^2} = \sqrt{m}$.
    The privacy guarantee follows from Lemma \ref{lem:gaussian}.
\end{proof}
\fi

\ifarxiv
We use some standard properties of differential privacy in our proofs.

\begin{lemma}[Post-processing~{\cite[Proposition~4]{DongRS22-GDP}}]
    \label{lem:post-processing}
    Let $\mathcal{M} \colon \mathcal{U}^* \rightarrow \mathcal{R}$ denote any $\mu$-GDP mechanism. 
    Then for any (randomized) function $g \colon \mathcal{R} \rightarrow \mathcal{R}'$ the composed mechanism $g \circ \mathcal{M} \colon \mathcal{U}^* \rightarrow \mathcal{R}'$ also satisfies $\mu$-GDP.
\end{lemma}

\begin{lemma}[Composition~{\cite[Corollary~3.3]{DongRS22-GDP}}]
    \label{lem:composition}
    Let $\mathcal{M}_1 \colon \mathcal{U}^* \rightarrow \mathcal{R}_1$ and $\mathcal{M}_2 \colon \mathcal{U}^* \rightarrow \mathcal{R}_2$ denote a pair of mechanisms that satisfies $\mu_1$-GDP and $\mu_2$-GDP, respectively. 
    Then the mechanism $\mathcal{M}(D) = (\mathcal{M}_1(D), \mathcal{M}_2(D))$ that outputs the result from both mechanisms satisfies $\sqrt{\mu^2_1 + \mu_2^2}$-GDP.
\end{lemma}
\fi

\ifpods
\else
\subsection{Complex Numbers} 
\label{sec:prelim-complex}

For any complex number $z$, we use $\overline{z}$ for its complex conjugate, and \[\abs{z} = \sqrt{z\overline{z}} = \sqrt{\mathrm{Re}(a)^2 + \mathrm{Im}(a)^2}\] for the absolute value. We recall the standard inner product over $\C^d$, defined as $\ip{x,y} = \sum_{i=1}^d x_i \overline{y_i}$. The standard $\ell_2$ norm on $\C^d$ is $\|x\|_2 = \sqrt{\ip{x,x}} = \sqrt{\sum_{i=1}^d \abs{x_i}^2}$. For a matrix $A$ with complex entries, we use $A^*$ for its conjugate transpose, i.e., $A^*_{i,j} = \overline{A_{j,i}}$ for all $i$ and $j$.

\begin{definition}[Roots of unity]
    \label{def:roots-unity}
    For any positive integer $n$, the $n$-th roots of unity is the set of complex numbers $z$ satisfying $z^n = 1$. For the primitive $n$-th root of unity, we write 
    \[
      \omega_n \coloneq e^{2\pi i / n} = \cos(2\pi/n) + i \cdot \sin(2\pi/n) \,.
    \]
\end{definition}
Note that all roots of unity and their powers lie on the unit circle in $\C$, i.e., $|\omega_n| = 1$.

\begin{lemma}[Multidimensional fast Fourier transform~\cite{cooley1965algorithm}]
    \label{lem:multidim-fft}
    Let $\uni := \uni_1 \times \ldots \times \uni_d$, where $\uni_i := \{0, \ldots, m_i - 1\}$, and let $x \in \C^\uni$ be a $d$-dimensional vector of complex values indexed by $\uni$. 
    Consider for any $t \in \uni$ the $d$-dimensional discrete Fourier transform (DFT) 
    \[
        y_t = \sum_{a \in \uni} x_a \prod_{i = 1}^d \omega_{m_i}^{a_it_i} \,.
    \]
    All entries of $y \in \C^\uni$ can be computed from $x$ in time $O(m \log m)$ where $m = \prod_{i \in [d]} m_i$ using a multi-dimensional DFT (see e.g. \cite[Section~30-3]{CLRS} for a construction).
\end{lemma}

\begin{definition}[Complex Gaussian distribution]
    \label{def:complex-gaussian}
    We denote the complex zero-mean normal distribution as $\mathcal{CN}(0, \sigma^2)$, where $\sigma^2 \ge 0$ is the variance.
    If $Z \in \mathbb{C}$ is a sample from the distribution $Z \sim \mathcal{CN}(0, \sigma^2)$, then the real and imaginary parts of $Z$ are distributed as independent samples from $\mathcal{N}(0, \sigma^2/2)$.

    More generally, we denote the complex $d$-variate zero-mean normal distribution as $\mathcal{CN}(0, \Sigma)$, where the covariance matrix $\Sigma \in \C^{d\times d}$ is a Hermitian positive semidefinite matrix. $\mathcal{CN}(0, \Sigma)$ is the probability distribution on $\C^d$ with probability density function \[p(z) := \frac{1}{\pi^d \sqrt{\det(\Sigma)}}\exp\left(-\ip{\Sigma z, z}\right).\]
\end{definition}

The following lemma, which is standard, will be useful in analyzing our algorithms.

\begin{lemma}\label{lem:gauss-lincomb}
    Suppose that $Z\sim \mathcal{CN}(0,\Sigma)$ is a $d$-variate normally distributed random vector, and that $A$ is an $\ell \times d$ matrix with complex entries. Then $AZ \sim \mathcal{CN}(0,A\Sigma A^*)$. In particular, if $Z_1, \ldots, Z_d$ are independent, and $Z_i \sim \mathcal{CN}(0,\sigma_i^2)$, and $a \in \C^d$ then $\sum_{i=1}^d a_i Z_i \sim \mathcal{CN}(0,\sum_{i=1}^d |a_i|^2 \sigma_i^2)$.
\end{lemma}


\fi

\ifpods
\else
\section{Optimal Factorization for Weighted Marginal Queries}
\label{sec:algorithm}

In this section we introduce our technique for adding noise to marginal queries. 
We first show how the marginal queries can be recovered from aggregate queries in the Fourier basis of $\uni$. 
This observation immediately yields a correlated Gaussian noise mechanism, or, equivalently, a factorization that achieves noise with lower variance than the standard Gaussian mechanism. 
It is then easy to observe that some of the queries in the Fourier basis are used for more marginal queries than others. 
By allocating privacy budget weighted by the importance of the queries, we can further reduce the error of our factorization.
In \Cref{sec:lower-bound}, we show that our factorization is, in fact, optimal! 

\subsection{Marginal Queries in the Fourier Basis}
\label{sec:fourier-prelim}

Let $a \in \uni$ be a choice of value for each attribute, and recall that $\omega_{m_i} \coloneq \exp(2\pi i / m_i)$ is an $m_i$-th root of unity. Our algorithm relies on aggregate queries of the form
\begin{equation}
    \label{eq:fourier-query}
    \fq_a(D) := \sum_{i = 1}^{\vert D \vert} \overline{\chi_a(x^{(i)})}
    \text{\ \ where\ \ }
    \chi_{a}(x) :=  \prod_{j =1}^d \omega_{m_j}^{a_j \cdot x_j}.
\end{equation}
As in the example of binary domains from \cref{sec:overview}, these queries give, up to scaling, the Fourier coefficients of the empirical distribution of $D$. The functions $\chi_a$ are the Fourier characters, which form an orthogonal basis for functions on $\uni$. 
In the common special case where $\vert \uni_1 \vert = \vert \uni_2 \vert = \dots = \vert \uni_d \vert = m$, the queries take the form
\begin{equation}
    \label{eq:fourier-same-domain}
    \fq_a(D) := \sum_{i = 1}^{\vert D \vert} \omega_{m}^{-\ip{a,x}}.
\end{equation}
We denote the support of $a$ by $\supp(a) \coloneq \{j: a_j \neq 0\}$. The size of the support, i.e., the weight of $a$, is denoted by $\|a\|_0$ and is an important parameter.

It is easy to see that adding or removing one data point from $D$ always changes the value of any $F_{a}(D)$ by a unit complex number (see Definition~\ref{def:roots-unity} for background).
We use this fact later to bound the sensitivity for privatizing the queries.
Note that in the special case where $a = 0$ all values of $\uni$ evaluate to $1$, so $F_0(D) = \vert D \vert$ is the dataset size.

Next, we show how we recover marginal queries using the aggregate Fourier queries above using the inverse discrete Fourier transform. 
The main observation is that, for any $x\in \uni$, $j \in [d]$, and $t_j \in \uni_j$, 
\[
\frac{\sum_{a=0}^{m_j-1} \omega_{m_j}^{a t_j}\omega_{m_j}^{-a x_j}}{m_j}
=
\frac{\sum_{a=0}^{m_j-1} \omega_{m_j}^{a (t_j-x_j)}}{m_j}
=
\begin{cases}
    1 & t_j = x_j\\
    0 & t_j \neq x_j
\end{cases}, 
\]
and, therefore,
\begin{equation}\label{eq:marginals-inv-x}
q_{S,t}(x) = \prod_{j \in S} \mathbbm{1}\{t_j = x_j\} = 
\prod_{j \in S}\frac{\sum_{a=0}^{m_j-1} \omega_{m_j}^{a t_j}\omega_{m_j}^{-a x_j}}{m_j}
=
\frac{1}{\vert \uni_S \vert}\sum_{\substack{a \in \uni\\ \supp(a) \subseteq S}} {\chi_a(t)}\overline{\chi_a(x)} \,.
\end{equation}
Above, we slightly abused notation: we used $\chi_a(t)$ even though $t_j$ is defined only for $j \in S$. Note, nevertheless, that this is well defined since $\supp(a) \subseteq S$, and $\chi_a$ only depends on coordinates in $\supp(a)$.

Summing over dataset points now gives us
\begin{align}
q_{S,t}(D) = 
\sum_{i = 1}^{\vert D\vert} q_{S,t}(x^{(i)}) 
&= 
\frac{1}{\vert \uni_S \vert}\sum_{\substack{a \in \uni\\ \supp(a) \subseteq S}} {\chi_a(t)}{\fq_a(D)} \,.
\label{eq:marginals-inverse-fourier}
\end{align}

\subsection{Warm-up: Estimating All $k$-way Marginals}
\label{sec:all-kway}

In this subsection we show how to privately answer all $k$-way marginals using the Fourier representation discussed above. 
We focus on the setting where all attributes are defined on the same domain, such that $\vert \uni_1 \vert = \vert \uni_2 \vert = \dots = \vert \uni_d \vert = m$.
We use this simpler problem as a warm-up before discussing the more general setting at the end of the section.

First, we show that we can privately release estimates of the complex-valued Fourier aggregate queries using complex Gaussian noise. 
Then, we show that correlating the noise across marginal queries by simply reusing estimates gives a utility improvement.
We further improve the variance by distributing the privacy budget similar to \cite{lebeda2024}.
In fact, we recover Lebeda's mechanism in the simplest setting of estimating all $1$-way marginals for binary attributes ($m = 2$).

We start with a basic lemma that the complex Gaussian mechanism satisfies GDP, since it is equivalent to running the standard real-valued Gaussian mechanism separately for the real and imaginary components.

\begin{lemma}
    \label{lem:complex-to-real-reduction}
    Let $q \colon \uni^* \rightarrow \mathbb{C}^d$ be a set of queries with $\ell_2$ sensitivity $\Delta q \coloneq \max_{D \sim D'} \|q(D) - q(D')\|_2$. 
    Then the mechanism that outputs $q(D) + Z$ where $Z \sim \mathcal{CN}\left(0, \frac{2(\Delta q)^2}{\mu^2} I_d\right)$ satisfies $\mu$-GDP.
\end{lemma}

\begin{proof}
    Define an alternative query set $\hat{q} \colon \uni^* \rightarrow \mathbb{R}^{2d}$ as 
    \[
    \hat{q}_i(D) = \begin{cases}
        \mathrm{Re}(q_i(D)) & \mathrm{if~} i \leq d, \\ 
        \mathrm{Im}(q_{i - d}(D)) & \mathrm{otherwise.}
    \end{cases}
    \]
    Then for any pair of datasets $D$ and $D'$ we have 
    \begin{align*}
        \|\hat{q}(D) - \hat{q}(D') \|_2 &= \sqrt{\sum_{i \in [2d]} (\hat q(D)_i - \hat q(D')_i) ^2} \\
        &= \sqrt{\sum_{i \in [d]} (\hat q(D)_i - \hat q(D')_i) ^2 + (\hat q(D)_{i + d} - \hat q(D')_{i + d}) ^2} = \|{q}(D) - {q}(D') \|_2 \,,
    \end{align*}
    where the last equality follows from $(\hat q(D)_i - \hat q(D')_i) ^2 + (\hat q(D)_{i + d} - \hat q(D')_{i + d}) ^2 = \vert {q}(D)_i - {q}(D')_i \vert ^2$ by definition of $\hat q$.  
    As such, we have $\Delta q = \Delta \hat q$ and we can release $\hat q(D) + \hat Z$ where $\hat Z \sim \mathcal{N}\left(0, \frac{(\Delta q)^2}{\mu^2} I_{2d}\right)$ under $\mu$-GDP by \cref{lem:gaussian}.
    If we post processes $\hat q(D) + \hat Z$ by constructing complex numbers such that $\mathrm{Re}(\tilde q_i(D)) = (\hat q(D) + \hat Z)_i$ and $\mathrm{Im}(\tilde q_i(D)) = (\hat q(D) + \hat Z)_{i + d}$ then $\tilde q$ is distributed as $q(D) + Z$. 
    The lemma therefore holds by the post processing property of GDP (\cref{lem:post-processing}).
\end{proof}

We now consider the setting where we want to privately estimate all $k$-way marginals. 
Notice in \cref{eq:marginals-inverse-fourier} that queries $F_{a}(D)$ for which $\|a\|_0 \leq \vert S \vert$ are used to estimate $q_{S,t}(D)$ for any set of attributes $S$ that contains $\supp(a)$. Rather than privately estimating $F_a(D)$ separately for each such $S$, we can, of course, reuse the estimate. In total, to estimate all $k$-way marginals, we need to estimate $F_a(D)$ for each $a$ of weight $\|a\|_0 \le k$, which brings the number of aggregate queries to
\[
    \sum_{j = 0}^k {d \choose j} (m-1)^j \,.
\]
We can then estimate any $k$-way marginal query by post-processing private estimates of the $m^k$ relevant Fourier queries. 
As a final step, we always remove any imaginary part of the estimate, since that must come from noise.
Estimating all $F_a(D)$ queries privately and reusing values already gives us a slight improvement over the baseline where the error is distributed as $\mathcal{N}(0,{d \choose k}/\mu^2)$ since $m^{-k} \sum_{j = 0}^k {d \choose j} (m-1)^j \leq {d \choose k}$. 
However, we can further reduce the error by releasing more accurate answers to Fourier coefficient that are reused.
We scale the variance of each estimate inversely proportional to the square root of the number of queries it is used for.
This scaling is optimized for $\ell_2^2$ error with Gaussian noise as shown for different settings in \cite{AumullerLNP24,DaganJYZZ24,lebeda2024}.
The pseudocode of our mechanism is in \Cref{alg:allkway}.

\begin{algorithm}
\begin{algorithmic}
\State  For each $a \in \uni$ of weight $\|a\|_0 \le k$, compute the value 
    \[
         \fq_a(D) := \sum_{i = 1}^{\vert D \vert} \omega_{m}^{-\ip{a,x^{(i)}}}.
    \]

\State Let 
    \[
       \tau := \frac{1}{\mu^2} \sum_{j=0}^k {d \choose j} (m - 1)^j \sqrt{{d - j}\choose {k - j}} \enspace .
    \]

\State Privately release estimate of each $\fq_{a}(D)$, $\|a\|_0 \le k$, by adding independent noise such that
    \[
        \tilde \fq_{a}(D) =\fq_{a}(D) + Z_{a}, \mathrm{~where~}Z_{a} \sim \mathcal{CN}\left(0, 2\tau \Big/ \sqrt{{d - \|a\|_0} \choose {k - \|a\|_0}}\right) \enspace .
    \]
    
\State Estimates for $k$-way marginal queries can be recovered using post-processing by computing
    \[
      \tilde{q}_{S,t}(D) = \mathrm{Re}\left( \frac{1}{ m^k }\sum_{\substack{a \in \uni\\ \supp(a) \subseteq S}} {\chi_a(t)}{\tilde\fq_a(D)}  \right), \text{ where } \chi_a(t) = \omega_m^{\sum_{i \in S}a_i t_i}.
    \]    
\end{algorithmic}
\caption{Differentially private estimates of all $k$-way marginals.}
\label{alg:allkway}
\end{algorithm}

Next, we show the privacy properties and error of our mechanism.
Later we discuss how to speed up computation time over a trivial implementation by computing estimates using FFT.

\begin{lemma}
    \label{lem:algo-DP}
    Algorithm~\ref{alg:allkway} satisfies $\mu$-GDP.
\end{lemma}
\begin{proof}
    Since the released estimates $\tilde{q}_{S,t}(D)$ are just post-processing of the Fourier query estimates $\tilde{F}_a(D)$, it is enough to show that releasing $\tilde{F}_a(D)$ for all $a \in \uni$ of weight $\|a\|_0 \le k$ satisfied $\mu$-GDP. We show this via composition. Observe that, since the sensitivity of $F_a(D)$ is $1$, by Lemma~\ref{lem:complex-to-real-reduction} releasing $\tilde{F}_a(D)$ satisfies $\mu_a$-GDP, where 
    \(
    \mu_a^2 := {\sqrt{d-\|a\|_0 \choose k-\|a\|_0}}\Big/{\tau}. 
    \)
    Moreover, notice that there are ${d \choose \ell} (m-1)^\ell$ choices of $a$ with weight $\|a\|_0 = \ell$. Therefore, the lemma follows from the composition property of GDP (Lemma~\ref{lem:composition}), since
    \[
    \sum_{a: \|a\|_0 \le k} \mu_a^2 = \sum_{\ell=0}^k {d \choose \ell} (m-1)^\ell \frac{\sqrt{d-\ell \choose k-\ell}}{\tau} = \mu^2.\qedhere
    \]
\end{proof}

\begin{lemma}
    \label{lem:error-k-way}
    For any $S\subseteq [d]$, $|S| = k$, and any $t \in \{0,\dots,m-1\}^S$, the estimate $\tilde{q}_{S,t}(D)$ computed by Algorithm~\ref{alg:allkway} has error $\tilde{q}_{S,t}(D) - {q}_{S,t}(D)$ distributed as $\mathcal{N}(0, \sigma^2)$, where
    \[
        \sigma = \frac{1}{\mu m^{k} \sqrt{{d\choose k}} }\sum_{\ell = 0}^k {d \choose \ell} (m - 1)^\ell\sqrt{{d - \ell} \choose {k - \ell}}.
    \]
\end{lemma}
\begin{proof}
    Let $\tilde{q}_{S,t}'(D)$ be $\tilde{q}_{S,t}(D)$ without removing the imaginary part. By equation~\eqref{eq:marginals-inverse-fourier}, 
    \[
    \tilde{q}'_{S,t}(D) - {q}_{S,t}(D)
    =
    \frac{1}{m^k} \sum_{\substack{a \in \uni\\ \supp(a) \subseteq S}} {\chi_a(t)}{Z_a}.
    \]
    Therefore, by Lemma~\ref{lem:gauss-lincomb}, and since $|\chi_a(t)| = 1$, $\tilde{q}'_{S,t}(D) - {q}_{S,t}(D)\sim \mathcal{CN}(0,\sigma_\C^2)$, where
    \[
    \sigma_\C^2 = 
    \frac{2\tau}{m^{2k}} \sum_{\substack{a \in \uni\\ \supp(a) \subseteq S}} \frac{1}{\sqrt{d-\|a\|_0\choose k-\|a\|_0}}
    = 
    \frac{2\tau}{m^{2k}} \sum_{\ell=0}^k \frac{{k\choose \ell}(m-1)^\ell}{\sqrt{d-\ell\choose k-\ell}}.
    \]
    The second equality holds because there are ${k\choose \ell}(m-1)^\ell$ choices of $a$ with $\supp(a) \subseteq S$ and $\|a\|_0 = \ell$. Using the identity ${d \choose k}{k \choose \ell} = {d\choose \ell}{d-\ell\choose k-\ell}$, and plugging in the value of $\tau$, we get
    \[
    \sigma_{\C}^2 = \frac{2\tau}{m^{2k}}\sum_{\ell=0}^k \frac{{d\choose \ell}(m-1)^\ell\sqrt{d-\ell\choose k-\ell}}{{d\choose k}}
    = \frac{2}{\mu^2m^{2k}{d\choose k} } \left( \sum_{\ell = 0}^k {d \choose \ell} (m - 1)^\ell\sqrt{{d - \ell} \choose {k - \ell}}\right)^2.
    \]
    The lemma now follows after observing that $\tilde{q}_{S,t}(D) - {q}_{S,t}(D)$ is distributed as the real part of $\tilde{q}'_{S,t}(D) - {q}_{S,t}(D)$ which has variance $\sigma^2 = \frac12 \sigma_\C^2$.
\end{proof}

The improvement of \cref{lem:error-k-way} over the Gaussian mechanism baseline with standard deviation $\sqrt{d \choose k}/\mu$ depends on the parameters $d$, $k$, and $m$. 
When $d = k$ both mechanisms have $\sigma = 1/\mu$, and for fixed $k$ and large $m$ and $d$ we have $\sigma \approx ((m-1)/m)^k\sqrt{d \choose k}/\mu$ in \cref{lem:error-k-way}.
The intuition behind this improvement factor is that $(m-1)^k$ of the coefficients used in \cref{eq:marginals-inverse-fourier} are unique to $S$, while the remaining $m^k - (m-1)^k$ coefficients are reused for other marginals.
When $d$ is sufficiently large, we can estimate these remaining coefficients with little noise at only a small privacy cost, because they are reused in many marginals queries.
We plot the relative improvement over the baseline for small parameters in \cref{fig:improvement-factor}.

\begin{figure}[h]
    \centering
    \includegraphics[width=0.6\textwidth]{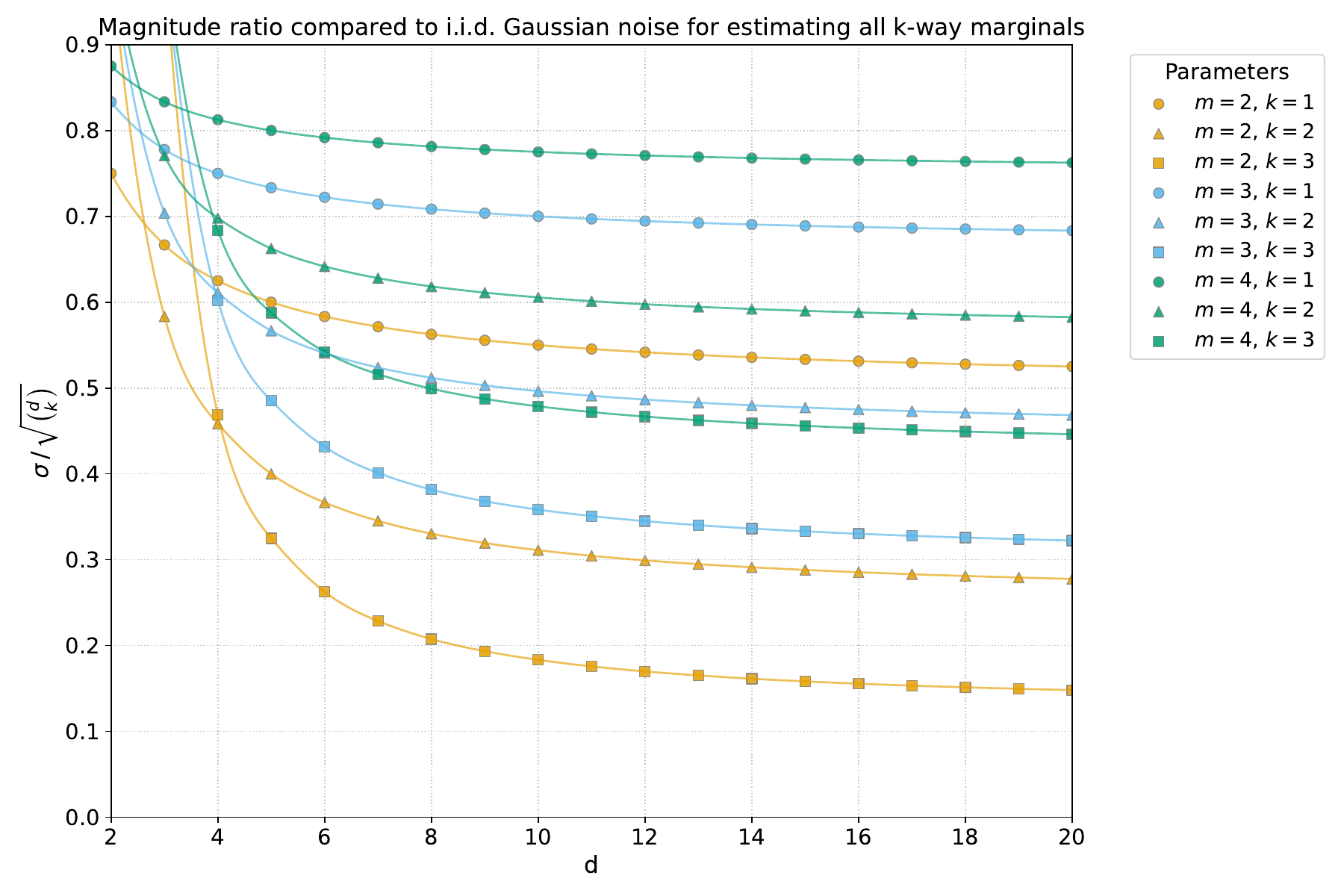}
    \caption{
    Comparison for small values of $k$, $m$, and $d$ of the standard deviation from \cref{thm:upper-bound} relative to the baseline that adds i.i.d. noise with magnitude $\sqrt{d \choose k}$.
    When $d$ increases the improvement ratio approaches $(1-1/m)^k$. 
    }
    \label{fig:improvement-factor}
\end{figure}

Next, we discuss the computation time of \Cref{alg:allkway}.
We focus on the time for adding noise, since all baselines (even non-private solutions) must compute all $q_{S,t}(D)$.
Note that the mechanism that adds i.i.d. noise to each marginal query runs in time $O(\gausstime {d \choose k}m^k)$, where $O(\gausstime)$ is the time needed for sampling a standard Gaussian\footnote{In practice, implementations for (approximately) sampling a standard Gaussian often have randomized running time~(e.g. \cite{Canonne_Kamath_Steinke_2022}). For simplicity, we assume that sampling runs in time $O(\gausstime)$ deterministically. Our results technically only bound the expected running time of our algorithms.}.
Our mechanism uses fewer Gaussian samples because we reuse samples across marginals. 
However, we combine $m^k$ samples for each marginal estimate.
This results in a running time overhead of $O({d \choose k} m^{2k}k)$ 
using a straightforward implementation. 
We can speed up the computation time if we compute all assignments for each $S$ at the same time using FFT.

\begin{lemma}
    \label{lem:FFT-speedup}
    Assume that we sample a standard Gaussian in time $O(\gausstime)$, and that we have access to all non-private $k$-way marginals $q_{S,t}(D)$. 
    Then all private $k$-way marginal estimates $\tilde q_{S,t}(D)$ from \cref{alg:allkway} can be computed in time 
    $O({d \choose k}m^kk\log{(m)} + \gausstime \sum_{\ell = 0}^k {d \choose \ell}(m-1)^\ell)$.
\end{lemma}

\begin{proof}
    It is easy to see that only $2\sum_{\ell = 0}^k {d \choose \ell}(m-1)^\ell$ Gaussian samples are needed. 
    Thus, it suffices to prove that for any fixed $S$ of size $k$, we can compute all marginal queries with support $S$ in time $O(m^k k\log{m})$ given access to the relevant Gaussian samples. 
    This is a standard multi-dimensional FFT, which can be computed by a multi-dimensional FFT algorithm 
    within time $O(m^kk\log{(m)})$ (see \cref{lem:multidim-fft}).
\end{proof}

We are now ready to state one of our main results by summarizing the properties of our technique.

\begin{theorem}
    \label{thm:upper-bound}
    Let $D$ be a dataset containing data points $x \in \uni$, where $\uni \coloneq \{0,1,\dots,m-1\}^d$. 
    Then there exists a $\mu$-GDP mechanism that estimates all $k$-way marginal queries of $D$, where the error of each estimate of $q_{S,t}(D)$ for $\vert S \vert = k$ is distributed as $\mathcal{N}(0, \sigma^2)$ where
    \[
    \sigma = \frac{1}{\mu m^{k} \sqrt{{d\choose k}} }\sum_{\ell = 0}^k {d \choose \ell} (m - 1)^\ell\sqrt{{d - \ell} \choose {k - \ell}} \enspace .
    \]
    Additionally, the noise for all marginal estimates can be sampled in time 
    $O({d \choose k}m^kk\log{(m)} + \gausstime \sum_{\ell = 0}^k {d \choose \ell}(m-1)^\ell)$
    where $O(\gausstime)$ is the time required for sampling a standard Gaussian. 
\end{theorem}

\begin{proof}
    That mechanism is \cref{alg:allkway}. The privacy guarantees, error distribution, and running time follow from
    \cref{lem:algo-DP}, \ref{lem:error-k-way}, and \ref{lem:FFT-speedup}, respectively.
\end{proof}

\subsection{Estimating Arbitrary Marginal Query Workloads}
\label{sec:upper-bound-arbitrary-marginal}

We now consider a general set up for answering arbitrary workloads of marginal queries. Recall that the workload $Q_{\SS}$ of marginal queries on the universe $\uni := \uni_1 \times \ldots \times \uni_d$, where $\uni_i := \{0, \ldots, m_i - 1\}$, is defined by a collection $\SS$ of subsets of $[d]$. For each subset $S \in \SS$, the workload contains all queries $q_{S,t}$ for all $t \in \uni_S := \prod_{i \in S}\uni_i$. 

We first consider an error metric which is a weighted version of root mean squared error. In this case, together with the workload $Q_{\SS}$, we are also given a weight function $p:\SS \to \R_{\ge 0}$. Although it is not essential to our algorithm, we will assume the weights are normalized, i.e., $\sum_{S \in \SS} p(S) = 1$. Then the goal is to compute estimates $\tilde{q}_{S,t}(D)$ for all $S \in \SS$ and all $t \in \uni_S$ that minimize the error
\[
\err_p(q,\tilde{q}) := 
\left(\sum_{S \in \SS} \frac{p(S)}{|\uni_S|} \sum_{t \in \uni_S}\E[(\tilde{q}_{S,t}(D) - q_{S,t}(D))^2] \right)^{1/2}.
\]
Notice that, in this definition, the weight $p(S)$ of a set of
attributes $S$ is evenly split between the $\uni_S$ marginal queries
corresponding to $S$.
The error definition allows for zero-weight queries. 
This is required when we optimize for maximum variance later in the section.
Whenever a zero-weight query is a subset of a non-zero weight query we can release an unbiased estimate at no additional privacy cost.
This is the case e.g. when releasing all $k$-way marginals with Algorithm~\ref{alg:allkway}, as we can estimate all $1,2,\dots,(k-1)$-way marginals from the Fourier queries.
In the rest of the section we do not explicitly handle zero-weight queries that are not a subset of a non-zero weight query. 
We cannot release an unbiased estimate for these queries but since they do not affect the error measure it does not matter how they are estimated.

Our algorithm in this setting is an extension of
Algorithm~\ref{alg:allkway}. We answer all Fourier aggregate  queries $\fq_{a}$
necessary for reconstructing the answers to queries in $Q_{\SS}$. The
amount of noise added to $\fq_a$ is proportional to how much this noise contributes to $\err_p(q,\tilde{q})^2$ in expectation. 
The full description is given in Algorithm~\ref{alg:gen}. In \Cref{sec:lower-bound} we also show that this choice of noise magnitudes gives an optimal
factorization for $Q_{\SS}$.
Below we present some intuition behind the design of our algorithm, and then present the proofs of the privacy and error guarantees.

\begin{algorithm}
\begin{algorithmic}


\State For any $a \in \uni$, let
\[
\tau_a :=
\sqrt{\sum_{\substack{S \in \SS\\S \supseteq \supp(a)}}\frac{p(S)}{|\uni_S|^2}},
\hspace{2em}
\text{and }
\hspace{2em}
\tau := \frac{1}{\mu^2} \sum_{a \in \uni}\tau_a.
\]

\State For any $a$ such that $\tau_a > 0$, privately release an estimate of $\fq_{a}(D)$ (see \cref{eq:fourier-query}) by adding independent noise such that
    \[
        \tilde \fq_{a}(D) =\fq_{a}(D) + Z_{a}, \mathrm{~where~}Z_{a} \sim \mathcal{CN}\left(0, \frac{2\tau}{\tau_{a}}\right).
    \]
    
\State Estimates for marginal queries for $S \in \SS$ can be recovered using post-processing by computing
    \[
      \tilde{q}_{S,t}(D) = \mathrm{Re}\left( \frac{1}{ |\uni_S| }\sum_{\substack{a \in \uni\\ \supp(a) \subseteq S}} {\chi_a(t)}{\tilde \fq_a(D)}  \right) , \text{ where } \chi_a(t) = \prod_{i \in S} \omega_{m_i}^{a_i t_i}.
    \] 
\end{algorithmic}
\caption{Differentially private estimate of a workload $Q_{\SS}$ of weighted marginals.}
\label{alg:gen}
\end{algorithm}

It is easy to check that Algorithm~\ref{alg:allkway} is a special case of Algorithm~\ref{alg:gen} when $\SS$ consists of all subsets of $[d]$ of cardinality $k$, and $p(S) = \frac{1}{{d\choose k}}$ for all $S \in \SS$.

As for the intuition behind $\tau$, note that the error $\tilde{q}_{S,t}(D) - q_{S,t}(D)$ is a sum of $\vert \uni_S \vert$ random variables. 
Specifically, each $\fq_a$ where $\supp(a) \subseteq S$ contributes a random variable to the sum with variance 
\[
    \E\left[\left(\mathrm{Re}\left(\frac{1}{\vert \uni_S \vert} {\chi_a(t)}{(\tilde \fq_a(D) - \fq_a(D))} \right)\right)^2\right] = \E\left[\left(\mathrm{Re}\left(\frac{\tilde{\fq}_a(D) - \fq_a(D)}{\vert \uni_S \vert}\right)\right)^2\right] = \frac{\sigma_a^2}{\vert \uni_S \vert^2} ,
\]
where $\sigma_a^2$ is half the variance of $Z_a$. 
Since the noise variables added to all $\fq_a$ queries are independent, the variance of $\tilde{q}_{S,t}(D) - q_{S,t}(D)$ is simply the sum of variances for each of the $\vert \uni_S \vert$ random variables.
Likewise, the noise of the estimate for $\fq_a$ contributes to the error for all marginals where $\supp(a) \subseteq S$.
In total, the private estimate of the query $\fq_a(D)$ adds 
\[
    \sum_{\substack{S \in \SS\\\supp(a) \subseteq S}} \frac{p(S)}{\vert \uni_S \vert} \sum_{t \in \uni_S} \frac{\sigma_a^2}{|\uni_S|^2} 
    = \sigma_a^2 \sum_{\substack{S \in \SS\\\supp(a) \subseteq S}}\frac{p(S)}{|\uni_S|^2}
\]
to the expectation of $\err_p(q,\tilde{q})^2$.
As such, we set $\tau_a$ to the square root of the weight for $\tilde{\fq}_a(D)$. 
This leads to the optimal values for all $\sigma_a$ when minimizing $\err_p(q,\tilde{q})^2$.
See the work of Lebeda and Pagh for more details on calibrating Gaussian noise for $\ell_2^2$ error when queries have different scales~\cite[Chapter~4]{lebeda2023phd}.

We first give the privacy guarantees.

\begin{lemma}
    \label{lem:gen-DP}
    Algorithm~\ref{alg:gen} satisfies $\mu$-GDP.
\end{lemma}
\begin{proof}
The proof is analogous to that of Lemma~\ref{lem:algo-DP}.
We argue that releasing the Fourier query estimates $\tilde{\fq}_a(D)$ satisfies $\mu$-GDP, and then the privacy of the estimates $\tilde{q}_{S,t}(D)$ follows by post-processing. Once again, Lemma~\ref{lem:complex-to-real-reduction} gives us that $\tilde{\fq}_a(D)$ satisfies $\mu_a$-GDP, where 
\(
\mu_a^2 := \frac{\tau_{a}}{\tau}. 
\)
Now the lemma follows from composition and the definition of $\tau$. 
\end{proof}

Next we compute the error of the algorithm. 
Recall that 
we use the notation $\SS_{\downarrow}$ for the collection of all sets $R \subseteq [d]$ such that $R\subseteq S$ for some $S \in \SS$. I.e., this is the closure of $\SS$ under taking subsets.

\begin{lemma}
    \label{lem:error-gen}
    The estimates $\tilde{q}_{S,t}(D)$ computed by Algorithm~\ref{alg:gen} for $S\in \SS$ and $t \in \uni_S$ have weighted root mean squared error
    \begin{align*}
    \err_p(q,\tilde{q}) 
    &=\left(\sum_{S \in \SS} \frac{p(S)}{|\uni_S|} \sum_{t \in \uni_S}\E[(\tilde{q}_{S,t}(D) - q_{S,t}(D))^2] \right)^{1/2}\\
    &= \frac{1}{\mu}\sum_{R \in \SS_{\downarrow}}\left(\prod_{j \in R}(m_j-1)\right)\sqrt{\sum_{\substack{S \in \SS\\S \supseteq R}}\frac{p(S)}{|\uni_S|^2}}.
    \end{align*}
\end{lemma}
\begin{proof}
    Take some $S \in \SS$ for which $p(S) > 0$, and some $t \in \uni_S$.
    By equation~\eqref{eq:marginals-inverse-fourier}, 
    \[
      \tilde{q}_{S,t}(D) - {q}_{S,t}(D)
      = 
      \mathrm{Re}\left(\frac{1}{\vert \uni_S \vert} \sum_{\substack{a \in \uni\\ \supp(a) \subseteq S}} {\chi_a(t)}{Z_a}\right).
    \]
    By Lemma~\ref{lem:gauss-lincomb}, $\frac{1}{\vert \uni_S \vert} \sum_{a} {\chi_a(t)}{Z_a}$ 
    is a complex mean zero normal random variable, and its real part is also normally distributed with half the variance. Again
    by Lemma~\ref{lem:gauss-lincomb}, and since $|\chi_a(t)| = 1$, we get that $\tilde{q}_{S,t}(D) - {q}_{S,t}(D)\sim \mathcal{N}(0,\sigma_S^2)$, where
    \begin{equation}\label{eq:gen-variance}
    \sigma_S^2 = 
    \frac{\tau}{|\uni_S|^2} \sum_{\substack{a \in \uni\\\supp(a)\subseteq S}} \frac{1}{\tau_{a}}
    = 
     \frac{\tau}{|\uni_S|^2} \sum_{\substack{a \in  \uni\\\supp(a)\subseteq S}} \left(\sum_{\substack{T \in  \SS\\\supp(a) \subseteq T }}\frac{p(T)}{|\uni_T|^2}\right)^{-1/2}.
    \end{equation}
    Here we used the observation that $\tau_a > 0$ for any $a$ such that $\supp(a) \subseteq S$, since $p(S)>0$ contributes a positive amount to $\tau_a$.
    Then the average squared error, weighted by $p$, is 
    \[
    \err_p(q,\tilde{q})^2 
    = 
    \sum_{S \in \SS} p(S) \sigma_S^2
    = 
    \tau \sum_{S \in \SS}\sum_{\substack{a \in \uni\\\supp(a)\subseteq S}} \frac{p(S)}{|\uni_S|^2} \left(\sum_{\substack{T \in \SS\\\supp(a) \subseteq T}}\frac{p(T)}{|\uni_T|^2}\right)^{-1/2}.
    \]
    Let $\SS_{p,\downarrow}$ 
    be the collection of all $R \subseteq [d]$ such that $R\subseteq S$ for some $S \in \SS$ with $p(S) > 0$.\footnote{It is natural to assume that $p(S) > 0$ for all $S \in \SS$, in which case $\SS_{p,\downarrow} = \SS_{\downarrow}$. We do not make this assumption here because later we optimize over choices of $p$ for a fixed $\SS$.}
    Changing the order of summation above, we get 
    \begin{align*}
    \err_p(q,\tilde{q})^2 
    &=
    \tau
    \sum_{\substack{a\in\uni\\\supp(a) \in \SS_{p,\downarrow}}}
    \sum_{\substack{S \in \SS\\S\supseteq \supp(a)}}
    \frac{p(S)}{|\uni_S|^2} \left(\sum_{\substack{T \in \SS\\T\supseteq\supp(a)}}\frac{p(T)}{|\uni_T|^2}\right)^{-1/2} \\
    &= \tau\sum_{\substack{a\in\uni\\\supp(a) \in \SS_{\downarrow}}}\sqrt{\sum_{\substack{S \in \SS\\S\supseteq\supp(a)}}\frac{p(S)}{|\uni_S|^2}}.
  \end{align*}
  After plugging in the value of $\tau$, and recalling that, for any $R$, there are $\prod_{j \in R}(m_j-1)$ choices of $a$ with $\supp(a) = R$, we have
  \begin{align*}
    \err_p(q,\tilde{q})^2
    &=
    \frac{1}{\mu^2}\left(\sum_{\substack{a\in\uni\\\supp(a) \in \SS_{p,\downarrow}}}\sqrt{\sum_{\substack{S \in \SS\\S\supseteq\supp(a)}}\frac{p(S)}{|\uni_S|^2}}\right)^2 \\
    &=
    \frac{1}{\mu^2}\left(\sum_{R \in \SS_{\downarrow}}\left(\prod_{j \in R}(m_j-1)\right)\sqrt{\sum_{\substack{S \in \SS\\S\supseteq R}}\frac{p(S)}{|\uni_S|^2}}\right)^2. \hspace{1em} \qedhere
  \end{align*}
\end{proof}

We now may conclude with the following theorem:
\begin{theorem}
    \label{thm:upper-bound-general-avg}
    Let $D$ be a dataset containing data points $x \in \uni$, where $\uni := \uni_1 \times \ldots \times \uni_d$, $\uni_i := \{0, \ldots, m_i - 1\}$. Let the workload $Q_{\SS}$ of marginal queries on the universe $\uni$ be defined by a collection $\SS$ of subsets of $[d]$ such that for each subset $S \in \SS$, the workload contains all queries $q_{S,t}$ for all $t \in \uni_S := \prod_{i \in S}\uni_i$. Given a weight function $p:\SS \to \R_{\ge 0}$,
    there exists a $\mu$-GDP mechanism that estimates all marginal queries in $Q_{\SS}$ and satisfies
     \begin{align*}
         \err_p(q,\tilde{q}) 
    &=\left(\sum_{S \in \SS} \frac{p(S)}{|\uni_S|} \sum_{t \in \uni_S}\E[(\tilde{q}_{S,t}(D) - q_{S,t}(D))^2] \right)^{1/2} \\
    &= \frac{1}{\mu}\sum_{R \in \SS_{\downarrow}}\left(\prod_{j \in R}(m_j-1)\right)\sqrt{\sum_{\substack{S \in \SS\\S \supseteq R}}\frac{p(S)}{|\uni_S|^2}}.
    \end{align*}
    Additionally, the noise for all marginal estimates can be sampled in time 
    \[
    O\left(\sum_{S\in \SS}\left(\prod _{i\in S}m_i \right)\log\left(\prod _{i\in S}m_i \right) + \gausstime \sum_{R\in \SS_{\downarrow}}\prod _{i\in R}(m_i - 1)\right) 
    \]
    where $O(\gausstime)$ is the time required for sampling a standard Gaussian. 
\end{theorem}
\begin{proof}
    The error bound is already proven in Lemma \ref{lem:error-gen}.
    The time complexity of sampling noise is straightforward with FFT, with the same argument as in Lemma \ref{lem:FFT-speedup}.
\end{proof}

Next we turn to the setting in which our measure of error is the maximum standard deviation of $\tilde{q}_{S,t}(D) - {q}_{S,t}(D)$ over all $S \in \SS$ and all $t \in \uni_S$. We reuse Algorithm~\ref{alg:gen}, with ``worst case'' choice of $p$, i.e., the $p$ that maximizes the error in Lemma~\ref{lem:error-gen}. The privacy of this algorithm follows directly from Lemma~\ref{lem:gen-DP}. We also have the following error bound.

\begin{lemma}\label{lem:error-gen-max}
  Let $\tilde{q}_{S,t}(D)$ be the estimates produced by Algorithm~\ref{alg:gen} with weights
  \begin{equation}\label{eq:pstar}
    p^* \in \arg \max\left\{\sum_{R \in \SS_{\downarrow}}\left(\prod_{j \in R}(m_j-1)\right)\sqrt{\sum_{\substack{S \in \SS\\S \supseteq R}}\frac{p(S)}{|\uni_S|^2}}: \sum_{S \in \SS}p(S) = 1, p(S) \ge 0\ \forall S \in \SS\right\}.
  \end{equation}
  Then, 
  \[
    \max_{S \in \SS, t \in \uni_S} (\E[(\tilde{q}_{S,t}(D)-q_{S,t}(D))^2])^{1/2} =
    \frac{1}{\mu}\sum_{R \in \SS_{\downarrow}}\left(\prod_{j \in R}(m_j-1)\right)\sqrt{\sum_{\substack{T\in \SS\\T\supseteq R}}\frac{p^*(T)}{|\uni_T|^2}}.
  \]
\end{lemma}

Lemma~\ref{lem:error-gen-max} follows easily from the following lemma, which gives first order optimality conditions for the optimization problem in \eqref{eq:pstar}.

\begin{lemma}\label{lem:pstar-optimality}
    The function $p^*:\SS\to\R_{\ge 0}$ defined in \eqref{eq:pstar} is such that 
    \(
    \sum_{\substack{S \in \SS\\S \supseteq R}}\frac{p(S)}{|\uni_S|^2} > 0
    \)
    for all $R \in \SS_{\downarrow}$. Moreover, for any $S\in \SS$, and any $S' \in \SS$ such that $p^*(S') > 0$, we have the inequality
    \[
    \frac{1}{|\uni_S|^2} \sum_{\substack{a\in \uni \\\supp(a) \subseteq S}} \left(\sum_{\substack{T \in \SS\\T \supseteq \supp(a)}}\frac{p^*(T)}{|\uni_T|^2}\right)^{-1/2}
    \le 
    \frac{1}{|\uni_{S'}|^2} \sum_{\substack{a\in \uni \\\supp(a) \subseteq S'}} \left(\sum_{\substack{T \in \SS\\T \supseteq \supp(a)}}\frac{p^*(T)}{|\uni_T|^2}\right)^{-1/2}.
    \]
\end{lemma}
\begin{proof}
  Let us denote the objective in the optimization problem in \eqref{eq:pstar} by
  \[
    f(p) := \sum_{R \in \SS_{\downarrow}} f_R(p) \enspace\text{ where } \enspace f_R(p) := \left(\prod_{j \in R}(m_j-1)\right)\sqrt{\sum_{\substack{S \in \SS\\S \supseteq R}}\frac{p(S)}{|\uni_S|^2}}.
  \]
  Let $p^*$ be as in \eqref{eq:pstar}. We first argue that 
  \(f_R(p^*) > 0\) for all $R \in \SS_{\downarrow}$, which is equivalent to the first claim. Assume this were not the case for some $R \in \SS_{\downarrow}$, and consider increasing $p^*(S)$ by a small amount $\beta$ for some $S \in \SS$ that contains $R$, and decreasing $p^*(S')$ by $\beta$ for some $S' \in \SS$ such that $p^*(S') > 0$. Increasing $p^*(S)$ increases $f_R(p^*)$, and thus also $f(p^*)$, by at least $c\sqrt{\beta}$ for a constant $c> 0$ independent of $\beta$ (but possibly dependent on $\SS$).
  At the same time, decreasing $p^*(S')$ only affects $f_R(p^*)$ when $R\subseteq S'$. For any such $R$, $p^*(S') > 0$ implies $f_R(p^*) > 0$, and, therefore, $\frac{\partial f_R(p^*)}{\partial p(S')}(p^*)$ 
  exists, and $f_R(p^*)$ decreases by at most $c'\beta$ for all small enough $\beta$ and some constant $c'$ independent of $\beta$. From this, we conclude that decreasing $p^*(S')$ decreases $f(p^*)$ by at most $c''\beta$ for another constant $c''$ and all small enough $\beta$. The overall change in $f(p^*)$ is thus at least $c\sqrt{\beta} - c''\beta > 0$ as long as $\beta$ is small enough, contradicting the optimality of $p^*$.

  Let us now choose any $S' \in \SS$ such that $p^*(S') > 0$. For any $S \in \SS$, first order optimality conditions imply that
  \begin{equation}\label{eq:partialder}
    \frac{\partial f}{\partial p(S)}(p^*) - \frac{\partial f}{\partial p(S')}(p^*) \le 0,
  \end{equation}
  where the fact that $f_R(p^*) > 0$ for all $R\in \SS_{\downarrow}$ guarantees that the partial derivatives exist at $p^*$.
  Indeed, if \eqref{eq:partialder} did not hold, then adding $\beta$ to $p^*(S)$ and subtracting $\beta$ from $p^*(S')$ for a small enough $\beta > 0$ would increase the value of $f$, contradicting the optimality of $p^*$. We calculate the partial derivatives as 
  \begin{align*}
    \frac{\partial f}{\partial p(S)}(p^*)
    &=
    \frac{1}{2|\uni_S|^2} \sum_{\substack{R \subseteq [d]\\R \subseteq S}} \left(\prod_{j \in R} (m_j-1)\right)\left(\sum_{\substack{T \in \SS\\T \supseteq R}}\frac{p^*(T)}{|\uni_T|^2}\right)^{-1/2}\\
    &=
      \frac{1}{2|\uni_S|^2} \sum_{\substack{a\in \uni \\\supp(a) \subseteq S}} \left(\sum_{\substack{T \in \SS\\T \supseteq \supp(a)}}\frac{p^*(T)}{|\uni_T|^2}\right)^{-1/2},
  \end{align*}
  where the second equality holds because the number of choices of $a$ with $\supp(a) = R$ equals $\prod_{j \in R} (m_j-1)$. Plugging back into \eqref{eq:partialder} and multiplying both sides by $2$ completes the proof.
\end{proof}

\begin{proof}[Proof of Lemma~\ref{lem:error-gen-max}]
   Recalling \eqref{eq:gen-variance}, Lemma~\ref{lem:pstar-optimality} implies that $\sigma^2_S \le \sigma^2_{S'}$ for any $S \in \SS$, and any $S' \in \SS$ such that $p(S') > 0$, where $\sigma_S^2 = \E[(\tilde{q}_{S,t}(D)-q_{S,t}(D))^2]$ for any $t \in \uni_S$. Therefore,
  \begin{align*}
    \err_{p^*}(q,\tilde{q})^2
    &=
      \sum_{S' \in \SS} \frac{p^*(S')}{|\uni_{S'}|} \sum_{t \in \uni_{S'}}\E[(\tilde{q}_{S',t}(D) - q_{S',t}(D))^2]
    = \sum_{S' \in \SS} p^*(S')\sigma_{S'}^2 = \max_{S \in \SS} \sigma^2_S.
  \end{align*}
  The lemma now follows from Lemma~\ref{lem:error-gen}.
\end{proof}

The following theorem is our main result for maximum error when releasing answers to general workloads of marginal queries. 

\begin{theorem}
    \label{thm:upper-bound-general}
    Let $D$ be a dataset containing data points $x \in \uni$, where $\uni := \uni_1 \times \ldots \times \uni_d$, $\uni_i := \{0, \ldots, m_i - 1\}$. Let workload $Q_{\SS}$ of marginal queries on universe $\uni$ be defined by a collection $\SS$ of subsets of $[d]$ such that for each subset $S \in \SS$, the workload contains all queries $q_{S,t}$ for all $t \in \uni_S := \prod_{i \in S}\uni_i$. Let $p^*:\SS \to \R_{\ge 0}$ be as defined in \eqref{eq:pstar}.
    There exists a $\mu$-GDP mechanism that estimates all marginal queries in $Q_{\SS}$ and satisfies
     \[
    \max_{S \in \SS, t \in \uni_S} (\E[(\tilde{q}_{S,t}(D)-q_{S,t}(D))^2])^{1/2} =
    \frac{1}{\mu}\sum_{R \in \SS_{\downarrow}}\left(\prod_{j \in R}(m_j-1)\right)\sqrt{\sum_{\substack{T\in \SS\\T\supseteq R}}\frac{p^*(T)}{|\uni_T|^2}}.
  \]
    Additionally, given the optimal $p^*$, which can be computed in time polynomial in $|\SS_{\downarrow}| + \max_{S \in \SS}|S|$, the noise for all marginal estimates can be sampled in time 
    \[
    O\left(\sum_{S\in \SS}\left(\prod _{i\in S}m_i \right)\log\left(\prod _{i\in S}m_i \right) + \gausstime \sum_{R\in \SS_{\downarrow}}\prod _{i\in R}(m_i - 1)\right) 
    \]
    where $O(\gausstime)$ is the time required for sampling a standard Gaussian. 
\end{theorem}
\begin{proof}
    The weight function $p^*$ can be computed is the solution to a concave maximization problem in $|\SS|$ variables, and the objective and its gradients can be computed to arbitrary precision in time polynomial in $\abs{\SS_{\downarrow}} + \max_{S \in \SS} |S|$ (see, e.g., \cite{Bubeck-convopt}). Therefore, $p^*$ can be computed within polynomial time in $\abs{\SS_{\downarrow}} + \max_{S \in \SS} |S|$. 
    The error bound was proven in  Lemma \ref{lem:error-gen-max}. The time complexity of sampling noise is straightforward with FFT, with the same argument as in Lemma \ref{lem:FFT-speedup}.
\end{proof}

Assuming, without loss of generality that $m_i > 1$ for each $i\in [d]$, one can check that 
\[
|\SS_{\downarrow}| \le \sum_{R\in \SS_{\downarrow}}\prod _{i\in R}(m_i - 1)
\le \sum_{S \in \SS} |\uni_S|
= |Q_{\SS}|,
\]
and the running time bound above is certainly polynomial in the number of queries.

It is worth noting also that, since the noise added to each query is normally distributed, standard concentration bounds show that 
\[
\E\left[\max_{S \in \SS, t \in \uni_S} |\tilde{q}_{S,t}(D)-q_{S,t}(D)|\right] \le
\sqrt{2 \ln |Q_{\SS}|} \, \max_{S \in \SS, t \in \uni_S} (\E[(\tilde{q}_{S,t}(D)-q_{S,t}(D))^2])^{1/2}.
\]

\begin{remark}\label{rem:ellp-algo}
    One may also consider error measures that interpolate between weighted root mean squared error and maximum variance, as was done in~\cite{NikolovT24,LiuUZ24}. For example, we can consider, for an exponent $r\ge 1$, the error measure 
    \[
    \E\left[\sum_{S \in \SS, t \in \uni_S} |\tilde{q}_{S,t}(D)-q_{S,t}(D)|^{2r}\right]^{1/2r}
    \lesssim
    \sqrt{r} \left(\sum_{S \in \SS, t \in \uni_S} \E[|\tilde{q}_{S,t}(D)-q_{S,t}(D)|^2]^{r}\right)^{1/2r},
    \]
    where the inequality follows from standard Gaussian moment bounds - see~\cite{NikolovT24,LiuUZ24} for the details. The quantity 
    \(
    \left(\sum_{S \in \SS, t \in \uni_S} \E[|\tilde{q}_{S,t}(D)-q_{S,t}(D))|^2]^{r}\right)^{1/r}
    \)
    is the $\ell_r$ norm of the vector of noise variances for each query. We can derive a mechanism minimizing this error measure analogously to the proof of Theorem~\ref{thm:upper-bound-general}. In particular, we can show an analogous result to Lemma~\ref{lem:error-gen-max}, with $p^*$ defined as
    \[
    p^* \in \arg \max\left\{\sum_{R \in \SS_{\downarrow}}\left(\prod_{j \in R}(m_j-1)\right)\sqrt{\sum_{\substack{S \in \SS\\S \supseteq R}}\frac{p(S)}{|\uni_S|^2}}: \sum_{S \in \SS}p(S)^{r'} \le 1, p(S) \ge 0\ \forall S \in \SS\right\},
    \]
    where $r' = \frac{r}{r-1}$. We can then show that, for this choice of $p^*$,
    \[
    \left(\sum_{S \in \SS, t \in \uni_S} \E[|\tilde{q}_{S,t}(D)-q_{S,t}(D)|^2]^{r}\right)^{1/2r} 
    = 
    \frac{1}{\mu} \sum_{R \in \SS_{\downarrow}}\left(\prod_{j \in R}(m_j-1)\right)\sqrt{\sum_{\substack{S \in \SS\\S \supseteq R}}\frac{p^*(S)}{|\uni_S|^2}}.
    \]
    More generally, we can extend this approach to any error measure which is a norm of the vector of noise variances by defining $p^*$ as the solution to an analogous optimization problem where the optimization is over the dual norm ball. We omit the details of these extensions from this paper, and focus on the two most commonly studied and natural error measures - the weighted root mean squared error, and the maximum variance measures. 
\end{remark}
\fi

\ifarxiv
\section{Upper Bounds for Product Queries and Extended Marginals}
\label{sec:upper-bound-product-and-extended}

In this section, we first introduce a generalization of marginal queries, which we call product queries. Our algorithms extend, with only small modifications, to answering arbitrary workloads of product queries. Later, in Section~\ref{sec:lower-bounds-product}, we also show that the resulting upper bounds are optimal within the class of factorization mechanisms. We further show how to use product queries to answer extended marginal queries, in which the data points have both categorical and numerical attributes, and the predicates on numerical attributes are threshold functions.

\subsection{Estimating Product Queries}\label{sec:weighted-product-queries}

 

Let us now extend our Fourier queries framework to product queries, as defined in \cref{sec:prelim-marginals}. For any $j \in [d]$, and $a \in \uni_j$, let 
\begin{equation}\label{eq:phi-hat}
    \widehat{\phi_j}(a) := \sum_{z = 0}^{m_j-1}\phi_j(z)\omega_{m_j}^{-a z}
\end{equation}
be the Fourier coefficient of $\phi_j$ corresponding to $a$. 
Then, by the inverse (discrete) Fourier transform, for any $t,z \in \uni_j$ we have
\begin{equation}\label{eq:phi-inv-fourier}
\phi_j((t-z)\bmod m_j) = \frac{1}{m_j}\sum_{a = 0}^{m_j-1}\widehat{\phi_j}(a)\omega_{m_j}^{a (t-z)} \,.
\end{equation}
Therefore, 
\begin{equation}\label{eq:prod-inv-x}
q^{\phi}_{S,t}(x) = \prod_{j \in S}\left(\frac{1}{m_j}\sum_{a_j = 0}^{m_j-1}\widehat{\phi_j}(a_j)\omega_{m_j}^{a_j\cdot (t_j-x_j)}\right)
= \frac{1}{|\uni_S|}\sum_{\substack{a \in \uni\\\supp(a) \subseteq S}}\left(\prod_{j\in S} \widehat{\phi_j}(a_j)\right)\chi_a(t)\overline{\chi_a(x)},
\end{equation}
and, summing over data points, 
\begin{equation}\label{eq:prod-inv-fourier}
    q^{\phi}_{S,t}(D) = \frac{1}{|\uni_S|}\sum_{\substack{a \in \uni\\\supp(a) \subseteq S}}\left(\prod_{j\in S} \widehat{\phi_j}(a_j)\right)\chi_a(t)F_a(D).
\end{equation}
This is a generalization of our formulas for marginals, since when $\phi_j(z) = \mathbbm{1}\{z = 0\}$, then $\widehat{\phi_j}(a) = 1$ for all $a\in \uni_j$.

\begin{algorithm}
\begin{algorithmic}


\State For any $a \in \uni$, define $\widehat{\phi_1}(a_1), \ldots, \widehat{\phi_d}(a_d)$ as in \cref{eq:phi-hat}, and let
\[
\tau_a :=
\sqrt{\sum_{\substack{S \in \SS\\S \supseteq \supp(a)}}\frac{p(S)\prod_{j \in S}|\widehat{\phi_j}(a_j)|^2}{|\uni_S|^2}},
\hspace{2em}
\text{and }
\hspace{2em}
\tau := \frac{1}{\mu^2} \sum_{a \in \uni}\tau_a.
\]

\State For any $a$ such that $\tau_a > 0$, privately release an estimate of $\fq_{a}(D)$ (see \cref{eq:fourier-query}) by adding independent noise such that
    \[
        \tilde \fq_{a}(D) =\fq_{a}(D) + Z_{a}, \mathrm{~where~}Z_{a} \sim \mathcal{CN}\left(0, \frac{2\tau}{\tau_{a}}\right).
    \]
    
\State Estimates for product queries for $S \in \SS$ can be recovered using post-processing by computing
    \[
      \tilde{q}^\phi_{S,t}(D) = \mathrm{Re}\left( \frac{1}{ |\uni_S| }\sum_{\substack{a \in \uni\\ \supp(a) \subseteq S}} \left(\prod_{j \in S}\widehat{\phi_j}(a_j)\right){\chi_a(t)}{\tilde \fq_a(D)}  \right) , \text{ where } \chi_a(t) = \prod_{i \in S} \omega_{m_i}^{a_i t_i}.
    \] 
\end{algorithmic}
\caption{Algorithm for estimating a workload $Q^\phi_{\SS}$ of weighted product queries.}
\label{alg:gen-product}
\end{algorithm}

The privacy proof is analogous to the proof of Lemma~\ref{lem:gen-DP}. For the error analysis, observe that, reasoning as we did with Algorithm~\ref{alg:gen}, we have that for any $S$ and $t \in \uni_S$,  $\tilde{q}^\phi_{S,t} - q^\phi_{S,t}$ is distributed as a real Gaussian with mean $0$ and variance
\begin{equation*}
    \sigma_S^2 := \frac{\tau}{|\uni_S|^2}\sum_{\substack{a\in \uni\\\supp(a)\subseteq S}}\left(\prod_{j \in S}|\widehat{\phi_j}(a_j)|^2\right)\frac{1}{\tau_a}.
\end{equation*}
The rest of the error analysis is also analogous to the proof of Lemma~\ref{lem:error-gen}. 

To implement \cref{alg:gen-product} efficiently, we first compute, for each $j \in [d]$ that is contained in some $S \in \SS$, the Fourier coefficients $\widehat{\phi_j}(0), \ldots \widehat{\phi_j}(m_j-1)$ using standard (one-dimensional) FFT in time $O(m_j\log m_j)$. The total time to compute all the Fourier coefficients is then $O\left(\sum_{j: \{j\} \in \SS_\downarrow} m_j \log(m_j)\right)$. Then, we compute the noisy Fourier query answers $\tilde{F}_a(D)$ for those $a$ such that $\tau_a > 0$, which is a subset of those $a$ such that $\supp(a) \in \SS_{\downarrow}$. Finally, for any $S \in \SS$, to reconstruct $\tilde{q}^\phi_{S,t}$ for all $t \in \uni_S$, we apply the multi-dimensional FFT (\cref{lem:multidim-fft}) to the values $\prod_{j \in S}\widehat{\phi_j}(a_j){\tilde \fq_a(D)}$, allowing us to perform the reconstruction in time $O(|\uni_S|\log(|\uni_S|))$.

In summary, we have the following theorem, generalizing Theorem~\ref{thm:upper-bound-general-avg}.




\begin{theorem}
    \label{thm:upper-bound-product-avg}
    Let $D$ be a dataset containing data points $x \in \uni$, where $\uni := \uni_1 \times \ldots \times \uni_d$, $\uni_i := \{0, \ldots, m_i - 1\}$. For a choice of functions $\phi := (\phi_1, \ldots, \phi_d)$, $\phi_i:\uni_i \to \R$, let the workload $Q^\phi_{\SS}$ of product queries $q^\phi_{S,t}$ defined by a collection $\SS$ of subsets of $[d]$ be such that, for each subset $S \in \SS$, the workload contains all queries $q^\phi_{S,t}$ for all $t \in \uni_S := \prod_{i \in S}\uni_i$. Given a weight function $p:\SS \to \R_{\ge 0}$,
    there exists a $\mu$-GDP mechanism that estimates all queries in $Q^\phi_{\SS}$, and satisfies
    \begin{align*}
        \err_p(q,\tilde{q}) 
    &=\left(\sum_{S \in \SS} \frac{p(S)}{|\uni_S|} \sum_{t \in \uni_S}\E[(\tilde{q}^\phi_{S,t}(D) - q^\phi_{S,t}(D))^2] \right)^{1/2}
    \\
    &= \frac{1}{\mu}\sum_{\substack{a \in \uni\\\supp(a) \in \SS_{\downarrow}}}\sqrt{\sum_{\substack{S \in \SS\\S \supseteq \supp(a)}}\frac{p(S)\prod_{j \in S}|\widehat{\phi_j}(a_j)|^2}{|\uni_S|^2}}.
    \end{align*}
    Additionally, the noise for all product query estimates can be sampled in time 
    \[
    O\left(\sum_{S\in \SS}\left(\prod _{i\in S}m_i \right)\log\left(\prod _{i\in S}m_i \right) + \gausstime \sum_{R\in \SS_{\downarrow}}\prod _{i\in R}(m_i - 1)\right) 
    \]
    where $O(\gausstime)$ is the time required for sampling a standard Gaussian. 
\end{theorem}

We can extend this result to the maximum variance error measure, as we did for marginals. 

\begin{lemma}\label{lem:error-product-max}
  Let $\tilde{q}^\phi_{S,t}(D)$ be the estimates produced by Algorithm~\ref{alg:gen-product} with weights
  \begin{equation}\label{eq:pstar-prod}
    p^* \in \arg \max\left\{\sum_{\substack{a \in \uni\\\supp(a) \in \SS_{\downarrow}}}\sqrt{\sum_{\substack{S \in \SS\\S \supseteq \supp(a)}}\frac{p(S)\prod_{j \in S}|\widehat{\phi_j}(a_j)|^2}{|\uni_S|^2}}: \sum_{S \in \SS}p(S) = 1, p(S) \ge 0\ \forall S \in \SS\right\}.
  \end{equation}
  Then, 
  \[
    \max_{S \in \SS, t \in \uni_S} (\E[(\tilde{q}^\phi_{S,t}(D)-q^\phi_{S,t}(D))^2])^{1/2} =
    \frac{1}{\mu}\sum_{\substack{a \in \uni\\\supp(a) \in \SS_{\downarrow}}}\sqrt{\sum_{\substack{S \in \SS\\S \supseteq \supp(a)}}\frac{p^*(S)\prod_{j \in S}|\widehat{\phi_j}(a_j)|^2}{|\uni_S|^2}}.
  \]
\end{lemma}

The proof of Lemma~\ref{lem:error-product-max}, which is analogous to the proof of Lemma~\ref{lem:error-gen-max}, relies on the following lemma, whose proof is analogous to the proof of Lemma~\ref{lem:pstar-optimality}.

\begin{lemma}\label{lem:pstar-prod-optimality}
    The function $p^*:\SS\to\R_{\ge 0}$ defined in \eqref{eq:pstar-prod} is such that 
    \[
    \sum_{\substack{S \in \SS\\S \supseteq \supp(a)}}\frac{p(S)\prod_{j \in S}|\widehat{\phi_j}(a_j)|^2}{|\uni_S|^2} > 0
    \]
    for all $a\in \uni$ for which there is at least one $S \in \SS$ satisfying $S\supseteq \supp(a)$ and $\prod_{j \in S}|\widehat{\phi_j}(a_j)|^2 \neq 0$. Moreover, for any $S\in \SS$, and any $S' \in \SS$ such that $p^*(S') > 0$, we have the inequality
    \begin{multline*}
    \frac{1}{|\uni_S|^2} \sum_{\substack{a\in \uni \\\supp(a) \subseteq S}} \left(\sum_{\substack{T \in \SS\\T \supseteq \supp(a)}}\frac{p^*(T)\prod_{j \in T}|\widehat{\phi_j}(a_j)|^2}{|\uni_T|^2}\right)^{-1/2}\\
    \le 
    \frac{1}{|\uni_{S'}|^2} \sum_{\substack{a\in \uni \\\supp(a) \subseteq S'}} \left(\sum_{\substack{T \in \SS\\T \supseteq \supp(a)}}\frac{p^*(T)\prod_{j \in T}|\widehat{\phi_j}(a_j)|^2}{|\uni_T|^2}\right)^{-1/2}.
    \end{multline*}
\end{lemma}

The next Theorem follows from Lemma~\ref{lem:error-product-max}. 

\begin{theorem}
    \label{thm:upper-bound-product}
    Let $D$ be a dataset containing data points $x \in \uni$, where $\uni := \uni_1 \times \ldots \times \uni_d$, $\uni_i := \{0, \ldots, m_i - 1\}$. For a choice of functions $\phi := (\phi_1, \ldots, \phi_d)$, $\phi_i:\uni_i \to \R$, let the workload $Q^\phi_{\SS}$ of product queries $q^\phi_{S,t}$ defined by a collection $\SS$ of subsets of $[d]$ be such that, for each subset $S \in \SS$, the workload contains all queries $q^\phi_{S,t}$ for all $t \in \uni_S := \prod_{i \in S}\uni_i$. Let $p^*:\SS\to \R_{\ge 0}$ be the function defined by \eqref{eq:pstar-prod}.
    There exists a $\mu$-GDP mechanism that estimates all product queries in $Q^\phi_{\SS}$ and satisfies
     \[
    \max_{S \in \SS, t \in \uni_S} (\E[(\tilde{q}^\phi_{S,t}(D)-q^\phi_{S,t}(D))^2])^{1/2} =
    \frac{1}{\mu}\sum_{\substack{a \in \uni\\\supp(a) \in \SS_{\downarrow}}}\sqrt{\sum_{\substack{S \in \SS\\S \supseteq \supp(a)}}\frac{p^*(S)\prod_{j \in S}|\widehat{\phi_j}(a_j)|^2}{|\uni_S|^2}}.
  \]
    Additionally, given the optimal $p^*$, which can be computed in time polynomial in $|\SS_{\downarrow}| + \max_{S \in \SS} |S|$, the noise for all product query estimates can be sampled in time 
    \[
    O\left(\sum_{S\in \SS}\left(\prod _{i\in S}m_i \right)\log\left(\prod _{i\in S}m_i \right) + \gausstime \sum_{R\in \SS_{\downarrow}}\prod _{i\in R}(m_i - 1)\right) 
    \]
    where $O(\gausstime)$ is the time required for sampling a standard Gaussian. 
\end{theorem}

\begin{remark}
    The model can be generalized further by allowing the functions $(\phi_i)_{i \in S}$ to depend on the set of attributes $S$. We do not pursue this here, mostly to avoid introducing more complex notation.
\end{remark}

\subsection{Estimating Workloads of Extended Marginals}\label{sec:product-apps}


In this subsection we consider extended marginal queries, which we defined in \cref{sec:prelim-marginals}, and we recall the definition here. The set of attributes $[d]$ is partitioned into subsets $C$ (the categorical attributes) and $N$ (the numerical ones). As with standard marginal queries, the domain of attribute $i$ is $\uni_i$, and the universe is $\uni:= \uni_1 \times \ldots\times \uni_d$. 
We allow prefix and suffix queries on the numerical attributes. To encode them, we allow, for any numerical attribute $i \in N$, $t_i$ to be positive or negative, with positive values denoting prefix predicates, and negative values denoting suffix predicates. We thus define $T_i:= \uni_i$ for $i \in C$, and  $T_i := \{-m_i, \ldots, 0, \ldots, m_i-1\}$ 
for $i\in N$, and $T_S := \prod_{i \in S}T_i$.
We now redefine the query $q_{S,t}$ for $S \subseteq [d]$ and $t \in T_S$ as 
\begin{equation}\label{eq:marg-prefix}
    q_{S,t}(x) := \left(\prod_{j \in S \cap C} \mathbbm{1}\{x_j = t_j\}\right)\left(\prod_{\substack{j \in S \cap N\\t_j \ge 0}}\mathbbm{1}\{x_j \le t_j\}\right)\left(\prod_{\substack{j \in S \cap N\\t_j < 0}}\mathbbm{1}\{x_j \ge |t_j|\}\right).
\end{equation}
In the case when $C = [d]$, these are just marginal queries. On the other hand, if $N = [d]$, these are multi-dimensional (prefix and suffix) range queries. 
Note that we do not explicitly support the $\mathbbm{1}\{x_j \ge 0\}$ suffix query because it is equivalent to a $\mathbbm{1}\{x_j \le m_j - 1\}$ prefix query. 
Given a collection of subsets $\SS$ of $[d]$, let us define $Q^{\text{pr-suf}}_\SS$ to be the workload of all queries $q_{S,t}$ for $S \in \SS$ and $t \in T_S$. We also define a variant which contains only prefix predicates for numerical attributes: $Q^{\text{pref}}_{\SS}$ contains all queries $q_{S,t}$ for $S\in \SS$ and $t \in \uni_S$. 

\begin{remark}
When $t_i = m_i-1$, $\mathbbm{1}\{x_i \le t_i\}$ always evaluates to $1$, and when $t_i = -m_i$, $\mathbbm{1}\{x_i \ge |t_i|\}$ always evaluates to $0$. It is possible to achieve slightly stronger results by removing these values. Then the query with $t_i = m_i-1$ can then be recovered from other extended marginal queries that omit the attribute $i$. 
\end{remark}

Next, we show how to ``embed'' extended marginal queries into the product queries studied in the previous subsection. This idea was already used for prefix queries by Choquette-Choo et al.~on factorization mechanisms for optimization~\cite{Choquette-ChooM23}, and is implicit in the group algebra based factorization mechanism of Henzinger and Upadhyay~\cite{HenzingerU25} (see also Section 5.1 of~\cite{henzingerKU2026normalized}).

\begin{lemma}\label{lem:extended-marginals-embedding}
    Given a partition of $[d]$ into numerical attributes $N$ and categorical attributes $C$, and a corresponding universe $\uni = \uni_1\times \ldots \times \uni_d$, define $\uni' = \uni'_1\times \ldots \times \uni'_d$, where $\uni'_i = \uni_i$ for $i \in C$, and $\uni'_i = \{0, \ldots, 2m_i-1\}$ for $i \in N$. There exist functions $\phi_1, \ldots, \phi_d$, $\phi_i:\uni'_i \to \{0,1\}$, such that for any $S\subseteq [d]$ and any $t \in T_S$, there exists a $t' \in \uni'_S$ for which $q_{S,t}(x) = q^\phi_{S,t'}(x)$ for all $x \in \uni$. Moreover, the mapping from $t$ to $t'$ is a bijection.
\end{lemma}

\begin{proof}
    We define $\phi_i(z) := \mathbbm{1}\{z = 0\}$ for $i \in C$, and $\phi_i(z):= \mathbbm{1}\{z \le m_i-1\}$ for $i \in N$. 
    Notice that, for $t,z \in \{0, \ldots, 2m_i-1\}$, 
    \begin{equation}\label{eq:t-z-mod}
    (t - z) \bmod 2m_i =
    \begin{cases}
        t-z & z \le t\\
        2m_i + (t-z) & z > t
    \end{cases}.
    \end{equation}
    Suppose that $z,t \in \{0, \ldots, m_i - 1\}$. When $z \le t$, clearly $t-z \le m_i-1$. When $z>t$, since $z-t \le m_i-1$, we get $2m_i + (t-z) > m_i-1$. Therefore, for $z,t \in \{0, \ldots, m_i - 1\}$, $\phi_i(t-z \bmod 2m_i) = \mathbbm{1}\{z \le t\}$. 

    Now consider $z \in \{0, \ldots, m_i-1\}$ but $t \in \{m_i, \ldots, 2m_i-1\}$. In this case, $z\le t$ always holds, so $(t - z) \bmod 2m_i = t-z$. Then $t-z \le m_i -1$ if and only if $z \ge t-m_i+1$. Therefore, for $z \in \{0, \ldots, m_i - 1\}$, and $t \in \{-m_i, \ldots, -1\}$, $\phi_i(t'-z \bmod 2m_i) = \mathbbm{1}\{z \ge |t|\}$ where $t':= |t| + m_i-1 = m_i -1 -t$. 
    
    This means that, when $x \in \uni \subseteq \uni'$, $S\subseteq [d]$, and $t \in T_S$, $q^{\phi}_{S,t'}(x) = q_{S,t}(x)$, where
    \[
    t'_i := 
    \begin{cases}
        t_i & i \in C, \text{ or } i\in N, t_i \ge 0\\
        m_i - 1 -t_i & i \in N \text{ and } t_i < 0
    \end{cases}.
    \]
    It is easy to check that this mapping between $t$ and $t'$ is a bijection. 
\end{proof}

We can then use Algorithm~\ref{alg:gen-product} to answer the extended marginal queries $Q_{\SS}$. Doing so gives us the guarantees in the following theorem.
In \Cref{sec:lower-bound-extended-marginals} we give a lower bound for $Q^{\text{pref}}_\SS$. 
Note that unlike our other results our bounds for extended marginals are ``only'' nearly optimal.
The gap occurs because our product query embedding is slightly more expressive than $Q^{\text{pref}}_\SS$.
Nevertheless, no optimal construction is known even for a single prefix query despite significant attention from the research community~(see \cite{henzingerKU2026normalized} for the current best known bounds).


\begin{theorem}
    \label{thm:upper-bound-pref-avg}
    Let $D$ be a dataset containing data points $x \in \uni$, where $\uni := \uni_1 \times \ldots \times \uni_d$, $\uni_i := \{0, \ldots, m_i - 1\}$, and the attributes $[d]$ are partitioned into categorical attributes $C$, and numerical attributes $N$. For any collection $\SS$ of subsets of $[d]$, the workload $Q^{\text{pr-suf}}_\SS$ defined above, and a weight function $p:\SS \to \R_{\ge 0}$, there exists a $\mu$-GDP mechanism that estimates all queries in $Q^{\text{pr-suf}}_\SS$ and satisfies 
    \begin{align*}
    \err_p(q,\tilde{q}) 
    &:=\left(\sum_{S \in \SS} \frac{p(S)}{|T_S|} \sum_{t \in T_S}\E[(\tilde{q}_{S,t}(D) - q_{S,t}(D))^2] \right)^{1/2}\\
    &= \frac{1}{\mu}\sum_{\substack{R\subseteq C\\O \subseteq N}}
    \left(\prod_{j \in R}(m_j-1)\right)\cdot\left(\prod_{j \in O}\eta(m_j)\right)
    \sqrt{\sum_{\substack{S \in \SS\\S\supseteq R \cup O}}\frac{p(S)}{|\uni_{S\cap C}|^2 \cdot 4^{|S\cap N|}}},
    \end{align*}
    where $\eta(m):= \frac{1}{m}\sum_{\ell=1}^m\frac{1}{\sin\left(\frac{\pi (2\ell-1)}{2m}\right)}.$ Furthermore, there is a weight function $p^*:\SS \to \R_{\ge 0}$ computable in time polynomial in $|{\SS}_{\downarrow}| +\max_{S \in \SS} |S|$, such that 
    \begin{multline*}
    \max_{S \in \SS,t \in T_S}\left(\E[(\tilde{q}_{S,t}(D) - q_{S,t}(D))^2] \right)^{1/2}\\
    = \frac{1}{\mu}\sum_{\substack{R\subseteq C\\O \subseteq N}}
    \left(\prod_{j \in R}(m_j-1)\right)\cdot\left(\prod_{j \in O}\eta(m_j)\right)
    \sqrt{\sum_{\substack{S \in \SS\\S\supseteq R \cup O}}\frac{p^*(S)}{|\uni_{S\cap C}|^2 \cdot 4^{|S\cap N|}}}.
    \end{multline*}

    Additionally, the noise for all extended marginal estimates can be sampled in time 
    \[
    O\left(\sum_{S\in \SS}\left(\prod _{i\in S}m'_i \right)\log\left(\prod _{i\in S}m'_i \right) + \gausstime \sum_{R\in \SS_{\downarrow}}\prod _{i\in R}(m_i - 1)\right) 
    \]
    where $O(\gausstime)$ is the time required for sampling a standard Gaussian, and $m'_i = m_i$ if $i \in C$ and $m'_i = 2m_i$ if $i \in N$

    The same guarantees hold when $Q^{\text{pr-suf}}_{\SS}$ is replaced by $Q^{\text{pref}}_{\SS}$.
\end{theorem}
\begin{proof}
The proof is the same for $Q^{\text{pr-suf}}_{\SS}$ and $Q^{\text{pref}}_{\SS}$, since the noise variance $\E[(\tilde{q}_{S,t}(D) - q_{S,t}(D))^2]$ is independent of $t$ for each $S$. Moreover, the proofs for weighted mean squared error and maximum variance are the same, and only differ in whether we use Theorem~\ref{thm:upper-bound-product-avg} or Theorem~\ref{thm:upper-bound-product}. We present the proof for weighted root mean squared error.

The privacy and running time guarantees follow directly from Theorems~\ref{thm:upper-bound-product-avg}~and~\ref{thm:upper-bound-product}. For error, we get that for any weight function $p:\SS\to \R_{\ge 0},$ the weighted root mean squared error is
\[
    \frac{1}{\mu}\sum_{\substack{a \in \uni'\\\supp(a) \in \SS_{\downarrow}}}\sqrt{\sum_{\substack{S \in \SS\\S \supseteq \supp(a)}}\frac{p(S)\prod_{j \in S}|\widehat{\phi_j}(a_j)|^2}{|T_S|^2}}.
\]
To evaluate this expression, we need to compute $\widehat{\phi_j}(a_j)$. When $j \in C$, we have $\widehat{\phi_j}(a) = 1$. When $j \in N$, we have $\widehat{\phi_j}(0) = m_j$, and, for $a_j > 0$, 
\[
\widehat{\phi_j}(a_j) = \sum_{z = 0}^{m_j-1} \omega_{2m_j}^{-a_j\cdot z}
= \frac{1-\omega_{2m_j}^{-a_j\cdot m_j}}{1-\omega_{2m_j}^{-a_j}}; 
\hspace{2em}
|\widehat{\phi_j}(a_j)|^2 = \frac{|1-\omega_{2m_j}^{-a_j\cdot m_j}|^2}{|1-\omega_{2m_j}^{-a_j}|^2} = \frac{1-\cos(\pi a_j)}{1-\cos\left(\frac{\pi a_j}{m_j}\right)}.
\]
Since $\omega_{2m_j}^{-a_j} \neq 1$ for $a_j > 0$, we could use the standard formula for a finite geometric series in the left equation. In the last equality we used the identity $\vert 1 - \exp(i\theta) \vert^2 = 2(1 - \cos(\theta))$. 
Note that when $a_j \neq 0$ is even, $\cos(\pi a_j) = 1$, and when $a_j$ is odd, $\cos(\pi a_j) = -1$, so 
\[
|\widehat{\phi_j}(a_j)|^2 = \frac{2\cdot \mathbbm{1}\{a_j \text{ is odd}\}}{1-\cos\left(\frac{\pi a_j}{m_j}\right)}
= \frac{\mathbbm{1}\{a_j \text{ is odd}\}}{\sin\left(\frac{\pi a_j}{2m_j}\right)^2}.
\]  
Let $\mathrm{odd}(a):= \{j \in N:  a_j \text{ is odd}\}$, and $\mathrm{null}(a):= \{j \in N: a_j = 0\} = N\setminus \supp(a)$. If $\prod_{j \in S}|\widehat{\phi_j}(a_j)|^2 = \prod_{j \in S\cap N}|\widehat{\phi_j}(a_j)|^2\neq 0$, then it must be the case that, for any $j \in S\cap N$, $a_j$ is either odd or $0$, i.e., $S\cap N \subseteq \mathrm{odd}(a)\cup \mathrm{null}(a)$. Notice that this contradicts $S\supseteq \supp(a)$ unless $\supp(a)\cap N = \mathrm{odd}(a)$, since, otherwise, there would be some $j \in N$ such that $a_j$ is nonzero and even, and $S\supseteq \supp(a)$ would imply that this $j$ is also in $S$. 
Therefore, for any $S$ and $a$ such that $\supp(a) \cap N = \mathrm{odd}(a)$, and $S\supseteq \supp(a)$ (which implies $S \supseteq \mathrm{odd}(a)$), we have
\begin{align*}
    \prod_{j \in S}|\widehat{\phi_j}(a_j)|^2
=
    \prod_{j\in S\cap\mathrm{null}(a)} m_j^2 \prod_{j\in \mathrm{odd}(a)}\frac{1}{\sin\left(\frac{\pi a_j}{2m_j}\right)^2}
=
    |\uni_{S\cap\mathrm{null}(a)}|^2 \prod_{j\in \mathrm{odd}(a)}\frac{1}{\sin\left(\frac{\pi a_j}{2m_j}\right)^2}.
\end{align*}
If $\supp(a) \cap N \neq \mathrm{odd}(a)$, then $\prod_{j \in S}|\widehat{\phi_j}(a_j)|^2 = 0$ for all $S\supseteq \supp(a)$. 

We can now rewrite the error as
\begin{align*}
     \err_p(q,\tilde{q}) 
     & = \frac{1}{\mu}\sum_{\substack{a \in \uni'\\\supp(a) \in \SS_{\downarrow}}}\sqrt{\sum_{\substack{S \in \SS\\S \supseteq \supp(a)}}\frac{p(S)\prod_{j \in S}|\widehat{\phi_j}(a_j)|^2}{|T_S|^2}}
     \\
     & =
     \frac{1}{\mu}\sum_{\substack{a \in \uni'\\\supp(a) \cap N = \mathrm{odd}(a)}}\prod_{j \in \mathrm{odd}(a)}\frac{1}{\sin\left(\frac{\pi a_j}{2m_j}\right)}\sqrt{\sum_{\substack{S \in \SS\\S\supseteq \supp(a)}}\frac{p(S)|\uni_{S\cap\mathrm{null}(a)}|^2}{|T_S|^2}}.
 \end{align*}
 Here we used the observation that the sum equals $0$ for any $a$ where $\supp(a) \cap N \neq \mathrm{odd}(a)$ as discussed above.
 Using the observation that $|T_S| = |\uni_{S}|\cdot 2^{|S\cap N|}$, we can rewrite further as
 \begin{align}
     \err_p(q,\tilde{q}) 
     &=
     \frac{1}{\mu}\sum_{\substack{a \in \uni'\\\supp(a) \cap N = \mathrm{odd}(a)}}\prod_{j \in \mathrm{odd}(a)}\frac{1}{\sin\left(\frac{\pi a_j}{2m_j}\right)}\sqrt{\sum_{\substack{S \in \SS\\S\supseteq \supp(a)}}\frac{p(S)}{|\uni_{S\setminus \mathrm{null}(a)}|^2 \cdot 4^{|S\cap N|}}}\notag\\
     &= 
     \frac{1}{\mu}\sum_{\substack{a \in \uni'\\\supp(a) \cap N = \mathrm{odd}(a)}}\prod_{j \in \mathrm{odd}(a)}\frac{1}{m_j\sin\left(\frac{\pi a_j}{2m_j}\right)}\sqrt{\sum_{\substack{S \in \SS\\S\supseteq \supp(a)}}\frac{p(S)}{|\uni_{S\cap C}|^2 \cdot 4^{|S\cap N|}}},\label{eq:error-marg-pref}
 \end{align}
 where, in the second line, we used the observation that if $\supp(a) \cap N = \mathrm{odd}(a)$ and $S \supseteq \supp(a)$, then $S \setminus\mathrm{null}(a) = (S\cap C)\cup \mathrm{odd}(a)$, so $|\uni_{S\setminus\mathrm{null}(a)}| = |\uni_{S\cap C}||\uni_{\mathrm{odd}(a)}| = |\uni_{S\cap C}|\prod_{j \in \mathrm{odd}(a)}m_j$.
Note that in these formulas we use the convention $|\uni_{\emptyset}| = 1$. 

We can re-write the right hand side further by first choosing the sets $R := \supp(a) \cap C$, and $O := \mathrm{odd}(a)$. For any fixed such choice of $R\subseteq C$ and $O\subseteq N$, we have 
\begin{multline}\label{eq:error-marg-pref-sets}
\sum_{\substack{a \in \uni'\\\supp(a)\cap C = R\\ \supp(a) \cap N = \mathrm{odd}(a) = O}}\prod_{j \in O}\frac{1}{m_j\sin\left(\frac{\pi a_j}{2m_j}\right)}\\
= \left(\prod_{j \in R}(m_j-1)\right)\prod_{j \in O}\left(\frac{1}{m_j}\sum_{\ell = 1}^{m_j}\frac{1}{\sin\left(\frac{\pi (2\ell-1)}{2m_j}\right)}\right).    
\end{multline}
Recall that in the theorem statement we defined $\eta(m):= \frac{1}{m}\sum_{\ell = 1}^{m}\frac{1}{\sin\left(\frac{\pi (2\ell-1)}{2m}\right)}$. From \eqref{eq:error-marg-pref} and \eqref{eq:error-marg-pref-sets}, we have
\begin{equation*}
    \err_p(q,\tilde{q}) = 
    \frac{1}{\mu}\sum_{\substack{R\subseteq C\\O \subseteq N}}
    \left(\prod_{j \in R}(m_j-1)\right)\cdot\left(\prod_{j \in O}\eta(m_j)\right)
    \sqrt{\sum_{\substack{S \in \SS\\S\supseteq R \cup O}}\frac{p(S)}{|\uni_{S\cap C}|^2 \cdot 4^{|S\cap N|}}}.\qedhere
\end{equation*}
    
\end{proof}

In the case when $d := 1$, $m_1 := m$, $N := \{1\}$, and the only set in $\SS$ is $S:= \{1\}$, we simply have the workload of prefix and suffix queries on $\uni = \uni_1 = \{0,\ldots, m - 1\}$. Then Theorem~\ref{thm:upper-bound-pref-avg} gives us (for $p(S) = 1$)
\begin{equation*}
    \err_p(q,\tilde{q}) = 
    \frac{1}{2\mu}(1 + \eta(m))
     = \frac{1}{2\mu} + 
     \frac{1}{2\mu m}\sum_{\ell = 1}^{m}\frac{1}{\sin\left(\frac{\pi (2\ell-1)}{2m}\right)},
\end{equation*}
which recovers the explicit upper bound on the error of a factorization algorithm for this workload due to Henzinger and Upadhyay~\cite{HenzingerU25}. We remark that the function $\eta(m)$ is asymptotically $\eta(m) = \frac{1}{\pi}\ln(m) + O(1)$.
\ifarxiv
Our bound is a small additive constant worse than the best known explicit factorization~\cite{henzingerKU2026normalized} for a single prefix query. 
We note that the only known explicit factorization that outperforms ours on prefix queries is concurrent work \citep{henzingerKU2026normalized}, and our mechanism handles extended marginals which is a much more general class of queries.  
\fi

\fi

\ifarxiv
\section{Lower Bounds}
\label{sec:lower-bound}
  
In this section we present lower bounds for factorization mechanisms and marginal queries that show that the algorithms from Sections~\ref{sec:algorithm}~and~\ref{sec:weighted-product-queries} are optimal, and that the algorithm for extended marginals in Section~\ref{sec:product-apps} is optimal up to lower order terms.
First, we present properties of factorization mechanisms.
\ifpods
\else
Then, we show how our algorithms can be represented in the factorization framework.
\fi
Then, we present lower bounds for marginal queries that match our upper bounds. We also show how to extend these results to product queries, and, finally, we give a lower bound for extended marginals that nearly matches the factorization implied by  the algorithms in \cref{sec:product-apps}.

\subsection{Properties of Factorization Norms}\label{sec:fact-duality}


We start with some general properties of factorization norms. 
For a matrix $M$, let $\|M\|_{1\to 2}$ be the maximum $\ell_2$ norm of a column of $M$, and $\|M\|_{2\to\infty}$ the maximum $\ell_2$ norm of a row of $M$. Let $\|M\|_F$ be the Frobenius norm of $M$, i.e., $\|M\|_F = \sqrt{\tr(M^TM)}$ for real matrices, and $\sqrt{\tr(M^*M)}$ for complex matrices. Finally, let $\|M\|_{tr}$ be the trace norm of $M$, i.e., the sum of its singular values, and let $\|M\|_{op}$ be the operator norm of $M$, i.e., its largest singular value.

For a $K \times N$ real matrix $W$, we have the factorization norms
\begin{align*}
\gamma_F(W) &:= \inf\{\|L\|_F \|R\|_{1\to 2}: LR = W\};\\
\gamma_2(W) &:= \inf\{\|L\|_{2\to\infty} \|R\|_{1\to 2}: LR = W\}.
\end{align*}
Here the infima are over real matrices $L$ and $R$ for which $LR = W$.
The $\gamma_2$ norm is classical, while the definition of the $\gamma_F$ norm is due to Edmonds, Nikolov, and Ullman~\cite{EdmondsNU20}, except the definition in the latter paper is normalized differently. 


We deviate slightly from these definitions and allow the factors $L$ and $R$ to be complex matrices. This does not change the quantities we study, as explained in the following remark.



\begin{remark}\label{rem:complex-fact}
It is not hard to see that, when the matrix $W$ has real entries, allowing the matrices $L$ and $R$ to have complex entries does not change $\gamma_F(W)$ or $\gamma_2(W)$. It is obvious that allowing complex entries can only decrease the norms. In the other direction, take any complex factorization $LR = W$ where $L = L_{\R} + i L_{\C}$ and $R = R_{\R} + i R_{\C}$ for real matrices $L_{\R}, L_{\C}, R_{\R}, R_{\C}$. Since $\mathrm{Re}(LR) = W$, we have the real factorization
\[
\begin{pmatrix}
    L_{\R} & L_{\C}
\end{pmatrix}
\begin{pmatrix}
    R_{\R}\\
    -R_{\C}
\end{pmatrix} = W
\]
achieving the same value as the complex factorization. This observation was also made by Henzinger and Upadhyay~\cite{HenzingerU25}.
\end{remark}

We also note that, for any diagonal matrix $P$ with non-negative entries such that $\tr(P) = 1$, using the fact that
\[
  \|P^{1/2}L\|_{F}^2 = \sum_{i = 1}^K P_{i,i} \|\ell_i\|_2^2 \le \max_{i = 1}^K \|\ell_i\|_2^2 = \|L\|_{2\to\infty},
\]
we have that $\gamma_F(P^{1/2} W) \le \gamma_2(W)$. Above, we used $\ell_i$ for the $i$-th row of $L$. 

The next lemma gives a lower bound on $\gamma_F$, as well as a necessary and sufficient condition for a factorization to achieve the lower bound. 


\begin{lemma}\label{lm:gammaF-dual}
    For any $K\times N$ matrix $W$, and any $N\times N$ diagonal matrix $S$ with non-negative entries such that $\tr(S) = 1$, we have
    \[
    \gamma_F(W) \ge \|WS^{1/2}\|_{tr}.
    \]
    Moreover, for any factorization $W = LR$, we have $\|WS^{1/2}\|_{tr} = \gamma_F(W) = \|L\|_F\|R\|_{1\to 2}$ if and only if
    \begin{itemize}
        \item for any $j$ such that $S_{j,j} \neq 0$, the $\ell_2$ norm of the $j$-th column of $R$ equals $\|R\|_{1\to 2}$, and
        \item $L^*L = c\,RSR^*$ for some real number $c\ge 0$.
    \end{itemize}
\end{lemma}

The next lemma is the analogous result for $\gamma_2$.

\begin{lemma}\label{lm:gamma2-dual}
    For any $K\times N$ matrix $W$, and any diagonal matrices $P$ and $S$ with non-negative entries, respectively of dimensions $K\times K$ and $N\times N$, and such that $\tr(P) = \tr(S) = 1$, we have
    \[
    \gamma_2(W) \ge \|P^{1/2}WS^{1/2}\|_{tr}.
    \]
    Moreover, for any factorization $W = LR$, we have $\|P^{1/2} WS^{1/2}\|_{tr} = \gamma_2(W) = \|L\|_{2\to \infty}\|R\|_{1\to 2}$ if and only if
    \begin{itemize}
        \item for any $i$ such that $P_{i,i} \neq 0$, the $\ell_2$ norm of the $i$-th row of $L$ equals $\|L\|_{2\to \infty}$, and
        \item for any $j$ such that $S_{j,j} \neq 0$, the $\ell_2$ norm of the $j$-th column of $R$ equals $\|R\|_{1\to 2}$, and
        \item $L^*PL = c\,RSR^*$ for some real number $c\ge 0$.
    \end{itemize}
\end{lemma}

The lower bounds in Lemmas~\ref{lm:gammaF-dual}~and~\ref{lm:gamma2-dual} themselves are not new. A special case of the lower bound on $\gamma_F$, known as the singular value lower bound, was shown by Li and Miklau~\cite{LiM13,LiM15}, and the lower bound for $\gamma_2$ goes back to work of Haagerup in the 1980s~\cite{Haagerup85}, see also~\cite{LeeSS08}. Both lower bounds also follow from the duality results in~\cite{NikolovT24}. Moreover, these duality results show that, for any $W$, both lower bounds are achieved for some choice of diagonal matrices. We have, however, not found the characterization of the tight cases in Lemmas~\ref{lm:gammaF-dual}~and~\ref{lm:gamma2-dual} in the literature, which is why we include proofs of the lemmas in Appendix~\ref{app:duality}.


\subsection{Our Algorithms as Factorizations}
\label{sec:our-algs-are-factorizations}

Let us recall the factorization mechanisms framework. Suppose we are given a query workload $Q$ over a universe $\uni$: $Q$ is a set of queries, where each $q$ is a function from $\uni$ to $\R$, and induces a query on datasets $D = (x^{(1)}, \ldots, x^{(n)})$ given by $q(D) := \sum_{i=1}^n q(x^{(i)})$. We also use $Q(D)$ to denote the vector of query answers $(q(D))_{q \in Q}$. We represent $Q$ by a workload matrix $W \in \R^{Q \times \uni}$ with entries $W_{q,x} := q(x)$. We also represent the dataset $D$ by a histogram vector $h \in \mathbb{Z}^\uni$, where $h_x := |\{i: x^{(i)} = x\}|$ is the number of occurrences of $x$ in $D$. These definitions turn the workload into a linear function of the histogram: $Q(D) = Wh$. 

For a factorization $W = LR$ of workload matrix, we define a corresponding private factorization mechanism for answering queries in $Q$ as follows. We draw a normally distributed vector $Z \sim \mathcal{N}\left(0,\frac{\|R\|_{1\to 2}^2}{\mu^2}I_{r}\right)$, where $r$ is the number of rows of $R$, and output the vector of query answers $L(Rh + Z) = Wh +LZ = Q(D) + LZ$. Note that, using $r_x$ to denote the column of $R$ indexed by $x\in \uni$, $Rh = \sum_{i=1}^n r_{x^{(i)}}$ is simply another workload of $r$ queries, and they have sensitivity $\max_{x \in \uni} \|r_x\|_2 = \|R\|_{1\to 2}$. 
These are usually called the strategy queries, and we can view the factorization mechanism as answering a well chosen set of strategy queries and using them to reconstruct answers to the queries in $Q$. It is easy to verify that the mechanism satisfies $\mu$-GDP: outputting $Rh + Z$ satisfies $\mu$-GDP by Lemma~\ref{lem:gaussian}, and then multiplying by $L$ is just post-processing. Moreover, an easy calculation shows that the factorization mechanism achieves error bounds
\begin{align*}
\E[\|Wh - (L(Rh+Z))\|_2^2]^{1/2}
&= \E[\|LZ\|_2^2]^{1/2} = \frac1\mu \|L\|_{F}\|R\|_{1\to 2};\\
\max_{q \in Q}\E[\|Wh - (L(Rh+Z))_q\|_2^2]^{1/2}
&= \max_{q \in Q} \E[|(LZ)_q|^2]^{1/2} = \frac1\mu \|L\|_{2\to\infty}\|R\|_{1\to 2}.
\end{align*}
This motivates the definitions of the $\gamma_F(W)$ and $\gamma_2(W)$ norms: they correspond to the minimal error bounds achievable by a factorization mechanism. Note also that if we are interested in weighted root mean squared error, i.e., if each query $q \in Q$ is assigned non-negative weight $p(q)$, then we get the error bound
\[
\E\left[\sum_{q \in Q} p(q)( (Wh)_q - (L(Rh + Z))_q )^2\right]^{1/2}
= \E[\|P^{1/2} LZ\|_2^2]^{1/2} = \frac1\mu \|P^{1/2}L\|_{F}\|R\|_{1\to 2},
\]
where $P$ is the diagonal matrix indexed by queries and with diagonal entries $P_{q,q} := p(q)$. This error bound is optimized by the factorization that achieves $\gamma_F(P^{1/2}W)$.


Our algorithms fall in this framework: the Fourier queries defined in \eqref{eq:fourier-query} are the strategy queries, and the formula~\eqref{eq:marginals-inverse-fourier} shows how to reconstruct the query answers. We now describe this re-interpretation of the algorithm more precisely. We only do so for Algorithm~\ref{alg:gen}, since our other algorithms are special cases of it. 

Let us fix a set of marginals $\SS$ and a weight function $p:\SS\to \R_{\ge 0}$, as in Section~\ref{sec:upper-bound-arbitrary-marginal}. Let $W$, from now on, be the workload matrix $W$ of $Q_{\SS}$. Recall the definition of $\tau_a$ for $a \in \uni$ from Algorithm~\ref{alg:gen}:
\begin{equation}\label{eq:tau-def}
\tau_a :=
\sqrt{\sum_{\substack{S \in \SS\\S \supseteq \supp(a)}}\frac{p(S)}{|\uni_S|^2}} \,.
\end{equation}
Let $\tilde V$ be the complex matrix with columns indexed by those $a \in \uni$ for which $\tau_a > 0$, and rows indexed by $\uni$, and with entries 
\(
\tilde{V}_{x,a} = \chi_a(x).
\)
Next, let
$\tilde{U}$ be the matrix for which $\tilde{U}\tilde{V}^* = W$. The rows of $\tilde{U}$ are indexed by  $(S,t)$ where $S \in \SS$, and $t \in \uni_S$. The columns are indexed by the same set as the columns of $\tilde{v}$, i.e., by those $a \in \uni$ for which $\tau_a > 0$. The entries of $\tilde{U}$ can be deduced from \eqref{eq:marginals-inv-x}, and are
\[
\tilde{U}_{(S,t),a}=
\begin{cases}
    \frac{1}{|\uni_S|} \chi_a(t) & \supp(a)\subseteq S\\
    0 & \supp(a) \not\subseteq S
\end{cases}.
\]
Now $\tilde{U}\tilde{V}^* = W$ follows from \eqref{eq:marginals-inv-x}. This factorization, however, does not, in general, correspond to Algorithm~\ref{alg:gen} because the noise variance we use for different Fourier queries is different. Instead, let $L$ be a matrix of the same dimensions as $\tilde{U}$, and let $R$ be a matrix of the same dimensions as $\tilde{V}^*$, and define their entries by 
\begin{equation}\label{eq:factorization-gen}
L_{(S,t),a} := \tilde{U}_{(S,t),a}\sqrt{\frac{\sum_{b \in \uni} \tau_b}{\tau_a}}
\hspace{2em}
R_{a,x} := \sqrt{\frac{\tau_a}{\sum_{b \in \uni} \tau_b}} \overline{\tilde{V}_{x,a}}
= \sqrt{\frac{\tau_a}{\sum_{b \in \uni} \tau_b}} (\tilde{V}^*)_{a,x} \,.
\end{equation}
In other words, if we define a diagonal matrix $E$ with rows and columns indexed by $a\in \uni$ such that $\tau_a >0$, and let $E_{a,a} := \frac{\tau_a}{\sum_{b \in \uni} \tau_b}$, then $L := \tilde{U} E^{-1/2}$ and $R := E^{1/2} \tilde{V}^*$. Clearly $LR = \tilde{U}\tilde{V}^* = W$. 

To see the equivalence with Algorithm~\ref{alg:gen}, consider the mechanism that outputs $L(Rh + \tilde{Z})$ for $\tilde{Z} \sim \mathcal{CN}\bigl(0,\tfrac{2}{\mu^2}I\bigr)$.
We can think of this mechanism as first releasing $E^{1/2}(Rh + \tilde{Z}) = \tilde{V}^* h + E^{1/2} \tilde{Z}$, and then multiplying on the left by $\tilde{U}$. 
Now it is immediate that that $\tilde{V}^* h$ gives the answers to the Fourier queries $F_a(D)$, and that $(E^{1/2} \tilde{Z})_a$ is distributed as $Z_a$ in Algorithm~\ref{alg:gen}. Therefore $E^{1/2}(Rh + \tilde{Z})$ is distributed identically to the estimates $\tilde{F}_a$ in Algorithm~\ref{alg:gen}. Finally, multiplying this vector on the left by $\tilde{U}$ and taking real parts gives the estimates $\tilde{q}_{S,t}(D)$. Thus, Algorithm~\ref{alg:gen} corresponds to running a complex version of the factorization mechanism and taking real parts of the result. It is easy to see that this is also equivalent to making the factorization real (as described in Remark~\ref{rem:complex-fact}), and running the usual real-valued factorization mechanism.

The following lemma records, for future reference, the fact that the error bounds from Lemmas~\ref{lem:error-gen}~and~\ref{lem:error-gen-max} are achieved by factorization mechanisms. 

\begin{lemma}\label{lem:factorization-gen}
    For a workload of marginal queries $Q_{\SS}$ over a universe $\uni$ with workload matrix $W$, a weight function $p:\SS \to \R_{\ge 0}$, and a corresponding diagonal matrix $P$ indexed by queries $(S,t)$, $S\in \SS$, $t \in \uni_S$, with entries $P_{(S,t),(S,t)} = \frac{p(S)}{|\uni_S|}$, the factorization $LR = W$ in \eqref{eq:factorization-gen} satisfies
    \[
    \gamma_F(P^{1/2} W) \le \|P^{1/2} L\|_F \|R\|_{1\to 2} 
    = 
    \sum_{R \in \SS_{\downarrow}}\left(\prod_{j \in R}(m_j-1)\right)\sqrt{\sum_{\substack{T \in \SS\\T \supseteq R}}\frac{p(T)}{|\uni_T|^2}} \,.
    \]
    Moreover, for $p^*:\SS\to \R_{\ge 0}$ chosen as in \eqref{eq:pstar}, the factorization satisfies
    \[
    \gamma_2(W) \le \|L\|_{2\to\infty}\|R\|_{1\to 2} 
    = 
    \sum_{R \in \SS_{\downarrow}}\left(\prod_{j \in R}(m_j-1)\right)\sqrt{\sum_{\substack{T\in \SS\\T\supseteq R}}\frac{p^*(T)}{|\uni_T|^2}} \,,
    \]
    where we recall that $\SS_{\downarrow}$ is the collection of sets $R$ which are subsets of $S$ for some $S \in \SS$.
\end{lemma}
\begin{proof}
    The lemma follows from the calculations in Lemmas~\ref{lem:error-gen}~and~\ref{lem:error-gen-max}. First, observe that, since every entry of $\tilde{V}$ is a root of unity, so has absolute value $1$, the squared $\ell_2$ norm of every column of $R$ is
    \(
    \left(\sum_{a \in \uni} \tau_a\right)/\left({\sum_{b\in\uni} \tau_b}\right) = 1.
    \)
    Therefore, $\|R\|_{1\to 2} = 1$. To compute $\|P^{1/2}L\|_F$, observe that the squared $\ell_2$ norm of the row of $L$ indexed by $(S,t)$ is
    \[
    \frac{\sum_{b\in \uni}\tau_b}{|\uni_S|^2}\sum_{\substack{a \in  \uni\\\supp(a)\subseteq S}}\frac{1}{\tau_a}
    =\frac{\sum_{b\in \uni}\tau_b}{|\uni_S|^2} \sum_{\substack{a \in  \uni\\\supp(a)\subseteq S}} \left(\sum_{\substack{T \in  \SS\\T\supseteq \supp(a) }}\frac{p(T)}{|\uni_T|^2}\right)^{-1/2}.
    \]
    This is exactly $\mu^2\sigma_S^2$, for $\sigma_S^2$ as defined in \eqref{eq:gen-variance}. We now have $\|P^{1/2}L\|_F^2 = \mu^2\sum_{S \in \SS}p(S)\sigma_S^2$, and the result follows as in the proof of Lemma~\ref{lem:error-gen}.

    These observations also hold for $p^*:\SS\to \R_{\ge 0}$ chosen as in \eqref{eq:pstar}: the squared $\ell_2$ norm of the row of $L$ indexed by $(S,t)$ is $\mu^2\sigma_S^2$. Moreover, by Lemma~\ref{lem:pstar-optimality} we know that for any $S \in \SS$ such that $p^*(S) > 0$, we have $\sigma_S = \max_{T \in \SS} \sigma_T$. Therefore, $\|L\|_{2\to\infty} = \|(P^*)^{1/2} L\|_F$, where $P^*$ is defined as $P$ above but with $p^*$ in place of $p$. The result follows from the equation for $\|(P^*)^{1/2} L\|_F$ above.
\end{proof}


\begin{remark}
When we transform the matrices $L$ and $R$ in our factorization into real matrices as in Remark~\ref{rem:complex-fact}, we create some redundant rows of $R$ and columns of $L$. Indeed, for any $a \in \uni$, let $a'$ be defined by $a_j = m_j - a_j \bmod m_j$ for all $j \in [d]$. Then, for any $x \in \uni$, $\chi_a(x) = \overline{\chi_{a'}(x)}$. Writing $R = R_{\R} + i \R_{\C}$ as in Remark~\ref{rem:complex-fact}, we see that the rows of $R_{\R}$ corresponding to $a$ and $a'$ are the same, and the rows of $R_{\C}$ corresponding to $a$ and $a'$ are either both all-zero, or negations of each other. Similar observations can be made about $L_{\R}$ and $L_{\C}$. It is possible to remove these redundancies without changing the quality of the factorization, resulting in smaller matrices. We leave the details to the reader. 
\end{remark}

\subsection{Lower Bounds for Marginals}
\label{sec:lower-bounds-marginals}

Now we are ready prove that the upper bounds on $\gamma_F(P^{1/2}W)$ and $\gamma_2(W)$ in Lemma~\ref{lem:factorization-gen} in fact hold with equality, and, therefore, the factorization $LR = W$ in \eqref{eq:factorization-gen} is optimal. We use the same notation as the previous subsection. Our proof relies on Lemmas~\ref{lm:gammaF-dual}~and~\ref{lm:gamma2-dual}.

We will make extensive use of the following basic identity, which states the classical and easy to check fact that the Fourier basis is orthonormal (with the correct normalization). We have that, for any $a,b \in \uni$,
\begin{equation}\label{eq:Fourier-orth}
    \frac{1}{|\uni|}\sum_{x \in \uni} \chi_a(x) \overline{\chi_b(x)} = 
    \mathbbm{1}{\{a=b\}},
\end{equation}
where $\mathbbm{1}\{a=b\}$ is $1$ if $a=b$ and $0$ otherwise.

\begin{lemma}\label{lem:gammaF-opt}
    For any workload of marginal queries $Q_{\SS}$, and for any weights $p:\SS\to\R_{\ge 0}$, with $W$ and $P$ as in Lemma~\ref{lem:factorization-gen}, the factorization $LR = W$ for $L$ and $R$ defined in \eqref{eq:factorization-gen} satisfies $\gamma_F(P^{1/2}W) = \|P^{1/2}L\|_F \|R\|_{1\to 2}$.

    Similarly, for $p^*:\SS\to \R_{\ge 0}$ chosen as in \eqref{eq:pstar}, the factorization satisfies $\gamma_2(W)=\|L\|_{2\to \infty}\|R\|_{1\to 2}$.
\end{lemma}

\begin{proof}
    We will use Lemma~\ref{lm:gammaF-dual} with $S := \frac{1}{|\uni|}I$. In the proof of Lemma~\ref{lem:factorization-gen} we observed that all columns of $R$ have the same $\ell_2$ norm. It then suffices to check that $L^*PL = \frac{c}{|\uni|} RR^*$ for some real constant $c \ge 0$. 
    
    First we compute $RR^*$. Recall that $E$ is a diagonal matrix with diagonal entries $E_{a,a} := \frac{\tau_a}{\sum_{b \in \uni}\tau_b}$ for all $a \in \uni$ such that $\tau_a > 0$. For any $a$ and $b$ such that $\tau_a > 0$ and $\tau_b > 0$, we have
    \begin{align*}
      (RR^*)_{a,b} = (E^{1/2} \tilde{V}^* \tilde{V} E^{1/2})_{a,b}
      &= \sqrt{E_{a,a}E_{b,b}} (\tilde{V}^*\tilde{V})_{a,b}\\
      &= \sqrt{E_{a,a}E_{b,b}} \sum_{x\in \uni} \chi_b(x) \overline{\chi_a(x)}
      = |\uni|\sqrt{E_{a,a}E_{b,b}} \ \mathbbm{1}\{a=b\}.
    \end{align*}
    The last equality follows from \eqref{eq:Fourier-orth}.
    Therefore, $RR^* = |\uni| E$.

    Now we turn to computing $L^*PL$. For any $a$ and $b$ such that $\tau_a > 0$ and $\tau_b > 0$, we have
    \begin{align*}
      (L^*PL)_{a,b} = (E^{-1/2} \tilde{U}^* P\tilde{U} E^{-1/2})_{a,b}
      &= \frac{1}{\sqrt{E_{a,a}E_{b,b}}} (\tilde{U}^*P\tilde{U})_{a,b}\\
      &= \frac{1}{\sqrt{E_{a,a}E_{b,b}}} \sum_{\substack{S\in \SS\\S \supseteq \supp(a)\cup\supp(b)}}\frac{p(S)}{|\uni_S|^3}\sum_{t \in \uni_S} \chi_b(t) \overline{\chi_a(t)}\\
      &= \left(\frac{1}{\sqrt{E_{a,a}E_{b,b}}} \sum_{\substack{S\in \SS\\S \supseteq \supp(a)\cup\supp(b)}}\frac{p(S)}{|\uni_S|^2}\right)\cdot \mathbbm{1}\{a = b\} .
    \end{align*}
    The last equality again follows from \eqref{eq:Fourier-orth} after observing that, since any $S$ in the outer sum contains both $\supp(a)$ and $\supp(b)$, $\chi_a$ and $\chi_b$ can be interpreted as Fourier characters over $\uni_S$. We then have that $L^*PL$ is a diagonal matrix, and the diagonal entries are
    \begin{align*}
      (L^*PL)_{a,a}  = \frac{1}{E_{a,a}} \sum_{\substack{S\in \SS\\S \supseteq \supp(a)}}\frac{p(S)}{|\uni_S|^2}
      &= \frac{\sum_{b \in \uni} \tau_b}{\tau_a} \sum_{\substack{S\in \SS\\S \supseteq \supp(a)}}\frac{p(S)}{|\uni_S|^2}\\
      &= \left(\sum_{b \in \uni} \tau_b\right) \tau_a 
      = \left(\sum_{b \in \uni} \tau_b\right)^2 E_{a,a}.
    \end{align*}
    In the penultimate equality we used the definition of $\tau_a$ in \eqref{eq:tau-def}.
    In conclusion, \[L^*PL = \left(\sum_{b \in \uni} \tau_b\right)^2 E = \frac{\left(\sum_{b \in \uni} \tau_b\right)^2}{|\uni|} RR^*,\] as we needed to show.

    Let $P^*$ be the diagonal matrix $P$ but with $p^*$ in place of $p$. By Lemma~\ref{lm:gamma2-dual}, and the proof so far, it is enough to check that for any $S \in \SS$ such that $p^*(S) > 0$, and any $t \in \uni_S$, the $\ell_2$ norm of the row of $L$ indexed by $(S,t)$ equals $\|L\|_{2\to\infty}$, i.e., is maximal. The square of this $\ell_2$ norm equals
    \[
    \frac{\sum_{b \in \uni} \tau_b}{|\uni_S|^2} \sum_{\substack{a \in \uni\\\supp(a) \subseteq S}}\frac{1}{\tau_a}
    = \frac{\sum_{b \in \uni} \tau_b}{|\uni_S|^2} \sum_{\substack{a \in \uni\\\supp(a) \subseteq S}}\left(\sum_{\substack{T \in \SS\\T \supseteq \supp(a)}}\frac{p^*(T)}{|\uni_T|^2}\right)^{-1/2}.
    \]
    The maximality of the row norm is now implied by Lemma~\ref{lem:pstar-optimality}.
\end{proof}

\begin{remark}\label{rem:svd}
    The proof of Lemma~\ref{lem:gammaF-opt} gives a singular value decomposition of the matrix $P^{1/2}W$. The matrix of right singular vectors is $V := \frac{1}{\sqrt{|\uni|}}\tilde{V}$. Let $\Sigma$ be a diagonal matrix indexed by those $a$ for which $\tau_a > 0$, and with entries $\Sigma_{a,a} := \tau_a$. Then the matrix of left singular vectors is $U := \tilde{U} \Sigma^{-1}$, and the diagonal entries of $\Sigma$ are the singular values. The singular value decomposition is then $P^{1/2} W = U \Sigma V^*$. Our factorization results from splitting the singular value decomposition into a left matrix $U\Sigma^{1/2}$ and a right matrix $\Sigma^{1/2} V^*$, and normalizing the right matrix to have unit norm columns. 
\end{remark}

The next theorem is our main result giving optimal factorization norms for the workload matrix of a workload of weighted marginal queries. The theorem is a direct consequence of Lemmas~\ref{lem:factorization-gen}~and~\ref{lem:gammaF-opt}.

\begin{theorem}\label{thm:fact-main}
For a workload of marginal queries $Q_{\SS}$ over a universe $\uni$ with workload matrix $W$, a weight function $p:\SS \to \R_{\ge 0}$, and a corresponding diagonal matrix $P$ indexed by queries $(S,t)$, $S\in \SS$, $t \in \uni_S$, with entries $P_{(S,t),(S,t)} = \frac{p(S)}{|\uni_S|}$, the factorization $LR = W$ in \eqref{eq:factorization-gen} satisfies
    \begin{equation}\label{eq:gammaF-opt}
    \gamma_F(P^{1/2} W) = \|P^{1/2} L\|_F \|R\|_{1\to 2} 
    = 
    \sum_{R \in \SS_{\downarrow}}\left(\prod_{j \in R}(m_j-1)\right)\sqrt{\sum_{\substack{T \in \SS\\T \supseteq R}}\frac{p(T)}{|\uni_T|^2}}.
    \end{equation}
    Moreover, for $p^*:\SS\to \R_{\ge 0}$ chosen as in \eqref{eq:pstar}, the factorization satisfies
    \[
    \gamma_2(W) = \|L\|_{2\to\infty}\|R\|_{1\to 2} 
    = 
    \sum_{R \in \SS_{\downarrow}}\left(\prod_{j \in R}(m_j-1)\right)\sqrt{\sum_{\substack{T\in \SS\\T\supseteq R}}\frac{p^*(T)}{|\uni_T|^2}},
    \]
    where we recall that $\SS_{\downarrow}$ is the collection of sets $R$ which are subsets of $S$ for some $S \in \SS$.    
\end{theorem}

\begin{remark}
     Since in the proof of Lemma~\ref{lem:gammaF-opt} we showed that $L^*PL$ is a constant multiple of $RR^*$, Lemma~\ref{lm:gammaF-dual} implies that the right hand side of \eqref{eq:gammaF-opt} equals the singular value lower bound of $P^{1/2} W$, i.e., the value of $\frac{1}{\sqrt{|\uni|}} \|P^{1/2} W\|_{tr}$. This formula is similar but different from the one given in by McKenna et al.~\cite{McKenna_Miklau_Hay_Machanavajjhala_2023}. We explain the discrepancy in Appendix~\ref{app:duality}.
\end{remark}

Finally, we record as a corollary the optimal factorization norms for the workload of all $k$-way marginal queries when $m_j = m$ for all $j \in [d]$. This is the special cases considered in Algorithm~\ref{alg:allkway}, and the corollary shows the algorithm is optimal among factorization mechanisms for this workload.

\begin{corollary}\label{cor:allkway}
    Let $\SS$ be the set of all subsets of $[d]$ of size $k$, and let $\uni$ be a universe such that $|\uni_j| = m_j = m$ for all $j \in [d].$ Let $W$ be the workload matrix of $Q_{\SS}$. The factorization $LR = W$ in \eqref{eq:factorization-gen} with $p(S) = 1/{d\choose k}$ for all $S \in \SS$ satisfies
    \[
    \frac{1}{\sqrt{m^{k}{d\choose k}}}\gamma_F(W) = \gamma_2(W) = \|L\|_{2\to\infty} \|R\|_{1\to 2} 
    = 
    \frac{1}{m^{k} \sqrt{{d\choose k}} }\sum_{\ell = 0}^k {d \choose \ell} (m - 1)^\ell\sqrt{{d - \ell} \choose {k - \ell}}.
    \]
\end{corollary}
\begin{proof}
    Follows from Theorem~\ref{thm:fact-main} and the observation that $p(S) = p^*(S) = \frac{1}{{d\choose k}}$ satisfies the conclusion in Lemma~\ref{lem:pstar-optimality}.
\end{proof}

\begin{remark}
    Similarly to Remark~\ref{rem:ellp-algo}, the results in this section can be extended to other factorization norms, corresponding to error measures that interpolate between root mean squared error, and maximum variance. These are the $\gamma_{(p)}$ norms defined in~\cite{NikolovT24}, and studied further in~\cite{LiuUZ24}. We believe our methods extend to these norms, but we do not pursue them further.
\end{remark}

\subsection{Lower Bounds for Product Queries}
\label{sec:lower-bounds-product}


Here we show that Algorithm~\ref{alg:gen-product} describes an optimal factorization mechanism for workloads of product queries.

First we describe how to formulate Algorithm~\ref{alg:gen-product} as a factorization, as we did with Algorithm~\ref{alg:gen}.
Let us fix the functions $\phi_1, \ldots, \phi_d$, a collection $\SS$ of subsets of $[d]$, and a weight function $p:\SS\to \R_{\ge 0}$. 
Let $W$ be, for the rest of this subsection, the workload matrix of $Q^\phi_{\SS}$. Recall the definition of $\tau_a$ for $a \in \uni$ from Algorithm~\ref{alg:gen-product}:
\begin{equation}\label{eq:tau-product-def}
\tau_a :=
\sqrt{\sum_{\substack{S \in \SS\\S \supseteq \supp(a)}}\frac{p(S)\prod_{j \in S}|\widehat{\phi_j}(a_j)|^2}{|\uni_S|^2}} \,.
\end{equation}
Let $\tilde V$ be, as before, the complex matrix with columns indexed by those $a \in \uni$ for which $\tau_a > 0$, and rows indexed by $\uni$, and with entries 
\(
\tilde{V}_{x,a} = \chi_a(x).
\)
The matrix $\tilde{U}$ is such that $\tilde{U}\tilde{V}^* = W$, and can be described explicitly using the formula \eqref{eq:prod-inv-fourier} used in Algorithm~\ref{alg:gen-product}. The rows of $\tilde{U}$ are indexed by  $(S,t)$ where $S \in \SS$, and $t \in \uni_S$, and the columns are indexed by the same set as the columns of $\tilde{V}$. The entries of $\tilde{U}$ are
\[
\tilde{U}_{(S,t),a}=
\begin{cases}
    \frac{1}{|\uni_S|}\left(\prod_{j \in S}\widehat{\phi_j}(a_j)\right) \chi_a(t) & \supp(a)\subseteq S\\
    0 & \supp(a) \not\subseteq S
\end{cases}.
\]
To define the factorization $LR = W$, we take $L$ and $R$ to have the same dimensions as, respectively $\tilde{U}$ and $\tilde{V}^*$, and define their entries by 
\begin{equation}\label{eq:factorization-gen-product}
L_{(S,t),a} := \tilde{U}_{(S,t),a}\sqrt{\frac{\sum_{b \in \uni} \tau_b}{\tau_a}}
\hspace{2em}
R_{a,x} := \sqrt{\frac{\tau_a}{\sum_{b \in \uni} \tau_b}} \overline{\tilde{V}_{x,a}}
= \sqrt{\frac{\tau_a}{\sum_{b \in \uni} \tau_b}} (\tilde{V}^*)_{a,x} \,.
\end{equation}
Once again, we can equivalently write this by defining a diagonal matrix $E$ with rows and columns indexed by $a\in \uni$ such that $\tau_a >0$, and letting $E_{a,a} := \frac{\tau_a}{\sum_{b \in \uni} \tau_b}$; then $L := \tilde{U} E^{-1/2}$, and $R := E^{1/2} \tilde{V}^*$. The equivalence with Algorithm~\ref{alg:gen-product} is easy to check, as we did with Algorithm~\ref{alg:gen}. This construction gives the following upper bounds on factorization norms of the workload matrix $W$, which simply record the error guarantees of Algorithm~\ref{alg:gen-product} in the language of factorization norms. The proof is analogous to the proof of Lemma~\ref{lem:factorization-gen}.

\begin{lemma}\label{lem:factorization-gen-product}
    Let $Q^\phi_{\SS}$ be a workload of product queries over a universe $\uni$, with functions $\phi := (\phi_1, \ldots, \phi_d)$ associated with the attributes, with workload matrix $W$. Let $p:\SS \to \R_{\ge 0}$ be a weight function, with corresponding diagonal matrix $P$ indexed by queries $(S,t)$, $S\in \SS$, $t \in \uni_S$, with entries $P_{(S,t),(S,t)} = \frac{p(S)}{|\uni_S|}$. Then the factorization $LR = W$ in \eqref{eq:factorization-gen-product} satisfies
    \[
    \gamma_F(P^{1/2} W) \le \|P^{1/2} L\|_F \|R\|_{1\to 2} 
    = 
    \sum_{\substack{a \in \uni\\\supp(a) \in \SS_{\downarrow}}}\sqrt{\sum_{\substack{S \in \SS\\S \supseteq \supp(a)}}\frac{p(S)\prod_{j \in S}|\widehat{\phi_j}(a_j)|^2}{|\uni_S|^2}} \,.
    \]
    Moreover, for $p^*:\SS\to \R_{\ge 0}$ chosen as in \eqref{eq:pstar-prod}, the factorization satisfies
    \[
    \gamma_2(W) \le \|L\|_{2\to\infty}\|R\|_{1\to 2} 
    = 
    \sum_{\substack{a \in \uni\\\supp(a) \in \SS_{\downarrow}}}\sqrt{\sum_{\substack{S \in \SS\\S \supseteq \supp(a)}}\frac{p^*(S)\prod_{j \in S}|\widehat{\phi_j}(a_j)|^2}{|\uni_S|^2}} \,,
    \]
    where we recall that $\SS_{\downarrow}$ is the collection of sets $R$ which are subsets of $S$ for some $S \in \SS$.
\end{lemma}

As with Lemma~\ref{lem:factorization-gen}, the upper bound in Lemma~\ref{lem:factorization-gen-product} holds with equality, as shown in the following lemma, which is analogous to Lemma~\ref{lem:gammaF-opt}.

\begin{lemma}\label{lem:gammaF-opt-product}
    For any workload of product queries $Q^\phi_{\SS}$, and for any weights $p:\SS\to\R_{\ge 0}$, with $W$ and $P$ as in Lemma~\ref{lem:factorization-gen-product}, the factorization $LR = W$ for $L$ and $R$ defined in \eqref{eq:factorization-gen-product} satisfies $\gamma_F(P^{1/2}W) = \|P^{1/2}L\|_F \|R\|_{1\to 2}$.

    Similarly, for $p^*:\SS\to \R_{\ge 0}$ chosen as in \eqref{eq:pstar-prod}, the factorization satisfies $\gamma_2(W)=\|L\|_{2\to \infty}\|R\|_{1\to 2}$.
\end{lemma}

\begin{proof}
    The proof is almost identical to the proof of Lemma~\ref{lem:gammaF-opt}. We again use Lemma~\ref{lm:gammaF-dual} with $S := \frac{1}{|\uni|}I$, and it suffices to check that $L^*PL = \frac{c}{|\uni|} RR^*$ for some real constant $c \ge 0$. Once again $RR^* = |\uni|E$ by the same calculation (despite the definition of $\tau_a$ and therefore $E$ being different).

    Now we turn to computing $L^*PL$. For any $a$ and $b$ such that $\tau_a > 0$ and $\tau_b > 0$, we have
    \begin{align*}
      (L^*PL)_{a,b} 
      &= \frac{1}{\sqrt{E_{a,a}E_{b,b}}} (\tilde{U}^*P\tilde{U})_{a,b}\\
      &= \frac{1}{\sqrt{E_{a,a}E_{b,b}}} \sum_{\substack{S\in \SS\\S \supseteq \supp(a)\cup\supp(b)}}\frac{p(S)\prod_{j \in S} (\widehat{\phi_j}(b_j)\overline{\widehat{\phi_j}(a_j)})}{|\uni_S|^3}\sum_{t \in \uni_S} \chi_b(t) \overline{\chi_a(t)}\\
      &= \left(\frac{1}{\sqrt{E_{a,a}E_{b,b}}} \sum_{\substack{S\in \SS\\S \supseteq \supp(a)\cup\supp(b)}}\frac{p(S)\prod_{j \in S} (\widehat{\phi_j}(b_j)\overline{\widehat{\phi_j}(a_j)})}{|\uni_S|^2}\right)\cdot \mathbbm{1}\{a = b\} .
    \end{align*}
    The last equality again follows from \eqref{eq:Fourier-orth} after identifying $\chi_a$ and $\chi_b$ with Fourier characters over $\uni_S$. We then have that $L^*PL$ is a diagonal matrix, and the diagonal entries are
    \begin{align*}
      (L^*PL)_{a,a}  &= \frac{1}{E_{a,a}} \sum_{\substack{S\in \SS\\S \supseteq \supp(a)}}\frac{p(S)\prod_{j \in S} |\widehat{\phi_j}(a_j)|^2}{|\uni_S|^2}\\
      &= \frac{\sum_{b \in \uni} \tau_b}{\tau_a} \sum_{\substack{S\in \SS\\S \supseteq \supp(a)}}\frac{p(S)\prod_{j \in S} |\widehat{\phi_j}(a_j)|^2}{|\uni_S|^2}\\
      &= \left(\sum_{b \in \uni} \tau_b\right) \tau_a 
      = \left(\sum_{b \in \uni} \tau_b\right)^2 E_{a,a}.
    \end{align*}
    Thus, \[L^*PL = \left(\sum_{b \in \uni} \tau_b\right)^2 E = \frac{\left(\sum_{b \in \uni} \tau_b\right)^2}{|\uni|} RR^*,\] as we needed to show.

    Let $P^*$ be the diagonal matrix $P$ but with $p^*$ in place of $p$. By Lemma~\ref{lm:gamma2-dual}, and the proof so far, it is enough to check that for any $S \in \SS$ such that $p^*(S) > 0$, and any $t \in \uni_S$, the $\ell_2$ norm of the row of $L$ indexed by $(S,t)$ equals $\|L\|_{2\to\infty}$, i.e., is maximal. This is verified analogously to Lemma~\ref{lem:gammaF-opt}, using Lemma~\ref{lem:pstar-prod-optimality}.
\end{proof}

The next theorem is a direct consequence of Lemmas~\ref{lem:factorization-gen-product}~and~\ref{lem:gammaF-opt-product}. It gives optimal factorization norms for a workload of weighted product queries, and shows optimality of Algorithm~\ref{alg:gen-product} among all factorization mechanisms.

\begin{theorem}\label{thm:fact-main-product}
Let $Q^\phi_{\SS}$ be a workload of product queries over a universe $\uni$, with functions $\phi:=(\phi_1, \ldots, \phi_d)$ associated with the attributes, and with workload matrix $W$. Let $p:\SS \to \R_{\ge 0}$ be a weight function, and define the corresponding diagonal matrix $P$ indexed by queries $(S,t)$, $S\in \SS$, $t \in \uni_S$, with entries $P_{(S,t),(S,t)} = \frac{p(S)}{|\uni_S|}$. Then the factorization $LR = W$ in \eqref{eq:factorization-gen} satisfies
    \begin{equation}\label{eq:gammaF-opt-product}
    \gamma_F(P^{1/2} W) = \|P^{1/2} L\|_F \|R\|_{1\to 2} 
    = 
    \sum_{\substack{a \in \uni\\\supp(a) \in \SS_{\downarrow}}}\sqrt{\sum_{\substack{S \in \SS\\S \supseteq \supp(a)}}\frac{p(S)\prod_{j \in S}|\widehat{\phi_j}(a_j)|^2}{|\uni_S|^2}}\, ,
    \end{equation}
    where we recall that $\SS_{\downarrow}$ is the collection of sets $R$ which are subsets of $S$ for some $S \in \SS$.
    
    Moreover, for $p^*:\SS\to \R_{\ge 0}$ chosen as in \eqref{eq:pstar-prod}, the factorization satisfies
    \[
    \gamma_2(W) = \|L\|_{2\to\infty}\|R\|_{1\to 2} 
    = 
    \sum_{\substack{a \in \uni\\\supp(a) \in \SS_{\downarrow}}}\sqrt{\sum_{\substack{S \in \SS\\S \supseteq \supp(a)}}\frac{p^*(S)\prod_{j \in S}|\widehat{\phi_j}(a_j)|^2}{|\uni_S|^2}}\, .
    \]    
\end{theorem}

\subsection{Lower Bounds for Extended Marginals}
\label{sec:lower-bound-extended-marginals}

Our goal in this subsection is to show that the upper bound in Theorem~\ref{thm:upper-bound-pref-avg} is \emph{nearly} optimal among all factorization mechanisms. We do so by proving a lower bound on the trace norm $\|P^{1/2}W\|_{tr}$ of the weighted workload matrix $P^{1/2}W$ of a workload of extended marginals. 

In the special case where $d=1$, and the only attribute is numerical, the workload is just the workload of prefix queries. The corresponding workload matrix is lower triangular, and has 1's on and below the main diagonal. The singular value decomposition (SVD) of this matrix is well known, see~\cite{matouvsek2020factorization}. 
Unfortunately, we are not able to give an explicit description of the singular value decomposition of $W$ when we have a mix of numerical and categorical attributes. A technical hurdle is that the all-ones vector is not a singular vector of the lower triangular  matrix mentioned above.

To get around this issue, we instead prove a lower bound on $\|P^{1/2}W\|_{tr}$ using the fact that $\|W\|_{tr} \ge \frac{|\tr(P^{1/2}WY^*)|}{\|Y\|_{op}}$ for any matrix $Y$ of the same dimensions as $W$ (Lemma~\ref{lm:trace-dual} in Appendix~\ref{app:duality}). To construct such a ``test matrix'' $Y$, we take inspiration from the proof of Theorem~\ref{thm:fact-main-product}. Note that, if the SVD of $P^{1/2}W$ is $P^{1/2}W = U\Sigma V^*$, then $UV^*$ is an optimal choice of $Y$, i.e., choosing $Y := UV^*$, we have $\tr(P^{1/2}WY^*) = \|P^{1/2}W\|_{tr}$ and $\|Y\|_{op} = 1$. While we do not know closed forms for the matrices $U$ and $V$ for extended marginals, we ``pretend'' $W$ is defined by product queries, and construct $U$ and $V$ analogously to the construction of $\tilde{U}$ and $\tilde{V}$ in Section~\ref{sec:lower-bounds-product}.

Before presenting our lower bound, let us illustrate this idea with the case of $d=1$ and a single numerical attribute, i.e., a workload of prefix queries on a universe of size $m$. Then the workload matrix $W$ is an $m\times m$ Toeplitz lower triangular matrix, with 1's on and below the main diagonal. If this matrix were circulant, then its left and right singular vectors would be proportional (up to a phase shift) to the Fourier basis vectors $v_a := \frac{1}{\sqrt{m}}(1, \omega_m^{a}, \omega_m^{2a}, \ldots, \omega_m^{(m-1)a})$, $a \in \{0, \ldots, m-1\}$. While $W$ is not circulant, we can still use the Fourier basis vectors to construct a matrix $Y$ such that $\tr(WY^*) \approx \|W\|_{tr}$.  In particular, it is not hard to calculate that 
\[
v_a^* W v_a = 
\begin{cases}
    \frac{m+1}{2} &a = 0\\
    \frac{1}{1-\omega_m^{-a}} & a\neq 0
\end{cases}.
\]
Using this observation, we can define an $m \times m$ matrix $V$ which columns the Fourier basis vectors $v_0, \ldots v_{m-1}$, and an $m\times m$ matrix $U := VE$, where $E$ is a diagonal matrix with entries 
\[
E_{a,a} := 
\begin{cases}
    1 & a = 0\\
    \frac{1-\omega_m^{a}}{|1-\omega_m^{-a}|} &a \neq 0
\end{cases}.
\]
This matrix is defined exactly so that $\overline{E_{a,a}} v_a^* W v_a = |v_a^* W v_a|$.
Then we define the matrix $Y := UV^* = VEV^*$, which clearly satisfies that $\|Y\|_{op} = 1$. We have
\begin{align*}
\|W\|_{tr} \ge |\tr(WY^*)|
&= 
|\tr(WVE^*V^*)|\\
&= |\tr(E^*V^*WV)|
= \left|\sum_{a = 0}^{m-1} \overline{E_{a,a}}\,v_a^* W v_a\right|
= \frac{m+1}{2} + \sum_{a = 1}^{m-1} \frac{1}{|1-\omega_m^{-a}|}.
\end{align*}
The right hand side is $\frac{m+1}{2} + \sum_{a = 1}^{m-1} \frac{1}{\sin\left(\frac{\pi a}{m}\right)}$, which can be shown to be within an additive constant of $\|W\|_{tr}$ (see, again, \cite{matouvsek2020factorization}).

We now extend this approach to arbitrary workloads of extended marginals. Let us first recall the notation for extended marginals from Section~\ref{sec:product-apps}. The universe is $\uni := \uni_1\times \ldots \times \uni_d$, and the $d$ attributes are partitioned into $C$ and $N$, where $C$ are the categorical attributes, and $N$ are the numerical ones. We consider the workload $Q_{\SS}^{\text{pref}}$ of extended marginal queries defined by a family of subsets $\SS$ of $[d] = C \cup N$. The corresponding workload matrix $W$ has rows indexed by pairs $(S,t)$ for $S \in \SS$ and $t \in \uni_S$, and columns indexed by $\uni$. The entries of the workload matrix are 
\(
W_{(S,t), x} = q_{S,t}(x),
\)
with $q_{S,t}(x)$ as defined in  \eqref{eq:marg-prefix}. We also fix a weight functions $p:\SS \to \R_{\ge 0}$, and define the corresponding diagonal matrix $P$  indexed by pairs $(S,t)$ for $S \in \SS$, and $t \in \uni_S$, with entries $P_{(S,t), (S,t)} := \frac{p(S)}{|\uni_S|}$.  We focus on lower bounds for factorization mechanisms for releasing $Q_{\SS}^{\text{pref}}$, but an analogous argument would apply to $Q_{\SS}^{\text{pr-suf}}$.

We  define a test matrix $Y$ such that $\|Y\|_{op} = 1$, and compute $\tr(P^{1/2} WY^*)$, which is a lower bound on $\|P^{1/2}W\|_{tr}$. First we define a few auxiliary matrices. Let $A\subseteq \uni$ be the set of all $a \in \uni$ such that $\supp(a) \in \SS_{p,\downarrow}$, i.e., such that $\supp(a)$ is contained in some $S \in \SS$ with $p(S) > 0$. As in our other lower bound arguments, let $\tilde{V}$ be a complex matrix with columns indexed by $A$ and rows indexed by $\uni$, and with entries $\tilde{V}_{x,a} := \chi_a(x)$. We also define another complex matrix $\tilde{U}$, with rows indexed by pairs $(S,t)$ for $S \in \SS$, and $t \in \uni_S$, and columns indexed by $A$. Let us define the functions $f_j:\uni_j \to \C$ for every $j \in N$ by
\[
f_j(a) := \begin{cases}
    \frac{m_j+1}{2} & a = 0\\[0.5em]
    \frac{1}{1-\omega_{m_j}^{-a}} &a\neq 0
\end{cases}.
\]
We define the entries of $\tilde{U}$ as
\[
\tilde{U}_{(S,t),a}:= 
\begin{cases}
    \frac{1}{|\uni_S|} \left(\prod_{\substack{j \in S \cap N}}f_j(a_j)\right) \chi_{a}(t)&\supp(a) \subseteq S\\
    0 &\text{otherwise}
\end{cases}.
\]
This choice of $\tilde{U}$ is analogous to the definition in Section~\ref{sec:lower-bounds-product}. In particular, if we had a workload of product queries on $\uni$ for which $\phi_j(z) := \mathbbm{1}\{z = 0\}$ for $j \in C$, and for $j \in N$ the function $\phi_j$ were such that $\widehat{\phi_j}(a) = f_j(a)$, then this would be the same $\tilde{U}$. While this workload of product queries is not same as the workload of extended marginals, it turns out that the same definition of $\tilde U$ works for our purposes. 

Next we normalize $\tilde{U}$ and $\tilde V$ so that they have orthonormal columns. We define $V := \frac{1}{\sqrt{|\uni|}}\tilde{V}$. We also define
\[
\kappa_a:= \sqrt{\sum_{\substack{S \in \SS\\t\in \uni_S}}\frac{p(S)}{|\uni_S|}|\tilde{U}_{(S,t),a}|^2}
= \sqrt{\sum_{\substack{S \in \SS\\S \supseteq \supp(a)}}\frac{p(S)\prod_{j \in S\cap N}|f_j(a_j)|^2}{|\uni_S|^2}}.
\]
We then define $U_{(S,t),a}:= \frac{\tilde{U}_{(S,t),a}}{\kappa_a}$. It is now easy to check that $V$ and $P^{1/2}U$ have orthonormal columns. Our test matrix is  $Y := P^{1/2} UV^*$. Our lower bound, derived from computing $\tr(P^{1/2}WY^*)$, is given by the following theorem.


\begin{theorem}\label{thm:ext-marginals-lb}
Let $Q^{\text{pref}}_{\SS}$ be a workload of extended marginal queries over a universe $\uni := \uni_1\times \ldots \times \uni_d$, where $[d]$ is partitioned in categorical attributes $C$, and numerical attributes $N$. Let $p:\SS \to \R_{\ge 0}$ be a weight function, and define the corresponding diagonal matrix $P$ indexed by queries $(S,t)$, $S\in \SS$, $t \in \uni_S$, with entries $P_{(S,t),(S,t)} = \frac{p(S)}{|\uni_S|}$. Then 
    \begin{equation*}
    \gamma_F(P^{1/2} W) \ge
    \sum_{\substack{R\subseteq C\\O \subseteq N}}
    \left(\prod_{j \in R}(m_j-1)\right)\cdot\left(\prod_{j \in O}\zeta(m_j)\right)
    \sqrt{\sum_{\substack{S \in \SS\\S\supseteq R \cup O}}\frac{p(S)\prod_{j \in (S\cap N)\setminus O} \left(1+\frac{1}{m_j}\right)^2}{|\uni_{S\cap C}|^2 \cdot 4^{|S\cap N|}}},
    \end{equation*}
    where $\zeta(m):= \frac{1}{m}\sum_{\ell=1}^{m-1}\frac{1}{\sin\left(\frac{\pi \ell}{m}\right)}.$
\end{theorem}

Before we prove the theorem, let us see how it implies that our factorizations for extended marginals are nearly optimal. Notice that, since $1+\frac{1}{m_j} > 1$ for every $j$, the lower bound in Theorem~\ref{thm:ext-marginals-lb} is at least as strong as 
\[
\gamma_F(P^{1/2}W)
        \ge
        \sum_{\substack{R\subseteq C\\O \subseteq N}}\left(\prod_{j \in R} (m_j-1)\right)\cdot\left(\prod_{j \in O}\zeta(m_j)\right)
        \sqrt{\sum_{\substack{S\in \SS\\S\supseteq R\cup O}}
        \frac{p(S)}{|\uni_{S\cap C}|^2 \cdot4^{|S\cap N|}}}.
\]
Recall also that $\gamma_2(W) \ge \gamma_F(P^{1/2}W)$ for any weight function $p$ and the corresponding diagonal matrix $P$.

On the other  hand, we can use Lemma~\ref{lem:factorization-gen-product} and the calculations in the proof of Theorem~\ref{thm:upper-bound-pref-avg} to get the upper bound 
\[
\gamma_F(P^{1/2}W)
        \le
        \sum_{\substack{R\subseteq C\\O \subseteq N}}\left(\prod_{j \in R} (m_j-1)\right)\cdot\left(\prod_{j \in O}\eta(m_j)\right)
        \sqrt{\sum_{\substack{S\in \SS\\S\supseteq R\cup O}}
        \frac{p(S)}{|\uni_{S\cap C}|^2 \cdot4^{|S\cap N|}}}.
\]
Moreover, there is a weight function $p^*$ for which 
\[
\gamma_2(W) \le
        \sum_{\substack{R\subseteq C\\O \subseteq N}}\left(\prod_{j \in R} (m_j-1)\right)\cdot\left(\prod_{j \in O}\eta(m_j)\right)
        \sqrt{\sum_{\substack{S\in \SS\\S\supseteq R\cup O}}
        \frac{p^*(S)}{|\uni_{S\cap C}|^2 \cdot4^{|S\cap N|}}}.
\]

The upper and lower bounds differ only in that $\eta(m):= \frac1m \sum_{\ell=1}^m \frac{1}{\sin\left(\frac{\pi (2\ell-1)}{2m}\right)}$ is replaced by 
\(
\zeta(m):= \frac{1}{m}\sum_{a = 1}^{m-1}\frac{1}{\sin\left(\frac{\pi a}{m}\right)}
= \frac1m \sum_{\ell=1}^{m-1} \frac{1}{\sin\left(\frac{\pi\cdot 2\ell}{2m}\right)}.
\)
These two functions are always within an additive constant of each other, and are both asymptotically equal to $\frac{1}{\pi}\ln(m) + O(1)$. See Figure~\ref{fig:zeta-vs-eta}. This means that the upper and lower bounds on $\gamma_F(P^{1/2}W)$ and $\gamma_2(W)$ are equal up to a multiplicative factor of $1 - o(1)$ as $\min_{j \in N} m_j \to \infty$ with $\max_{S \in \SS}|S|$ fixed. This regime of large domain size for the numerical attributes is arguably the most natural one.
\begin{figure}[tp]
    \centering
    \begin{minipage}{0.48\textwidth}
        \includegraphics[width=1\linewidth]{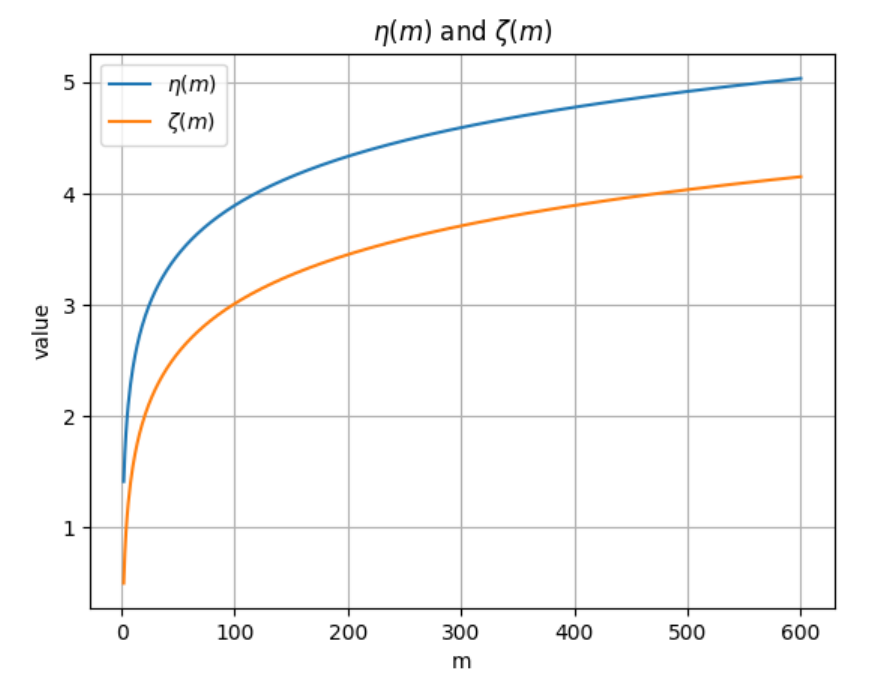}
    \end{minipage}
    \begin{minipage}{0.48\textwidth}
        \includegraphics[width=1\linewidth]{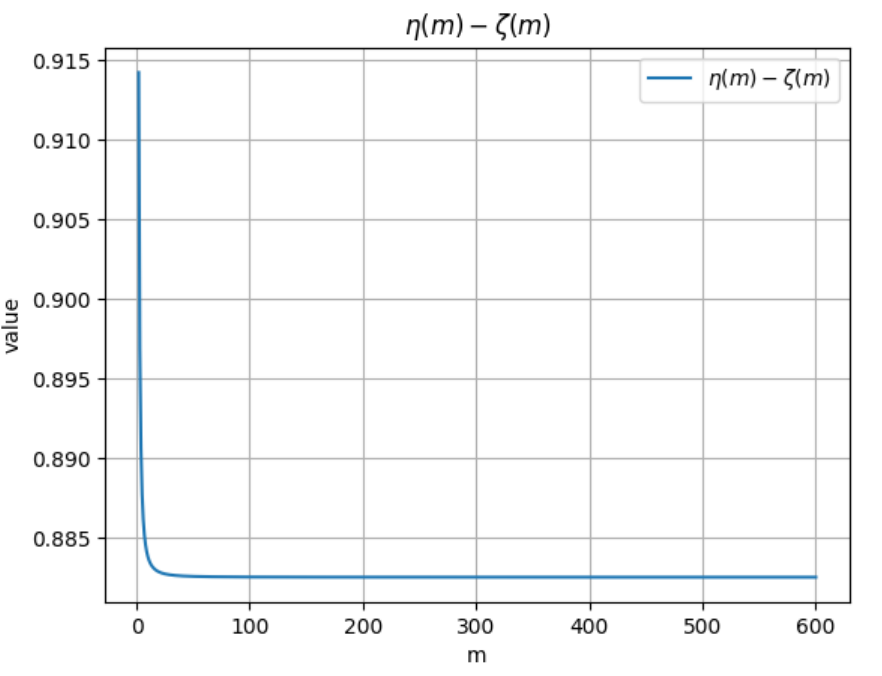}
    \end{minipage}
    \caption{The functions $\eta(m)$ and $\zeta(m)$, as well as their difference $\eta(m)-\zeta(m)$.}
    \label{fig:zeta-vs-eta}
\end{figure}

\begin{proof}[Proof of Theorem~\ref{thm:ext-marginals-lb}]
    We define $Y := P^{1/2} UV^*$ with $U$ and $V$ as defined above. First we claim that $\|P^{1/2}U\|_{op} = \|V\|_{op} = 1$, and, therefore, $\|Y\|_{op} \le 1$. We have
    \[
    (V^* V)_{a,b} = \frac{1}{|\uni|} (\tilde{V}^* \tilde{V})_{a,b}=\frac{1}{|\uni|}\sum_{x \in \uni}\chi_b(x)\overline{\chi_a(x)} = \mathbbm{1}\{a = b\}, 
    \]
    by \eqref{eq:Fourier-orth}. Therefore, $V^* V$ is the identity, and all singular values of $V$ are $1$, which implies $\|V\|_{op} \le 1$. For $P^{1/2} U$, we have
    \begin{align*}
      (U^*PU)_{a,b} 
      &= \frac{1}{\kappa_a\kappa_b} (\tilde{U}^*P\tilde{U})_{a,b}\\
      &= \frac{1}{\kappa_a\kappa_b} \sum_{\substack{S\in \SS\\S \supseteq \supp(a)\cup\supp(b)}}\frac{p(S)\prod_{j \in S\cap N} (f_j(b_j)\overline{f_j(a_j)})}{|\uni_S|^3}\sum_{t \in \uni_S} \chi_b(t) \overline{\chi_a(t)}\\
      &= \left(\frac{1}{\kappa_a\kappa_b} \sum_{\substack{S\in \SS\\S \supseteq \supp(a)\cup\supp(b)}}\frac{p(S)\prod_{j \in S\cap N} (f_j(b_j)\overline{f_j(a_j)})}{|\uni_S|^2}\right)\cdot \mathbbm{1}\{a = b\} .
    \end{align*}
    Therefore, $U^*PU$ is a diagonal matrix, with diagonal entries
    \begin{align*}
      (U^*PU)_{a,a}  &= \frac{1}{\kappa_a^2} \sum_{\substack{S\in \SS\\S \supseteq \supp(a)}}\frac{p(S)\prod_{j \in S\cap N} |f_j(a_j)|^2}{|\uni_S|^2} = 1.
    \end{align*}
    This shows that $U^*PU = I$, and $\|P^{1/2}U\|_{op} = 1$. Therefore, $\|Y\|_{op} \le \|P^{1/2}U\|_{op} \| V^*\|_{op} = 1$, and, by Lemma~\ref{lm:trace-dual}, $|\tr(P^{1/2}W Y^*)| \le \|P^{1/2}W\|_{tr}$. We proceed to compute $\tr(P^{1/2}W Y^*)$. We will need fact show that $\tr(P^{1/2}W Y^*)$ is real and non-negative, so the absolute value will be unnecessary.

    Note that 
    \[
    \tr(P^{1/2}W Y^*) = \tr(P^{1/2} WV U^*P^{1/2})
    = \tr(U^*P WV).
    \]
    Let us first compute $WV$. For each $S \in \SS$, $t \in \uni_S$, and each $a \in A$, we have
    \begin{align*}
        (WV)_{(S,t),a} &= \frac{1}{\sqrt{|\uni|}}\sum_{x \in \uni}q_{S,t}(x) \chi_a(x)\\
        &= \frac{1}{\sqrt{|\uni|}}\sum_{x \in \uni}\left(\prod_{j \in S \cap C} \mathbbm{1}\{x_j = t_j\}\right)\left(\prod_{j \in S \cap N} \mathbbm{1}\{x_j \le t_j\}\right) \chi_a(x)\\
        &= \frac{1}{\sqrt{|\uni|}} \left(\prod_{j \in S \cap C}\omega_{m_j}^{a_j\cdot t_j}\right) \left(\prod_{j \in S \cap N}\left(\sum_{x_j = 0}^{t_j} \omega_{m_j}^{a_j\cdot x_j}\right)\right)\left(\prod_{j \not \in S} \left(\sum_{x_j = 0}^{m_j-1} \omega_{m_j}^{a_j\cdot x_j}\right)\right).
    \end{align*}
    The final product on the right hand side is $0$ unless $a_j = 0$ for all $j \not \in S$. Therefore, $(WV)_{(S,t),a} = 0$ unless $\supp(a) \subseteq S$. Defining
    \[
    f_{j,t}(a) := \sum_{x = 0}^{t}\omega_{m_j}^{a \cdot x} = 
    \begin{cases}
        t+1 & a = 0\\
        \frac{1-\omega_{m_j}^{a(t+1)}}{1-\omega_{m_j}^{a}} & a \neq 0
    \end{cases},
    \]
    we have that, if $\supp(a) \subseteq S$, then 
    \begin{align*}
        (WV)_{(S,t),a} &= \frac{1}{\sqrt{|\uni|}} \left(\prod_{j \in S \cap C}\omega_{m_j}^{a_j\cdot t_j}\right) \left(\prod_{j \in S \cap N}f_{j,t_j}(a_j)\right)\prod_{j \not \in S} m_j\\
        &= \frac{|\uni_{[d]\setminus S}|}{\sqrt{|\uni|}} \left(\prod_{j \in S \cap C}\omega_{m_j}^{a_j\cdot t_j}\right) \left(\prod_{j \in S \cap N}f_{j,t_j}(a_j)\right)
        \\
        &=\frac{\sqrt{|\uni|}}{|\uni_{S}|} \left(\prod_{j \in S \cap C}\omega_{m_j}^{a_j\cdot t_j}\right) \left(\prod_{j \in S \cap N}f_{j,t_j}(a_j)\right).
    \end{align*}
    Next we fix $a \in A$, and proceed to compute $(U^*PWV)_{a,a}$. For any $S \in \SS$ such that $\supp(a) \subseteq S$, we have
    \begin{align}
        \sum_{t \in \uni_S}\overline{U_{(S,t),a}}(WV)_{(S,t),a} &= 
        \frac{\sqrt{|\uni|}}{\kappa_a|\uni_{S}|^2}
        \sum_{t \in \uni_S}\left(\prod_{j \in S \cap C}\omega_{m_j}^{a_j\cdot t_j}\right) \left(\prod_{j \in S \cap N}(f_{j,t_j}(a_j)\overline{f_j(a_j)})\right)\overline{\chi_{a}(t)}\notag\\
        &= 
        \frac{\sqrt{|\uni|}}{\kappa_a|\uni_{S}|^2}
        \sum_{t \in \uni_S} \left(\prod_{j \in S \cap N}(f_{j,t_j}(a_j)\omega_{m_j}^{-a_j \cdot t_j}\overline{f_j(a_j)})\right)\notag\\
        &= 
        \frac{\sqrt{|\uni|}\cdot|\uni_{S\cap C}|}{\kappa_a|\uni_{S}|^2}
        \sum_{t \in \uni_{S\cap N}} \left(\prod_{j \in S \cap N}(f_{j,t_j}(a_j)\omega_{m_j}^{-a_j \cdot t_j}\overline{f_j(a_j)})\right)\notag\\
        &= 
        \frac{\sqrt{|\uni|}|\cdot|\uni_{S\cap C}|}{\kappa_a|\uni_{S}|^2}
        \prod_{j \in S \cap N}\left(\overline{f_j(a_j)}\sum_{t_j = 0}^{m_j-1}f_{j,t_j}(a_j)\omega^{-a_j\cdot t_j}\right).\label{eq:extmarg-partialsum}
    \end{align}
    We now claim that $\sum_{t_j = 0}^{m_j-1}f_{j,t_j}(a_j)\omega^{-a_j\cdot t_j} = m_j f_j(a_j)$. If $a_j = 0$, then $\omega_{m_j}^{a_j\cdot t_j} = 1$, and we have 
    \[
    \sum_{t_j = 0}^{m_j-1}f_{j,t_j}(0)
    = \sum_{\ell = 1}^{m_j}\ell = \frac{m_j(m_j+1)}{2} = m_j f_j(0). 
    \]
    If $a_j \neq 0$, then 
    \[
    \sum_{t_j = 0}^{m_j-1}f_{j,t_j}(a_j)\omega^{-a_j\cdot t_j}
    = \sum_{t_j = 0}^{m_j-1} \frac{\omega_{m_j}^{-a_j \cdot t_j} -\omega_{m_j}^{a_j}}{1-\omega_{m_j}^{a_j}}
    = -\frac{m_j\omega_{m_j}^{a_j}}{1-\omega_{m_j}^{a_j}}
    = \frac{m_j}{1-\omega_{m_j}^{-a_j}} = m_j f_j(a_j),
    \]
    where we used that $\sum_{t_j= 0}^{m_j-1}\omega_{m_j}^{-a_j\cdot t_j} = 0$ whenever $a_j \neq 0$. Plugging back into \eqref{eq:extmarg-partialsum}, we have
    \begin{align*}
       \sum_{t \in \uni_S}\overline{U_{(S,t),a}}(WV)_{(S,t),a} &=
       \frac{\sqrt{|\uni}|\cdot|\uni_{S\cap C}|\cdot|\uni_{S\cap N}|}{\kappa_a|\uni_{S}|^2}
        \prod_{j \in S \cap N}|f_j(a_j)|^2
        = \frac{\sqrt{|\uni}|}{\kappa_a|\uni_{S}|}
        \prod_{j \in S \cap N}|f_j(a_j)|^2.
    \end{align*}
    Adding up these equalities over $S \in \SS$, and using that $U_{(S,t),a} = (WV)_{(S,t),a} = 0$ whenever $\supp(a) \not \subseteq S$, we have
    \begin{align*}
        (U^*PWV)_{a,a} &=
        \sum_{\substack{S\in \SS\\S\supseteq \supp(a)}}\left(\frac{p(S)}{|\uni_S|}\sum_{t \in \uni_S}\overline{U_{(S,t),a}}(WV)_{(S,t),a}\right)\\
        &= \frac{\sqrt{|\uni|}}{\kappa_a}\sum_{\substack{S\in \SS\\S\supseteq \supp(a)}}
        \frac{p(S)}{|\uni_{S}|^2} \prod_{j \in S \cap N}|f_j(a_j)|^2\\
        &= \sqrt{|\uni|}\sqrt{\sum_{\substack{S\in \SS\\S\supseteq \supp(a)}}
        \frac{p(S)}{|\uni_{S}|^2} \prod_{j \in S \cap N}|f_j(a_j)|^2}.
    \end{align*}
    Summing over $a \in A$, we finally get
    \begin{align}
        \gamma_F(P^{1/2}W)
        \ge
        \frac{1}{\sqrt{|\uni|}} \|P^{1/2}W\|_{tr}
        &\ge \frac{1}{\sqrt{|\uni|}} |\tr(P^{1/2}WY^*)|\notag\\
        &= \frac{1}{\sqrt{|\uni|}} |\tr(U^*PWV)|\notag\\
        &= \sum_{a \in A}\sqrt{\sum_{\substack{S\in \SS\\S\supseteq \supp(a)}}
        \frac{p(S)}{|\uni_{S}|^2} \prod_{j \in S \cap N}|f_j(a_j)|^2},\label{eq:extmarg-lb1}
    \end{align}
    where the first two inequalities are by Lemma~\ref{lm:gammaF-dual} and Lemma~\ref{lm:trace-dual}, respectively. 

    The final step in the proof is to manipulate \eqref{eq:extmarg-lb1} to put it in more explicit form. To this end, we first compute $|f_j(a_j)|$
    \[
    |f_j(a_j)| = 
    \begin{cases}
        \frac{m_j+1}{2} &a_j = 0\\
        \frac{1}{2\sin\left(\frac{\pi a_j}{m_j}\right)} & a_j \neq 0,
    \end{cases}
    \]
    where we use that $|1-\omega_{m_j}^{-a_j}|^2 = 2\left(1-\cos\left(\frac{2\pi a_j}{m_j}\right)\right) = 4\sin\left(\frac{\pi a_j}{m_j}\right)^2$. Therefore, for any $a$ with $R:= \supp(a) \cap C$ and $O:= \supp(a) \cap N$, we have
    \begin{align}
        \sum_{\substack{S\in \SS\\S\supseteq \supp(a)}}
        &\frac{p(S)}{|\uni_{S}|^2} \prod_{j \in S \cap N}|f_j(a_j)|^2
        = 
        \sum_{\substack{S\in \SS\\S\supseteq R\cup O}}
        \frac{p(S)}{|\uni_{S}|^2} \left(\prod_{j \in O}\frac{1}{4\sin\left(\frac{\pi a_j}{m_j}\right)^2}\right)\left(\prod_{j \in (S \cap N)\setminus O} \frac{(m_j+1)^2}{4}\right)\notag\\
       &=
       \left(\prod_{j \in O}\frac{1}{m_j^2\sin\left(\frac{\pi a_j}{m_j}\right)^2}\right)\sum_{\substack{S\in \SS\\S\supseteq R\cup O}}
        \frac{p(S)}{|\uni_{S\cap C}|^2 \cdot4^{|S\cap N|}}\prod_{j \in (S \cap N)\setminus O} \left(1+\frac{1}{m_j}\right)^2.\label{eq:extmarg-lb2-sqrt}
    \end{align}
    Moreover, for any $R \subseteq C$ and $O\subseteq N$, summing over $a$ such that $\supp(a) \cap C = R$ and $\supp(a) \cap N = O$, we get
    \begin{equation}\label{eq:extmarg-lb2-suma}
        \sum_{\substack{a \in A\\\supp(a) \cap C = R\\\supp(a)\cap N = O}}\prod_{j \in O}\frac{1}{m_j\sin\left(\frac{\pi a_j}{m_j}\right)}
        = 
        \left(\prod_{j \in R} (m_j-1)\right)\prod_{j \in O}\left(\frac{1}{m_j}\sum_{a_j = 1}^{m_j-1}\frac{1}{\sin\left(\frac{\pi a_j}{m_j}\right)}\right).
    \end{equation}
    
    Combining \eqref{eq:extmarg-lb1}, \eqref{eq:extmarg-lb2-sqrt}, and \eqref{eq:extmarg-lb2-suma}, and recalling $\zeta(m):= \frac{1}{m}\sum_{a = 1}^{m-1}\frac{1}{\sin\left(\frac{\pi a}{m}\right)}$ we get
    \begin{multline*}
        \gamma_F(P^{1/2}W)
        \ge\\
        \sum_{\substack{R\subseteq C\\O \subseteq N}}\left(\prod_{j \in R} (m_j-1)\right)\cdot\left(\prod_{j \in O}\zeta(m_j)\right)
        \sqrt{\sum_{\substack{S\in \SS\\S\supseteq R\cup O}}
        \frac{p(S)}{|\uni_{S\cap C}|^2 \cdot4^{|S\cap N|}} \prod_{j \in (S \cap N)\setminus O} \left(1+\frac{1}{m_j}\right)^2},
    \end{multline*}
    as we needed to show.
\end{proof}

\fi

\ifarxiv 
\section*{Acknowledgements}
We thank Ryan McKenna for pointing out relevant related work. 
The work of Christian Janos Lebeda is supported by grant ANR-20-CE23-0015 (Project PRIDE) and the ANR 22-PECY-0002 IPOP (Interdisciplinary Project on Privacy) project of the Cybersecurity PEPR. Aleksandar Nikolov and Haohua Tang were supported by an NSERC Discovery Grant.
\fi





\ifarxiv
\sasho{We can generalize our upper bounds for $\gamma_2$ (maximum error) to $\gamma_{(p)}$ $\ell_p$ error for any $p \ge 2$. The proofs are almost the same. I don't find this very important, but some reviewer may complain that the ResidualPlanner paper deals with more general objectives. We can also just make a remark that we focus on the most important objectives but we can handle others too.} \christian{How would the generalization to $\gamma_{(p)}$ work? Do you think we can get closed-form expression like in  $\gamma_{F}$ or do we need to optimize similar to the maximum error? If we add anything it makes sense to keep it short in my opinion.}
\fi



\ifarxiv
\christian{We could add a note about randomness complexity of DP. The recent work on this topic is motivated by the fact that the census required immense resources for generating high quality random bits. We use fewer samples for marginals. For k-way marginals the improvement is the most significant as we only use roughly a factor $(1 - 1/m)^k$ as many samples. (although we have to remember the factor 2 for non-binary attributes from using complex numbers) Conjecture: For binary attributes our approach is optimal also in terms of randomness complexity among all optimal factorizations.}
\fi

\bibliographystyle{alpha}
\bibliography{references}

\appendix

\ifpods

\fi

\section{Missing Proofs from Section~\ref{sec:fact-duality}}
\label{app:duality}

Towards proving Lemmas~\ref{lm:gammaF-dual}~and~\ref{lm:gamma2-dual}, let us first recall a well-known characterization of the trace norm. 

\begin{lemma}\label{lm:trace-dual}
    For any two matrices $X$ and $Y$ of the same dimensions, we have $|\tr(XY^*)| \le \|X\|_{tr}\|Y\|_{op}$. Moreover, for any matrix $X$, there is a matrix $Y$ of the same dimensions with $\|Y\|_{op} = 1$, and such that
    \[
    \tr(XY^*) = \|X\|_{tr}\|Y\|_{op} = \|X\|_{tr}.
    \]
\end{lemma}
\begin{proof}
    Let $X = U\Sigma V^*$ be the SVD of $X$, where $\Sigma$ is an $r\times r$ diagonal matrix with non-negative entries, and $U$ and $V$ are matrices with $r$ orthonormal columns each. Then 
    \[
    |\tr(XY^*)| = |\tr(\Sigma (U^* Y V)^*)| \le \sum_{i=1}^r \sigma_i |u_i^* Y v_i|,
    \]
    where $u_i$ is the $i$-th column of $U$ (i.e., the $i$-th left singular vector of $Y$), and $v_i$ is the $i$-th column of $V$ (i.e., the $i$-th right singular vector of $X$). We have that 
    \[
    |u_i^* Y v_i| \le \|u_i\|_2 \|Y\|_{op} \|v_i\|_2 = \|Y\|_{op},
    \]
    so $\tr(XY^*) \le  \sum_{i=1}^r\sigma_i\|Y\|_{op} = \|X\|_{tr}\|Y\|_{op}$.

    For the claim after ``moreover'', take $U$ and $V$ to be the matrices of left and right singular vectors of $X$, as above, and define $Y = UV^*$. Then 
    \[
    \tr(XY^*) = \tr(U^*XV) = \tr(\Sigma) = \|X\|_{tr},
    \]
    and $\|UV^*\|_{op} = 1$.
\end{proof}

To prove Lemma~\ref{lm:gammaF-dual}, we also need a matrix version of the Cauchy-Schwarz inequality.
\begin{lemma}\label{lm:mat-CS}
    For any two matrices $X$ and $Y$ for which the product $XY$ is well defined,
    \[
    \|XY\|_{tr} \le \|X\|_F \|Y\|_F,
    \]
    and equality holds if and only if $X^* X = c\,YY^*$ for some real number $c\ge 0$.
\end{lemma}
\begin{proof}
    Let $Z$ be a matrix with $\|Z\|_{op} = 1$, and of the same dimensions as the product $XY$, so that $\tr(XYZ^*) = \|XY\|_{tr}$. We have
    \begin{equation}\label{eq:CS}
    \|XY\|_{tr} = |\tr(X(ZY^*)^*)| \le \|X\|_F \|Z Y^*\|_F,
    \end{equation}
    where we just used the standard Cauchy-Schwarz inequality, treating $X$ and $Z Y^*$ as vectors, and $\tr((X)(ZY^*)^*)$  as the standard inner product between them. Now observe that
    \begin{equation}\label{eq:Y-Fnorm}
    \|ZY^*\|_F^2 = \tr(Y(Z^*Z)Y^*) \le \tr(Y Y^*) = \|Y\|_F^2.
    \end{equation}
    This is because, by assumption, $\|Z^*Z\|_{op} = \|Z\|_{op}^2\le 1$, so $y_i (Z^* Z) y_i^* \le y_i y_i^*$ for each row $y_i$ of $Y$ (seen as a row vector). 
    Substituting gives us the required inequality.

    If the inequality holds with equality, then \eqref{eq:CS} and \eqref{eq:Y-Fnorm} must also hold with equality. For the Cauchy-Schwarz inequality \eqref{eq:CS} to hold with equality, it must be the case that there is a number $c \in \C$ for which $X = c(Z Y^*)$. For \eqref{eq:Y-Fnorm} to hold with equality, $Z^*Z$ must act as the identity on the row space of $Y$, i.e., we must have $Z^* Z Y^* = Y^*$. Then we have
    \[
    X^* X = |c|^2 YZ^* ZY^* = |c|^2 YY^*,
    \]
    as we needed to show (after renaming $|c|^2$ to $c$).
    
    Finally, assume that $X^* X = c\,YY^*$ for some real number $c \ge 0$. If $c = 0$, then clearly $X = 0$ as well, so $\|XY\|_{tr} = 0 = \|X\|_F \|Y\|_F$. Let us assume then that $c> 0$. Then
    \[
    (XY)(XY)^* = XYY^*X^* = \frac{1}{c}\,(XX^*)^2.
    \]
    Therefore, the singular values of $XY$ are equal to the eigenvalues of the positive semidefinite matrix $\frac{1}{\sqrt{c}}XX^*$, and we have
    \begin{align*}
     \|XY\|_{tr} = \frac{1}{\sqrt{c}}\,\tr(XX^*) &= \sqrt{\tr(XX^*)}\sqrt{\tr((1/c)\,XX^*)}\\
    &= \sqrt{\tr(XX^*)}\sqrt{\tr(Y^*Y)}
    = \|X\|_F\|Y\|_F.    
    \end{align*}
    Thus the inequality holds with equality.
\end{proof}

We can now prove Lemma~\ref{lm:gammaF-dual}.

\begin{proof}[Proof of Lemma~\ref{lm:gammaF-dual}]
Let $W = LR$ be a factorization of $W$.
Applying Lemma~\ref{lm:mat-CS} to $X := L$ and $Y := RS^{1/2}$, we have
\begin{equation}\label{eq:gammaF-dual-CS}
\|WS^{1/2}\|_{tr}= \|L(RS^{1/2})\|_{tr} \le \|L\|_{F} \|RS^{1/2}\|_F.
\end{equation}
Let $r_j$ be the $j$-th column of $R$. We have
\begin{equation}\label{eq:D-avg}
\|RS^{1/2}\|_{F}^2=\sum_{j = 1}^N S_{j,j} \|r_j\|_2^2 \le \left(\sum_{j=1}^N S_{j,j}\right) \max_{j=1}^N \|r_j\|_2^2
= \|R\|_{1\to 2}^2.
\end{equation}
The inequalities \eqref{eq:gammaF-dual-CS} and \eqref{eq:D-avg} imply that $\|WS^{1/2}\|_{tr} \le \|L\|_F \|R\|_{1\to 2}$ for any $L$ and $R$ such that $W = LR$. Minimizing the right hand side over all such choices of $L$ and $R$ then gives us $\|WS^{1/2}\|_{tr} \le \gamma_F(W)$.

Suppose now that $LR = W$ is some factorization of $W$. Clearly $\|WS^{1/2}\|_{tr} = \gamma_F(W) = \|L\|_F\|R\|_{1\to 2}$ if and only if both \eqref{eq:gammaF-dual-CS} and \eqref{eq:D-avg} hold with equality. By Lemma~\ref{lm:mat-CS}, \eqref{eq:gammaF-dual-CS} is tight if and only if $L^*L = c RSR^*$ for some $c \ge 0$. The inequality \eqref{eq:D-avg} is clearly tight if and only if $\|r_j\|_{2} = \|R\|_{1\to 2}$ whenever $S_{j,j} \neq 0.$
\end{proof}

The proof of Lemma~\ref{lm:gamma2-dual} is analogous to the proof of Lemma~\ref{lm:gammaF-dual}, but we use Lemma~\ref{lm:mat-CS} with $X := P^{1/2}L$ and $Y := RS^{1/2}$ instead. 

\section{Discrepancy with the SVD Lower Bound of McKenna et al.}
\label{app:hdmm}

Theorem 12 in~\cite{McKenna_Miklau_Hay_Machanavajjhala_2023} claims the following formula holds for the singular value lower bound of weighted marginal queries. Here we use the notation from Section~\ref{sec:our-algs-are-factorizations}, rather the notation of McKenna et al.
\begin{equation}\label{eq:HDMM-svdbd1}
    \|P^{1/2} W\|_{tr}
    = 
    \frac{1}{\sqrt{|\uni|}}\sum_{R \subseteq [d]} |\uni_R| \sqrt{\sum_{T \supseteq R}\frac{p(T)|\uni_{[d]\setminus T}|}{|\uni_T|}}\,,
\end{equation}
where we define $|\uni_{\emptyset}| = 1$, and set $p(T) = 0$ if $T \not \in \SS$.
To see that this corresponds to Theorem 12 in~\cite{McKenna_Miklau_Hay_Machanavajjhala_2023}, note that the bit vectors $a,b \in \{0,1\}^d$ in their notation correspond to sets $R$ and $T$ whose indicator vectors are, respectively, $a$ and $b$; $c(\neg a)$ in their notation equals $|\uni_{R}|$, while $c(b) = |\uni_{[d]\setminus T}|$, and, with our normalization, $w(b) = \frac{p(R)}{|\uni_R|}$. 

Let us rewrite \eqref{eq:HDMM-svdbd1} to bring it closer to \eqref{eq:gammaF-opt}. By bringing the $1/\sqrt{|\uni|}$ term inside the square root, we get
\begin{align}
    \|P^{1/2} W\|_{tr}
    = 
    \sum_{R \subseteq [d]} |\uni_R| \sqrt{\sum_{T \supseteq R}\frac{p(T)|\uni_{[d]\setminus T}|}{|\uni_T|\cdot|\uni|}}
    &=
    \sum_{R \subseteq [d]} |\uni_R| \sqrt{\sum_{T \supseteq R}\frac{p(T)}{|\uni_T|^2}}\notag\\
    &=
    \sum_{R \subseteq [d]} \left(\prod_{j \in R}m_j\right) \sqrt{\sum_{T \supseteq R}\frac{p(T)}{|\uni_T|^2}}\,.\label{eq:HDMM-svdbd2}
\end{align}
Now we can see that the coefficient in front of each square root on the right hand side in \eqref{eq:HDMM-svdbd2} is $\prod_{j \in R}m_j$ rather than $\prod_{j \in R}(m_j-1)$ as in \eqref{eq:gammaF-opt}. In particular, \eqref{eq:HDMM-svdbd2} is always at least as large as \eqref{eq:gammaF-opt}, and would seem to contradict Lemma~\ref{lem:factorization-gen}. 

The reason for the discrepancy is an error in the proof of Theorem 12 in~\cite{McKenna_Miklau_Hay_Machanavajjhala_2023}. In their Theorem 9, McKenna et al.~define, for each set $R \subseteq [d]$, a matrix $V(R)$ (or $V(a)$ in their notation) with $|\uni_R|$ rows, where each row is an eigenvector of $W^TPW$ with eigenvalue $\kappa(R) := \sum_{T \supseteq R}\frac{p(T)|\uni_{[d]\setminus T}|}{|\uni_T|}$. The matrix is given by the formula
\(
V(R) := \prod_{j=1}^d V_j(R),
\)
where $V_j(R)$ equals the $1\times m_j$ all-ones matrix if $j \not \in R$, and $J-m_j I$ if $j \in R$, for the $m_j \times m_j$ all-ones matrix $J$, and the $m_j\times m_j$ identity matrix $I$. 
From this, they infer that, for any $R\subseteq [d]$, $W^TPW$ has $|\uni_R|$ singular values equal to $\sqrt{\kappa(R)}$, and add these singular values with these multiplicities to get \eqref{eq:HDMM-svdbd1}. This argument, however, overcounts the singular values. In particular, the rows of $V(R)$ are not linearly independent unless $R = \emptyset$, so the eigenspace spanned by $V(R)$ is not necessarily of dimension $|\uni_R|$. Indeed, notice that $J-m_j I$ has rank $m_j-1$, and, therefore, $V(R)$ has rank $\prod_{j\in R}(m_j -1)$. One needs to also verify that there are no non-trivial linear dependencies between the rows in different $V(R)$ matrices, but this turns out to be a non-issue since $V(R) V(R')^T = 0$ whenever $R \neq R'$. Correcting for the dimension of the rowspan of $V(R)$ in the McKenna et al.~proof recovers our formula \eqref{eq:gammaF-opt}.

Let us consider a small example to illustrate this. Suppose that $d=2$, $m_1 = m_2 = 2$, $\SS = \{\{1\},\{2\}\}$, and $p(\{1\}) = p(\{2\}) = \frac{1}{2}$. Then 
\[
W^T P W
= 
\frac14 
\begin{pmatrix}
    2 & 1 & 1 & 0\\
    1 & 2 & 0 & 1\\
    1 & 0 & 2 & 1\\
    0 & 1 & 1 & 2
\end{pmatrix}.
\]
This matrix has one eigenvalue equal to $1$, two eigenvalues equal to $\frac{1}{2}$, and one eigenvalue equal to $0$. The eigenvectors are just the columns of $\tilde V$:
\[
\tilde{V} = 
\begin{pmatrix}
    1 & 1 & 1 & 1\\
    1 & 1 & -1 & -1\\
    1 & -1 & 1 & -1\\
    1 & -1 & -1 & 1
\end{pmatrix}.
\]
Therefore, the SVD lower bound is 
\(
\frac{1}{2}\|P^{1/2}W\|_{tr} = \frac12 + \frac{1}{2\sqrt{2}} + \frac{1}{2\sqrt{2}} + 0 = \frac{1 + \sqrt{2}}{2}.
\)
At the same time, \eqref{eq:HDMM-svdbd1} and \eqref{eq:HDMM-svdbd2} give $\frac{1}{2}+ \sqrt{2}$ which is larger than the sensitivity of $W$, i.e.~$\|W\|_{1\to 2} = \sqrt{2}$. The eigenmatrices in Theorem 9 of~\cite{McKenna_Miklau_Hay_Machanavajjhala_2023} are
\begin{align*}
    V(\emptyset) &= 
    \begin{pmatrix}
        1 & 1 & 1 & 1
    \end{pmatrix};\\
    V(\{1\}) &=
    \begin{pmatrix}
        -1 & -1 & 1 & 1\\
        1 & 1 & -1 & -1
    \end{pmatrix};\\
    V(\{2\}) &=
    \begin{pmatrix}
        -1 & 1 & -1 & 1\\
        1 & -1 & 1 & -1
    \end{pmatrix}.
\end{align*}
Notice that the two rows of $V(\{1\})$ are colinear, as are the two rows of $V(\{2\})$, which is the cause of the overcounting.

\end{document}